\documentclass[sigconf]{acmart} %% CCS: DO NOT REMOVE
\AtBeginDocument{%
  }

\copyrightyear{2026}
\acmYear{2026}
\setcctype{by}
\acmConference[CCS '26]{Proceedings of the 2026 ACM SIGSAC Conference on Computer and Communications Security}{November 15--19, 2026}{The Hague, Netherlands}
\acmBooktitle{Proceedings of the 2026 ACM SIGSAC Conference on Computer and Communications Security (CCS '26), November 15--19, 2026, The Hague, Netherlands}
\acmDOI{10.1145/3830454.3846703}
\acmISBN{979-8-4007-2871-6/2026/11}

\usepackage{hyperref}
\usepackage{url}
\usepackage{enumitem}
\usepackage{amsmath}
\usepackage{amsthm}
\usepackage{array} % >{\hspace{...}} in tables
\usepackage{graphicx} % figure
\usepackage{subcaption}
\usepackage{multirow}
\usepackage{soul} % \hl
\usepackage[normalem]{ulem} % \sout
\usepackage{tikz}
\usepackage{booktabs}

\newcommand{\method}{PAPER}
\newtheorem{lemma}{Lemma}

\makeatletter
\renewenvironment{proof}[1][\proofname]{%
  \par\pushQED{\qed}%
  \normalfont\topsep6\p@\@plus6\p@\relax
  \trivlist
  \item[\hskip\labelsep\bfseries\itshape #1\@addpunct{.}]%
}{%
  \popQED\endtrivlist\@endpefalse
}
\makeatother

\newcommand{\ceil}[1]{\left\lceil #1 \right\rceil}

\newcommand{\round}[1]{\left\lfloor #1 \right\rceil}

\newcolumntype{C}[1]{>{\hspace{-3pt}\centering\arraybackslash}p{#1}<{\hspace{-3pt}}}

\ifcase 1%

\or

\or

\fi

\newcommand*\circled[1]{\tikz[baseline=(char.base)]{
            \node[shape=circle,fill=black,text=white,inner sep=1.5pt] (char) {#1};}}

\newcommand{\nodebox}[3][black]{
\tikz[baseline=(n.base)]\node[
  draw,
  fill=#2,
  text=#1,
  rounded corners=2pt,
  inner sep=1.5pt,
] (n) {\textsf{#3}};
}
\definecolor{BNcolor}{RGB}{255,242,204}
\definecolor{Convcolor}{RGB}{248,206,204}
\definecolor{Polycolor}{RGB}{218,232,252}
\definecolor{Bivarcolor}{RGB}{225,213,231}
\definecolor{Addcolor}{RGB}{213,232,212}

\newcommand{\BN}{\nodebox[black]{BNcolor}{B}}
\newcommand{\Conv}{\nodebox[black]{Convcolor}{C}}
\newcommand{\Poly}{\nodebox[black]{Polycolor}{P}}
\newcommand{\Bivar}{\nodebox[black]{Bivarcolor}{S}}

\definecolor{DonorColor}{RGB}{248,206,204}
\definecolor{ReceiverColor}{RGB}{255,242,204}
\newcommand{\Don}{\nodebox[black]{DonorColor}{$D$}}
\newcommand{\Rec}{\nodebox[black]{ReceiverColor}{$R$}}

\newcommand{\Add}{
\tikz[baseline=(n.base)]\node[
  draw,
  circle,
  inner sep=1.2pt,
  fill=Addcolor
] (n) {\textsf{+}};
}

\usepackage{amsmath,amsfonts,bm}

\def\eqref#1{equation~\ref{#1}}
\def\ceil#1{\lceil #1 \rceil}

\def\1{\bm{1}}

\DeclareMathAlphabet{\mathsfit}{\encodingdefault}{\sfdefault}{m}{sl}
\SetMathAlphabet{\mathsfit}{bold}{\encodingdefault}{\sfdefault}{bx}{n}

\usepackage{etoolbox}

\makeatletter

\patchcmd{\maketitle}
{\ifx\@acmDOI\@empty\else\@formatdoi{\@acmDOI}\fi}
{}
{}
{}

\patchcmd{\maketitle}
{{\itshape \acmConference@shortname, \acmConference@venue}\par}
{}
{}
{}

\makeatother

\begin{document}

\title[Privacy-Preserving CNNs using Low-Degree Polynomial Approx. and Structural Optimizations on Leveled FHE]{\method: Privacy-Preserving Convolutional Neural Networks using Low-Degree Polynomial Approximations and Structural Optimizations on Leveled FHE}

% \author{Anonymous Author(s)}

\author{Eduardo Chielle}
\orcid{0000-0002-1938-912X}
\affiliation{%
  \institution{Center for Cyber Security}
  \city{NYU Abu Dhabi}
  \country{UAE}
}
\email{eduardo.chielle@nyu.edu}
% \authornote{Corresponding author.}

\author{Manaar Alam}
\orcid{0000-0002-3338-2944}
\affiliation{%
  \institution{Center for Cyber Security}
  \city{NYU Abu Dhabi}
  \country{UAE}}
\email{alam.manaar@nyu.edu}

\author{Jinting Liu}
\orcid{0009-0008-2957-4354}
\affiliation{%
  \institution{Computer Science and Engineering}
  \city{NYU Shanghai}
  \country{China}}
\email{jl13148@nyu.edu}

\author{Jovan Kascelan}
\orcid{0009-0008-2839-7077}
\affiliation{%
  \institution{Center for Cyber Security}
  \city{NYU Abu Dhabi}
  \country{UAE}}
\email{jk7480@nyu.edu}

\author{Michail Maniatakos}
\orcid{0000-0001-6899-0651}
\affiliation{%
  \institution{Center for Cyber Security}
  \city{NYU Abu Dhabi}
  \country{UAE}}
\email{michail.maniatakos@nyu.edu}

\renewcommand{\shortauthors}{Chielle et al.}

\begin{abstract}
Recent work using Fully Homomorphic Encryption (FHE) has made \textit{non-interactive privacy-preserving inference} of deep Convolutional Neural Networks (CNNs) possible.
However, the performance of these methods remains limited by their heavy reliance on \textit{bootstrapping}, a costly FHE operation applied across multiple layers, severely slowing inference.
Moreover, they depend on \textit{high-degree polynomial approximations} of non-linear activations, which increase multiplicative depth and reduce accuracy by 2–5\% compared to plaintext ReLU models.
In this work, we close the accuracy gap between FHE-based non-interactive CNNs and their plaintext counterparts, while also achieving faster inference than existing methods.
We propose a \textit{quadratic polynomial approximation} of ReLU, which achieves the theoretical minimum multiplicative depth for non-linear activations, together with a penalty-based training strategy.
We further introduce \textit{structural optimizations} that reduce the required FHE levels in CNNs by a factor of five compared to prior work, allowing us to \textit{run deep CNN models under leveled FHE without bootstrapping}.
To further accelerate inference and recover accuracy typically lost with polynomial approximations, we introduce parameter clustering along with a joint strategy of data layout and ensemble techniques.
Experiments with VGG and ResNet models on CIFAR and Tiny-ImageNet datasets show that our approach achieves up to $4\times$ faster private inference than prior work, with accuracy comparable to plaintext ReLU models.\vspace{-2mm}
\end{abstract}




% \received{20 February 2007} 
% \received[revised]{12 March 2009}
% \received[accepted]{5 June 2009}

\maketitle

\section{Introduction}
\textit{Machine Learning as a Service} (MLaaS) is increasingly adopted across industries because it provides access to powerful models without requiring in-house development or maintenance~\cite{mlaas}.
However, it raises serious privacy concerns since models are often trained on proprietary or sensitive data~\cite{ml4medical,ml4finance}, and inference requires clients to share private information (e.g., medical or financial records) with external service providers. 
\textit{Privacy-Preserving Machine Learning} (PPML) addresses these risks using predominantly cryptographic methods, mainly secure \textit{Multi-Party Computation} (MPC)~\cite{ppmlmpcsok} and \textit{Fully Homomorphic Encryption} (FHE)~\cite{ppmlfhesok}, which protect both client data and provider models.
Broadly, PPML methods fall into \textit{interactive} and \textit{non-interactive} categories~\cite{mpcnn}. 
Interactive PPML, based on MPC and homomorphic encryption \cite{mpcnn}, requires the client and the service provider to jointly perform inference~\cite{minionn,gazelle,delphi,safenet,cheetah,pillar}. 
It keeps computational costs relatively low but requires multiple communication rounds, which leads to high communication and bandwidth requirements.
Non-interactive PPML, typically using FHE, adopts a fire-and-forget paradigm~\cite{cryptonets,lola,hemet,mpcnn,sarkar,autofhe}, where the client encrypts inputs, the provider computes on ciphertexts, and the client decrypts the output.
This reduces communication to a single round at the expense of heavy server-side computation. 
Despite the costs, non-interactive PPML is well-suited to third-party services because it eliminates client involvement during inference, supports clients with low-resource devices, and works in low-bandwidth settings. 
Motivated by these advantages, this work focuses on non-interactive PPML.
% \newpage

A key challenge in non-interactive PPML is implementing activation functions with cryptographic primitives in an efficient and accurate manner.
Commonly used activations are non-polynomial, whereas FHE natively supports only polynomial operations such as addition and multiplication.
Existing work addresses this by approximating activations with single polynomials~\cite{pillar} or piecewise polynomials~\cite{minionn}.
These approximations can be introduced before or after training.
In \textit{post-training approximation}, models are trained with standard activations (e.g., ReLU) and later replaced with high-precision polynomials. 
For example, MPCNN~\cite{mpcnn} uses minimax polynomials of orders $\{15, 15, 27\}$.
In \textit{pre-training approximation}, ReLU is replaced with a polynomial before training, which allows lower-degree approximations that reduce multiplicative depth and speed up inference.
For example, shallow CNNs like CryptoNets~\cite{cryptonets} and LoLa~\cite{lola}, similar to LeNet-5~\cite{lenet5}, achieve over 98\% accuracy on MNIST with degree-2 polynomials.
However, prior work~\cite{sisyphus} showed that training larger models (beyond AlexNet~\cite{alexnet} or VGG-11~\cite{vgg}) with such low-degree polynomials causes large approximation errors that destabilize training.
\textit{To date, PILLAR~\cite{pillar}, which uses a degree-4 polynomial, is the lowest-degree approximation that generalizes to deep neural networks.}

Implementing non-interactive PPML becomes increasingly difficult with deeper models.
While \textit{bootstrapping} enables evaluation of arbitrarily deep models, it is an expensive operation in FHE and remains the main scalability bottleneck. 
To avoid this cost, the literature has relied on \textit{Leveled Fully Homomorphic Encryption} (LFHE). 
LFHE eliminates the need for bootstrapping and allows faster private inference, but its computation is bounded by a fixed multiplicative depth determined by encryption parameters. 
As a result, the size and complexity of neural networks that can be evaluated under LFHE are strictly limited. 
\textit{The deepest CNNs shown to run with LFHE alone are AlexNet and VGG-11}~\cite{sisyphus}.
Larger models have so far required bootstrapping. 
Consequently, modern deep CNNs have largely been regarded as incompatible with LFHE, and bootstrapping has effectively become the default mechanism to address this limitation despite its high computational cost.
For example, AutoFHE executes ResNet-32 inference with 8 or 19 bootstrapping layers in its lower- or higher-accuracy configuration, respectively~\cite{autofhe}, resulting in heavy computational costs.

% This paper challenges that assumption.
This paper revisits the prevailing assumption that deep CNNs are incompatible with LFHE.
We argue that heavy reliance on bootstrapping has obscured an underexplored optimization space.
In FHE-based PPML, multiplicative depth determines the encryption parameters and the frequency of bootstrapping operations needed during evaluation, both of which have a dominant impact on inference time.
Crucially, the multiplicative depth of CNNs is not an immutable property of the model architecture, but emerges from how its computation is realized.
The network structure, degree of polynomial activations, linear transformations, data layout, function implementations, and placement of rescaling operations collectively govern depth.
Careful analysis and optimization of these components and their interactions can substantially reduce multiplicative depth without changing model functionality.
By minimizing multiplicative depth at every level of the model, rather than compensating for inefficiencies with bootstrapping, we show that deep CNNs can be evaluated efficiently under LFHE.

Guided by this philosophy, our goal is to design the \textit{fastest non-interactive PPML framework that does not compromise accuracy}, under LFHE. We achieve this by combining four techniques: (i) low-degree polynomial activations that attain the theoretical minimum multiplicative depth for non-linear activation functions, (ii) a suite of structural optimizations that remove redundant operations, (iii) an ensemble-based strategy that recovers accuracy comparable to ReLU levels, and (iv) a clustering technique to speed up the computation of convolutional layers.
Importantly, while this work focuses on LFHE, the proposed techniques are not limited to it.
They can be applied with FHE to run very deep models while substantially reducing the number of required bootstrapping operations.
%\newpage

\vspace{1mm}
\noindent \textbf{\underline{Contributions}}:
This work advances non-interactive PPML by enabling, for the first time, deep CNNs (e.g. ResNet-32) to run entirely under LFHE without bootstrapping, achieving accuracy comparable to plaintext ReLU models while outperforming prior work in both accuracy and inference time. 
Our contributions are fourfold:
\begin{itemize}[leftmargin=*, itemsep=1pt, topsep=2pt]
    % Poly2 activation + penalty function
    \item To the best of our knowledge, this is the first work to effectively train deep CNNs with \textit{degree-2 polynomial activations}. We introduce a novel \textit{penalty function} that ensures stable training and delivers accuracy on par with plaintext ReLU-based models.
    Our polynomial activation function achieves the theoretical minimum multiplicative depth of one for non-linear activations, whereas the lowest previously known stable alternative, PILLAR~\cite{pillar}, requires a 
    % multiplicative
    depth of three (see \S\ref{sec:model_approx_train}).

    % Node fusing
    % Weight redistribution
    % Tower Reuse
    \item We propose three \textit{structural optimization} techniques that reduce multiplicative depth of a model while preserving functional equivalence:
    (i) \textit{Node Fusing}, which replaces sequences of operations with equivalent functions, extending convolution--batch normalization folding with three new fusion cases involving polynomial activations and residual connections;
    (ii) \textit{Weight Redistribution}, which introduces a graph-level parameter redistribution framework for local coefficient normalization and compensating parameter updates; and
    (iii) \textit{Tower Reuse}, which tailors a sublevel-based CKKS modulus-chain parameterization to the model structure, enabling multiple multiplications to share a modulus.
    In the CKKS scheme used in this work~\cite{ckks}, ciphertexts carry a scaling factor that grows with multiplications.
    Rescaling keeps ciphertext values within range but consumes a level.
    Since each level corresponds to a modulus and homomorphic operations apply to all moduli, requiring fewer rescaling steps reduces the number of moduli, leading to faster homomorphic operations (see \S\ref{sec:struct_optmizations}).

    % Data layout + Ensemble
    \item We introduce a data layout and encoding strategy that utilizes unused ciphertext slots in FHE to pack multiple CNN model instances in the same ciphertext, enabling simultaneous \textit{ensemble inference} within a single homomorphic evaluation. This approach enables the use of multiple independently trained models to recover accuracy lost to low-degree polynomial approximations through careful reuse of ciphertext–plaintext structures, albeit with additional inference overhead (see \S\ref{sec:codesign_techniques}).
    
    % Clustering
    \item To offset the computational overhead during private inference from the choice of data layout that enables ensemble inference, we introduce a \textit{parameter clustering} technique tailored for convolutional layers. Our method takes into account the chosen data layout and groups similar parameters within each layer to reduce redundant plaintext encodings and homomorphic computation. This technique recovers much of the lost inference speed while preserving model accuracy (see \S\ref{sec:codesign_techniques}).
    
\end{itemize}

We evaluate our method on CIFAR \cite{cifar} and Tiny-ImageNet~\cite{imagenet} datasets using VGG~\cite{vgg} and ResNet~\cite{resnet} models, and show it achieves up to $4\times$ faster private inference than prior work with accuracy comparable to plaintext ReLU models (see \S\ref{sec:results}).

\section{Background}

\subsection{Threat Model}\label{sec:threat_model}
We consider non-interactive inference between a client and a server, following prior FHE inference work~\cite{cryptonets,lola,hemet,mpcnn,autofhe}.
We assume the server is honest-but-curious rather than malicious.
The client holds a private input and the CKKS secret key, while the server holds a CNN with plaintext parameters.
The client encrypts its input and sends only the ciphertext.
The server evaluates the network homomorphically under leveled FHE and returns encrypted logits that only the client decrypts.
Our goal is input confidentiality only.
We do not claim cryptographic model privacy, as the server evaluates plaintext parameters and a client observing returned logits learns about the model as in ordinary black-box inference.
Output leakage (e.g., model extraction~\cite{DBLP:conf/uss/TramerZJRR16} or membership inference~\cite{DBLP:conf/sp/ShokriSSS17} from returned logits) and side-channel leakage (e.g., timing~\cite{DBLP:conf/ndss/AkinsanyaB24} or memory-access patterns~\cite{DBLP:conf/sp/RakinCYF22}) are outside our scope.

\subsection{Leveled Fully Homomorphic Encryption}
\label{sec:lfhe}

% HE
\textit{Homomorphic Encryption} allows meaningful computations to be performed directly on encrypted data.
% CKKS
In this work, we employ Residue Number System (RNS) variant~\cite{rns-ckks} of the \textit{Cheon-Kim-Kim-Song} (CKKS)~\cite{ckks} encryption scheme, which supports approximate arithmetic over encrypted complex numbers.
% Describe CKKS
CKKS is an FHE scheme based on the Ring-Learning With Errors (RLWE) problem~\cite{rlwe}.
% CKKS spaces
The scheme operates across three spaces: message $\mathcal{M} = \mathbb{C}$, plaintext $\mathcal{P} = R_{Q} = \mathbb{Z}_{Q}[x]/(x^N+1)$, and ciphertext $\mathcal{C} = R_{Q} \times R_{Q}$, where $N = 2^k$ is a power of two (typically $2^{16}$ in modern PPML applications \cite{mpcnn, autofhe}) to enable efficient Number Theoretic Transform (NTT)-based polynomial arithmetic.
\textit{Encoding} maps a vector of $N/2$ numbers $\in \mathcal{M}$ into a degree-$N$ polynomial $\in \mathcal{P}$, and \textit{encryption} transforms a plaintext $\in \mathcal{P}$ into a ciphertext $\in \mathcal{C}$.
% RNS
In RNS-CKKS, the ciphertext modulus $Q$ is decomposed into $\mathcal{L}+1$ smaller primes $Q = \prod_{i=0}^\mathcal{L}{q_i}$, where $\mathcal{L}$ denotes the level.
A polynomial in $R_Q$ is represented as $\mathcal{L}+1$ polynomials in $R_{q_i} \forall i \in [0, \mathcal{L}] \cap \mathbb{Z}$, which allows computations using native 64-bit integer arithmetic~\cite{rns-ckks}.

% Supported homomorphic operations
CKKS supports homomorphic additions, multiplications, and rotations.
Additions and multiplications act element-wise on the encoded vector, akin to SIMD operations.
Rotations cyclically shift the vector of $N/2$ encoded messages by a step $\in \mathbb{Z}_{N/2}$.
% Encoding uses a scaling factor
During encoding, coefficients are scaled by a factor $\Delta$, rounded to the nearest integers, and reduced modulo $Q$.
Multiplying two encoded messages scales the result to $\Delta^2$, requiring \textit{rescaling} to restore the scale to $\Delta$.
Rescaling removes one RNS modulus, reducing the level by 1.
% Noise added during encryption and approximation errors from rescaling
% For security, encryption adds noise to ciphertexts.
% Each homomorphic operation increases this noise, and excessive noise eventually prevents correct decryption.
% CKKS also accumulates approximation errors from rescaling and floating-point encoding.
% Leveled FHE
When the level reaches zero, \textit{bootstrapping} can reset the ciphertext to a higher level, enabling unbounded computation.
However, it is computationally expensive and avoided whenever possible.
FHE without bootstrapping is called \textit{Leveled Fully Homomorphic Encryption} (LFHE).
In LFHE, the encryption parameters $(N,Q)$ are chosen to provide the desired security and sufficient levels to evaluate a circuit with known multiplicative depth.
In practical RNS-CKKS implementations, a special modulus $P$ is additionally used for hybrid key switching. Unlike $Q$, it does not contribute to the multiplicative depth but is included in security estimation.

CKKS is a well-suited FHE scheme for non-interactive privacy-preserving convolutional neural networks, as it natively supports approximate arithmetic over real and complex numbers. In contrast, schemes such as BGV and BFV operate over finite fields, requiring additional mechanisms to emulate real-valued computations.
For instance, fixed-point arithmetic can be implemented in BGV/BFV by selecting a sufficiently large plaintext modulus to avoid overflow. However, this often requires representing a single value across multiple ciphertexts (e.g., via CRT decomposition), which increases computational overhead and scales poorly with network depth.
TFHE, by design, targets boolean circuit evaluation, meaning that computations must be expressed as sequences of binary gates, resulting in significant overhead for neural network workloads.

\subsection{Challenges of Using Polynomial Activation}
\label{sec:relu_challenge}

ReLU, defined as $ReLU(x) = \max(0, x)$, is one of the most widely used activation functions in deep learning due to its simplicity and computational efficiency.
% Its piecewise linear structure introduces the non-linearity necessary for neural networks to capture complex patterns in data.
Its piecewise linear form provides the non-linearity required for neural networks to model complex patterns.
While alternative non-linear functions (e.g., sigmoid, tanh, GeLU) exist, ReLU remains the dominant activation in modern deep architectures due to its strong empirical performance and favorable optimization behavior~\cite{relu}.
% Accordingly, most plaintext models targeted for FHE inference are designed with ReLU activations.
% While ReLU is straightforward to implement in plaintext settings, it poses significant challenges in FHE environments due to its reliance on conditional branching.
% FHE can evaluate only polynomial functions as it supports only additions and multiplications, thus, any FHE representation of ReLU must be a polynomial approximation.
However, in FHE settings, ReLU is difficult to implement because it relies on conditional branching, whereas FHE natively supports only additions and multiplications. As a result, any FHE-compatible ReLU must be a polynomial approximation.
% Evidently, the higher degree the polynomial, the more precise the approximation.
% However, the multiplicative depth of the activation functions grows logarithmically to the degree of the polynomial used for approximation.
% Precisely, the multiplicative depth can be computed as $\delta = \ceil{\log_2{d}+1}$, where $d$ denotes the polynomial degree and the $+1$ refers to the coefficient multiplication.
Higher-degree polynomials produce more accurate approximations but incur larger multiplicative depth. For a degree-$d$ polynomial, the multiplicative depth is $\delta = \left\lceil \log_2 d \right\rceil + 1$, where the $(+1)$ accounts for coefficient multiplication. 
% In the literature, the lowest multiplicative depth for a polynomial approximation of ReLU is achieved by PILLAR~\cite{pillar}, which uses a degree-4 polynomial of the form $\sum_{i=0}^{4}{c_i x^i}$ with $\delta = 3$.
The lowest multiplicative depth for a ReLU approximation that generalizes to deep CNNs is achieved by PILLAR~\cite{pillar}, which uses a degree-4 polynomial $\sum_{i=0}^{4} c_i x^i$ with $\delta = 3$. 
% The theoretical minimum multiplicative depth for an activation function is one, achieved with a polynomial of the form $x^2 + c_1 x + c_0$.
The theoretical minimum multiplicative depth for an activation function is one, attainable with a quadratic polynomial $x^2 + c_1 x + c_0$.
% Polynomials like this have been effectively used in shallow CNNs like CryptoNets~\cite{cryptonets} and LoLa~\cite{lola}, but failed to work for deeper models due to the escaping activation problem~\cite{sisyphus}.
Such polynomials have been successfully used in shallow CNNs like CryptoNets~\cite{cryptonets} and LoLa~\cite{lola}, but they fail in deeper models due to the escaping activation problem~\cite{sisyphus}.
% In addition, polynomial approximations are inherently limited to bounded intervals of the input domain.
Moreover, polynomial approximations are accurate only over bounded input intervals.
% When substituted directly for ReLU during model training, these approximations often lead to a significant drop in model accuracy~\cite{sisyphus,coinn,pillar}.
When used as direct substitutes for ReLU during training, they often cause a substantial drop in model accuracy~\cite{sisyphus,coinn,pillar}.
% This section analyzes two primary causes of this degradation: \textit{escaping activations}, where intermediate layer outputs fall outside the polynomial's approximation interval and \textit{coefficient truncation}, which results from representing polynomial coefficients in fixed-point format to ensure compatibility with FHE arithmetic.
This section focuses on two main reasons for this degradation: \emph{escaping activations}, where intermediate layer outputs leave the polynomial's approximation range, and \emph{coefficient truncation}, arising from representing polynomial coefficients in fixed-point format to comply with FHE arithmetic.

\vspace{1mm}
\noindent\textbf{\underline{Escaping Activations.}}
% Let $p_d(x) = \sum_{k=0}^d a_k x^k$ denote a degree-$d$ polynomial approximation of the ReLU function, where each coefficient $a_k \in \mathbb{R}$ for all $k \in \{0, \dots, d\}$.
Let $p_d(x) = \sum_{k=0}^d a_k x^k$ be a degree-$d$ polynomial approximation of ReLU, with coefficients $a_k \in \mathbb{R}$, 
% for all 
$\forall k \in \{0, \dots, d\}$.
% These coefficients are typically obtained by minimizing the least-squares error between the polynomial and the ReLU function over a finite set of real-valued sample points $\{x_i\}_{i=1}^N \subset [-c, c]$, for some parameter $c > 0$:
These coefficients are typically obtained by minimizing the least-squares error between $p_d(x)$ and ReLU 
% over
on
a finite set of real-valued 
% sample 
points $\{x_i\}_{i=1}^N \subset [-c, c]$, for some parameter $c > 0$:
% \begin{equation*}
$
\min_{a_0, \ldots, a_d}\;\sum_{i=1}^N \left( ReLU(x_i)\;-\;p_d(x_i) \right)^2.
$
% \end{equation*}
% The resulting polynomial $p_d(x)$ provides a close approximation to ReLU within the interval $[-c, c]$.
% However, outside the interval, the polynomial behavior differs significantly from that of ReLU.
The resulting polynomial approximates ReLU well on $[-c, c]$ but can deviate substantially outside this range.
ReLU exhibits piecewise linear growth:
\begin{equation*}
ReLU(x)\;=\;
\begin{cases}
0, & \text{if } x \leq 0 \\
x, & \text{if } x > 0
\end{cases}
\quad \Rightarrow \quad ReLU(x)\;=\;\mathcal{O}(x)
\end{equation*}
% In contrast, the growth of $p_d(x)$ for $d \geq 2$ is polynomial, dominated asymptotically by the highest-degree term:
whereas any polynomial $p_d(x)$ with $d \geq 2$ grows as:
\begin{equation*}
p_d(x)\;=\;a_d x^d\;+\;a_{d-1} x^{d-1}\;+\;\cdots\;+\;a_1 x\;+\;a_0\;=\;\mathcal{O}(x^d)
\end{equation*}
% This fundamental mismatch in growth rates gives rise to the problem of \textit{escaping activations}, first identified by~\cite{sisyphus}.
% As inputs propagate through the layers of a neural network, the intermediate values passed into the polynomial activation can escape the intended interval $[-c, c]$, entering regions where $p_d(x)$ no longer approximates ReLU accurately.
This mismatch in growth rates leads to the \textit{escaping activations} phenomenon~\cite{sisyphus}: as inputs propagate through the network, intermediate values can move outside $[-c, c]$, into regions where $p_d(x)$ no longer resembles ReLU, invalidating the intended approximation. 
% As a consequence, the network can produce excessively large activations, which in turn cause exploding weights and rapid degradation of model performance unless specific modifications are made to the training procedure to contain them.
The resulting excessively large activations can cause exploding weights and rapid degradation of model performance, unless the training procedure is specifically adapted to control them.

\vspace{1mm}
\noindent\textbf{\underline{Coefficient Truncation.}}
% In privacy-preserving inference based on FHE, all computations are carried out using fixed-point arithmetic with limited precision.
In privacy-preserving FHE inference, all computations use fixed-point arithmetic with limited precision.
% This constraint restricts the range and resolution of values that can be accurately represented, affecting both the domain of activation function inputs and the magnitudes of polynomial approximation coefficients. 
This restriction limits both the range and resolution of representable values, constraining the domain of activation function inputs as well as the magnitudes of their polynomial approximation coefficients.
% The polynomial coefficients $\{a_k\}_{k=0}^{d}$ are typically derived via floating-point least-squares fitting.
% To enable fixed-point evaluation under FHE, these coefficients are quantized using:
The coefficients $\{a_k\}_{k=0}^{d}$ are usually obtained via floating-point least-squares fitting and then quantized for fixed-point evaluation under FHE:
$
\tilde{a}_k\;=\;\frac{\lfloor a_k \cdot 2^b \rceil}{2^b},
$
% where $b \in \mathbb{N}$ represents the number of fractional bits in the fixed-point format and $\lfloor \cdot \rceil$ denotes rounding to the nearest integer.
where $b \in \mathbb{N}$ is the number of fractional bits in the fixed-point format and $\lfloor \cdot \rceil$ denotes rounding to the nearest integer.
% The least-squares fitting procedure often produces small-magnitude coefficients, particularly for higher-degree monomials.
Least-squares fitting often produces small-magnitude coefficients, especially for higher-degree monomials.
% If any coefficient satisfies $\left|a_k\right| < 2^{-(b+1)}$, the quantized value $\tilde{a}_k$ becomes zero.
% This effectively discards the corresponding monomial $x^k$ from the approximation.
If $\left|a_k\right| < 2^{-(b+1)}$, the quantized coefficient $\tilde{a}_k$ becomes zero, effectively removing the monomial $x^k$ from the approximation.
% This phenomenon, referred to as \textit{coefficient truncation}~\cite{pillar}, changes the shape of the polynomial and can cause significant deviation from the intended ReLU behavior, even within the designated approximation interval $[-c, c]$.
This \textit{coefficient truncation}~\cite{pillar} changes the polynomial's shape and can cause substantial deviation from the intended ReLU behavior, even within the approximation interval $[-c, c]$.

\section{Training with Polynomial Approximation}
\label{sec:model_approx_train}

\subsection{Problem Setting}
\textbf{\underline{Dataset Configuration.}}
Consider a dataset $\mathcal{D} = \{(x_i, y_i)\}_{i=1}^{M}$ for a multi-class classification task, where $x_i \in \mathbb{R}^n$ denotes the feature vector and $y_i \in \{1, \dots, K\}$ the corresponding label.
The dataset contains $M$ samples distributed across $K$ classes.

\vspace{1mm}
\noindent\textbf{\underline{ReLU Network.}}
Let $f: \mathbb{R}^n \rightarrow \mathbb{R}^K$ denote an $L$-layer feed-forward neural network with ReLU activations. The network is parameterized by weight matrices $\{W^{(l)}\}_{l=1}^{L}$, where each $W^{(l)} \in \mathbb{R}^{n_l \times n_{l-1}}$ and $n_l$ denotes the number of neurons in layer $l$.
The network output is defined as $f(x) = h^{(L)}(x)$, where $h^{(0)}(x) = x$ and $h^{(l)}(x) = ReLU(W^{(l)} h^{(l-1)}(x))$ for all $1 \leq l \leq L$.
The activation function $ReLU(z)$ is applied coordinate-wise as $ReLU(z) = \max(0, z)$.

\vspace{1mm}
\noindent\textbf{\underline{Polynomial Network.}}
We define a network $g: \mathbb{R}^n \rightarrow \mathbb{R}^K$ with the same architecture as $f$, but replacing each ReLU activation with a degree-$d$ polynomial $p_d(z) = \sum_{k=0}^{d} a_k z^k$ that approximates the ReLU function on a bounded interval $[-c, c]$.
The polynomial coefficients $\{a_k\}_{k=0}^{d}$ are chosen so that the maximum approximation error satisfies: $\max_{z \in [-c, c]} | ReLU(z) - p_d(z) | \leq \varepsilon$, where $\varepsilon > 0$ is a prescribed approximation tolerance.
The resulting network output is given by $g(x) = h^{(L)}_{p_d}(x)$, where $h^{(0)}_{p_d}(x) = x$ and $h^{(l)}_{p_d}(x) = p_d(W^{(l)} h^{(l-1)}_{p_d}(x))$ for all $1 \leq l \leq L$.

\vspace{1mm}
\noindent\textbf{\underline{Challenges.}}
While polynomial activations enable compatibility with FHE schemes, they introduce several non-trivial optimization challenges, most notably \textit{escaping activations} and \textit{coefficient truncation}, as discussed in \S\ref{sec:relu_challenge}. 
% Escaping activations occur when intermediate outputs drift outside the interval $[-c, c]$, leading to large approximation errors. 
% Coefficient truncation occurs when polynomial coefficients are stored with limited fixed-point precision $b$, which changes the polynomial's shape and causes deviation from ReLU even inside the approximation interval.
Addressing these issues typically requires careful selection of the polynomial degree $d$, the approximation interval $[-c, c]$, and coefficient bounds $b$, together with regularization strategies designed to ensure that intermediate activations remain within the valid approximation range throughout training.

\subsection{Training Strategy}
\sloppy \textbf{\underline{Quantization-Aware Polynomial Fitting.}}
To avoid precision loss due to coefficient truncation, we incorporate fixed-point constraints directly into the coefficient estimation procedure using \textit{quantization-aware polynomial fitting}~\cite{pillar}.
Given an approximation interval $[-c, c]$ and fixed-point precision with $b$ fractional bits, we define a quantized input domain for polynomial fitting as:
$
X = \{ x \in [-c, c] \mid x = k \cdot 2^{-b},\ k \in \mathbb{Z}\}.
$
For each input $x_i \in X$, we compute scaled ReLU outputs
$
Y_i = 2^b \cdot ReLU(x_i)
$
so that regression targets are integers, which reduces the risk of precision loss during optimization.
We construct a matrix $B \in \mathbb{R}^{M \times (d+1)}$ with entries
$
B_{i,k} = x_i^k \ \text{for } 0 \le k \le d.
$
The coefficients $A = [a_0, \ldots, a_d]^\top \in \mathbb{Z}^{d+1}$ are obtained by solving a bounded integer least-squares problem:
$
\min_{A \in \mathbb{Z}^{d+1}} \| BA - Y \|_2^2 \ \text{subject to } a_k \in [ -2^{b - 1}, 2^{b - 1} - 1].
$
The resulting fixed-point polynomial is
$
p_d(x) = \sum_{k=0}^{d} ( a_k / 2^b ) x^k.
$
Fitting over $\mathbb{Z}^{d+1}$ rather than rounding a real-valued solution keeps every coefficient representable at precision $b$, so no monomial rounds to zero.
As described above, quantization is incorporated directly into the coefficient fitting procedure. This differs from standard \textit{quantization-aware training}~\cite{qat}, which is primarily used for model compression rather than constraining the activation approximation.

\vspace{1mm}
\noindent\textbf{\underline{Activation Regularization.}}
To avoid escaping activations, we introduce a \textit{regularization strategy} that constrains inputs to polynomial activation functions (i.e., pre-activations) to remain within the valid approximation interval $[-c, c]$.
The classification loss supervises only the final output and provides no mechanism to limit intermediate pre-activation values. 
Thus, minimizing it alone cannot prevent pre-activations from drifting outside the interval $[-c, c]$, where the polynomial diverges from ReLU and destabilizes training.
We add a layer-wise penalty that penalizes out-of-range pre-activations. 
Averaging it over the $L$ layers keeps its magnitude independent of depth, so a single $\zeta$ transfers across architectures.
For a mini-batch $\mathcal{B} \subset \mathcal{D}$, the following training loss is minimized:
% \vspace{-1mm}
% \begin{equation}\label{eq:penalty_loss}
%     \mathcal{L}_{\mathcal{B}} \;=\;\; \underbrace{\frac{1}{|\mathcal{B}|} \sum_{(x, y) \in \mathcal{B}} \ell_{\mathrm{CE}}\left(g\left(x\right), y\right)}_{\text{classification loss}}
%     \;+\; \zeta\;\underbrace{\frac{1}{L} \sum_{l=1}^{L} \left\| z_p^{(l)} \;-\; \mathrm{clip}\left(z_p^{(l)}; [-c, c]\right) \right\|_2}_{\text{clip-range penalty}}
% \end{equation}
\begin{equation}\label{eq:penalty_loss}
\begin{split}
\mathcal{L}_{\mathcal{B}}
    &= 
    \underbrace{\frac{1}{|\mathcal{B}|} \sum_{(x, y) \in \mathcal{B}} 
        \ell_{\mathrm{CE}}\!\left(g(x), y\right)}_{\text{classification loss}}
    \\
    &\quad+\;
    \zeta\;\underbrace{\frac{1}{L} \sum_{l=1}^{L} 
        \left\| z_p^{(l)} - \mathrm{clip}\!\left(z_p^{(l)}; [-c, c]\right) \right\|_2}%
        _{\text{clip-range penalty}}
\end{split}
\end{equation}
where $\ell_{\mathrm{CE}}$ is the \textit{cross-entropy loss}, $z_p^{(l)}$ are pre-activations at layer $l$, i.e., outputs of affine transformations $W^{(l)} h^{(l-1)}_{p_d}(x)$ before the polynomial activation.
The function $\mathrm{clip}(z; [-c, c])$ is applied element-wise and defined as: $\mathrm{clip}(z; [-c, c])_i = \max(-c, \min(z_i, c))$.
The second loss term penalizes pre-activations
that lie outside the interval $[-c, c]$, with strength controlled by the regularization parameter $\zeta>0$.
See \S\ref{sec:proof_penalty} for a proof of the penalty function's correctness. 
Training can still become unstable for two reasons: (i) in early epochs pre-activations can grow unbounded before the model learns to contain them and (ii) using full regularization weight $\zeta$ from the start can let the penalty dominate and destabilize optimization. 
As additional strategies, we also consider \textit{clipping pre-activations during training} and introducing a \textit{warm-up schedule for $\zeta$}.

\vspace{1mm}
\noindent\textbf{\underline{Pre-Activation Clipping.}}
During early training stages, network weights remain close to their random initialization, and the optimization process does not immediately constrain the pre-activation values to lie within the target approximation interval $[-c, c]$.
As a result, unbounded pre-activations may arise before the model learns to keep them within range, 
% leading to 
causing
training instability that can degrade model behavior. 
% beyond recovery
To prevent this, we restrict pre-activations to the interval $[-c, c]$ before evaluating the polynomial activation, which is achieved through a \textit{clipping strategy} $\mathrm{clip}(z; [-c, c])$~\cite{pillar}:
% \vspace{-1mm}
$$
h^{(l)}_{p_d}(x) \;=\; p_d\left(\mathrm{clip}\left(W^{(l)} h^{(l-1)}_{p_d}(x); [-c, c]\right)\right).
$$
Importantly, this \textit{clipping operation is applied only during training} to stabilize learning.
It is performed after computing the clip-range penalty to ensure that gradients from the regularization term, which encourages the network to keep pre-activations within the approximation interval $[-c, c]$, are preserved and not masked by the clipping.
At inference time, the clipping function is removed:
$
h^{(l)}_{p_d}(x) \;=\; p_d\left(W^{(l)} h^{(l-1)}_{p_d}(x)\right)
$.
The model, having learned to constrain pre-activations during training, is expected to remain within the approximation interval without explicit clipping at inference.
% \vspace{1mm}

\vspace{1mm}
\noindent\textbf{\underline{Regularization Warm-up.}}
While activation regularization and pre-activation clipping are 
% both 
necessary for training stability, applying the full regularization strength from the outset can lead to numerical instability, particularly in larger models.
In the early epochs, many pre-activation values lie outside the target interval $[-c, c]$ across several layers.
The clip-range penalty adds a contribution from every such layer, so the penalty term becomes very large.
In extreme cases, this can cause the total loss to become numerically undefined.
% \hl{<--- make sentense shorter or longer}

To address this issue, the regularization strength $\zeta$ is progressively increased over the initial training epochs, which is implemented through a \textit{regularization warm-up schedule}~\cite{pillar}.
Let $T_{warm}$ denote the total number of warm-up epochs. For each epoch $t$, we define a \textit{time-dependent regularization weight} $\zeta_t$:
% \vspace{-1mm}
% as:
$$
\zeta_t \;=\;
\begin{cases}
\alpha_t \cdot \zeta & \text{if } t \leq T_{warm} \\
\zeta & \text{if } t > T_{warm}
\end{cases},
$$
where $\{\alpha_t\}_{t=1}^{T_{warm}}$ is a fixed sequence of scaling factors satisfying $0 < \alpha_1 < \cdots < \alpha_{T_{warm}} < 1$. 
In practice, we empirically find that setting $T_{warm} = 4$ provides stable convergence. 
We use a predefined sequence of scaling factors: $\alpha = \left\{ \frac{1}{100}, \frac{1}{50}, \frac{1}{10}, \frac{1}{5} \right\}$, which produces the epoch-wise regularization strength: $\zeta_t \;\in\; \left\{ \frac{\zeta}{100}, \frac{\zeta}{50}, \frac{\zeta}{10}, \frac{\zeta}{5}, \zeta, \zeta, \dots \right\}$.
This warm-up schedule serves only as a training safeguard and ensures that the regularization penalty is introduced progressively in the initial training epochs without suffering from exploding loss values. 
Figure~\ref{fig:vis_training} shows an overview of the training process with polynomial approximation.
\begin{figure}[!t]
    \centering
    \includegraphics[width=0.98\linewidth]{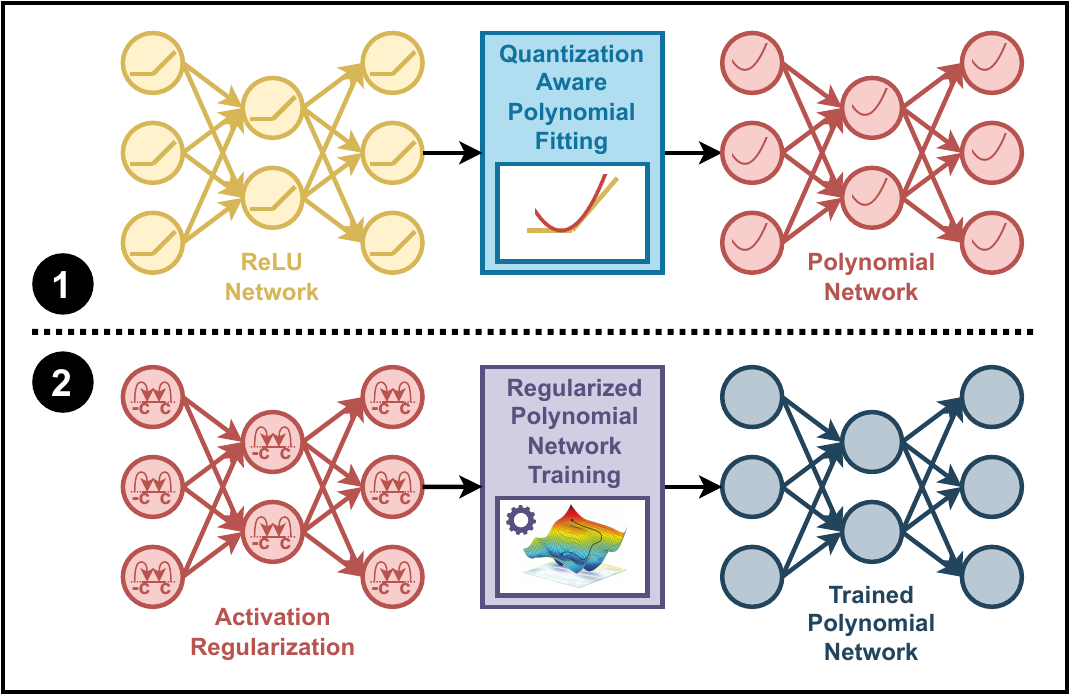}
    \vspace{-2mm}
    \caption{\small Overview of the two-stage training pipeline of a neural network with polynomial approximation. \protect\circled{1} A ReLU network is approximated using \textit{quantization-aware polynomial fitting} to obtain an initial polynomial network. \protect\circled{2} The polynomial network is then trained through \textit{regularized polynomial network training}, consisting of \textit{activation regularization}, \textit{pre-activation clipping}, and \textit{warm-up scheduling}, to ensure stable convergence.}
    \label{fig:vis_training}
\end{figure}

% While similar classes of techniques have been used in prior work (e.g., PILLAR~\cite{pillar}), those results are established for higher-degree polynomial activations and do not directly demonstrate stable training with quadratic activations in deep networks. 
% This setting is substantially more constrained, as degree-2 approximations are more sensitive to activation range and coefficient quantization, and errors can accumulate more rapidly across layers. 
% The training procedure described in this section shows that, with appropriate control of these factors, deep CNNs can be trained reliably using quadratic activations.
While similar classes of techniques have been explored in prior work (e.g., PILLAR~\cite{pillar}), they are primarily established for higher-degree polynomial activations. 
In contrast, quadratic activations are more constrained, as in this setting error accumulates more rapidly across layers, which makes training deep networks more challenging in practice. 
The training procedure described in this section shows that, despite these constraints, deep CNNs can still be trained reliably using quadratic activations.
% \\\hl{LINE -- Add 4 lines to this section}
% \\\hl{LINE -- Ideally to Section 3.2}
% \\\hl{LINE: < 4 lines leaves too much spacing for equations}
% \\\hl{LINE: = 4 lines best spacing}
% \\\hl{LINE: > 4 lines starts to be cramped}
% \\\hl{LINE}

\section{Structural Optimizations}\label{sec:struct_optmizations}

\subsection{Node Fusing}\label{sec:node_fusing}
Our first optimization is \textit{node fusing}, which merges consecutive computational nodes into a single equivalent function. 
This reduces both the number of operations and multiplicative depth. 
Formally, two sequential nodes $F(G(x))$ are replaced with $H(x)$ such that $F(G(x)) = H(x)$, with $H(x)$ more efficient to evaluate.
Following this principle, we consider four cases that arise in polynomial neural networks where compositions are fused into a single, more efficient operation, as illustrated in Figure \ref{fig:node_fusing}.
Together, they form a broader node-fusion framework for PPML.
Case 2 corresponds to the well-known convolution--batch normalization folding~\cite{markuvs2018fusing}, while Cases 1, 3, and 4 extend node fusing to polynomial activations and residual connections, reducing the number of homomorphic operations and multiplicative depth.
Since all cases use batch normalization, recall its polynomial form:
\begin{equation*}
    B(x) = b_1 x + b_0
\end{equation*}
\noindent with $b_1 = \frac{\gamma}{\sigma}, \; b_0 = \beta_b - b_1 \mu$, where $\mu$, $\sigma$ are the input mean and standard deviation, and $\gamma$, $\beta_b$ are learnable parameters.
We outline the fusion cases below. See \S\ref{sec:node_fusing_derivation} for full derivations.

\vspace{1mm}
\noindent\underline{\textbf{Case 1: $P(B(x)) \mapsto P(x)$.}}
When a polynomial activation $P(\cdot)$ follows batch normalization, the two can be fused into a single quadratic polynomial.
If the activation $P(x) = c_2 x^2 + c_1 x + c_0$, then $P(B(x))$ can be represented as $P(B(x)) = p_2 x^2 + p_1 x + p_0$, with fused coefficients $p_2 = b_1^2 c_2, \; p_1 = b_1(2 b_0 c_2 + c_1), \; p_0 = b_0^2 c_2 + b_0 c_1 + c_0.$
% \vspace{2mm}

\vspace{1mm}
\noindent\underline{\textbf{Case 2: $B(C(x)) \mapsto C(x)$.}}
Batch normalization can be fused into a preceding convolution~\cite{markuvs2018fusing}.
A convolution function is given by $C(x) = \sum_i w_i x_i + \beta_c$, where $w_i$ are the weights and $\beta_c$ is the bias.
The fused form becomes a single convolution $B(C(x)) = \sum_i \omega_i x_i + \alpha$, with fused coefficients $\omega_i = b_1 w_i, \; \alpha = b_1 \beta_c + b_0$.
% \vspace{2mm}

\vspace{1mm}
\noindent\underline{\textbf{Case 3: $P(B_X(x)+B_Y(y)) \mapsto S(x,y)$.}}
In residual networks, skip connections often merge two batch-normalized branches by summing them before applying an activation $P(\cdot)$. 
Specifically, $P(B_X(x) + B_Y(y))$.
Here, $B_X$ and $B_Y$ are independent batch-normalization layers applied to inputs $x$ and $y$, respectively.
This structure can be fused into a single quadratic bivariate polynomial $S(x,y) = d_{X2} x^2 + d_{Y2} y^2 + d_{XY} xy + d_X x + d_Y y + d_0$, with coefficients $d_{X2} = c_2 b_{X1}^2, \; d_{Y2} = c_2 b_{Y1}^2, \; d_{XY} = 2 c_2 b_{X1} b_{Y1}$, $d_X = b_{X1}(2 c_2(b_{X0}+b_{Y0}) + c_1), \; d_Y = b_{Y1}(2 c_2(b_{X0}+b_{Y0}) + c_1)$, $d_0 = c_2(b_{X0}+b_{Y0})^2 + c_1(b_{X0}+b_{Y0}) + c_0$.
% \vspace{2mm}

\vspace{1mm}
\noindent\underline{\textbf{Case 4: $P(B_X(x)+y) \mapsto S(x,y)$.}}
Some skip connections use identity shortcuts, where the input $y$ is added directly to the batch-normalized branch $B_X(x)$ before applying the activation.
This structure can be fused into a quadratic bivariate polynomial represented as
$S(x,y) = d_{X2} x^2 + d_{Y2} y^2 + d_{XY} xy + d_X x + d_Y y + d_0$,
with $d_{X2} = c_2 b_{X1}^2, \; d_{Y2} = c_2, \; d_{XY} = 2 c_2 b_{X1}$, $d_X = b_{X1}(2 c_2 b_{X0} + c_1), \; d_Y = 2 c_2 b_{X0} + c_1$, $d_0 = c_2 b_{X0}^2 + c_1 b_{X0} + c_0$.

\vspace{1mm}
Node fusing collapses batch normalizations, activation functions, convolutional layers, and skip connections into single fused polynomials, eliminating redundant nodes and reducing multiplicative depth.
This technique is applied after model training but before inference, transforming the model structure into a functionally equivalent, more efficient form.
% \\\hl{LINE}
% \\\hl{LINE}

\begin{figure}
    \centering
    \includegraphics[width=1.0\linewidth]{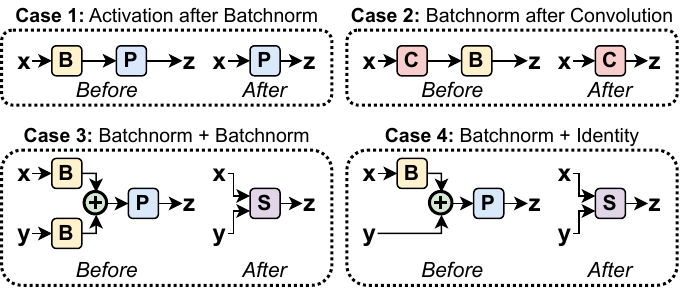}
    \vspace{-6mm}
    \caption{\small Illustration of node fusing cases. Nodes are batch normalization (\protect\BN), convolution (\protect\Conv), univariate (\protect\Poly)~and bivariate (\protect\Bivar) polynomial activations, and addition (\protect\Add).}
    \label{fig:node_fusing}
    \vspace{-3mm}
\end{figure}

\subsection{Weight Redistribution}\label{sec:weight_redistribution}
Our second optimization is \textit{weight redistribution},
a graph-level parameter redistribution framework that enables coefficient normalization while preserving the functional equivalence of privacy-preserving CNN models.
Its goal is to reduce the multiplicative depth of selected operations by one by normalizing the coefficient associated with the operation's longest multiplicative path.
Specifically, the optimization targets average pooling, batch normalization, and polynomial activation.
For average pooling, normalization eliminates the divisor by setting it to one.
For polynomial functions, coefficients are normalized so that the highest-order coefficient is set to one.
For example, the polynomial activation $c_2 x^2 + c_1 x + c_0$ is transformed into $x^2 + \bar{c}_1 x + \bar{c}_0$, reducing its multiplicative depth from two to one, \textit{the theoretical minimum for an activation function}.
Applied independently, however, these normalizations alter the network function, breaking model equivalence.
Weight redistribution addresses this limitation by propagating compensating parameter updates throughout the computational graph, enabling local coefficient normalizations while preserving model equivalence.
We refer to the nodes that initiate redistribution as \textit{donors}, and to nodes that adjust their parameters to restore equivalence as \textit{receivers}.
We represent the network as a directed graph and traverse it to identify all eligible donor nodes. 
For each donor, redistribution can be applied either \textit{forward} or \textit{backward}, affecting receivers among the donor's successors or predecessors, as illustrated in Figure \ref{fig:weight_redistribution}.

\vspace{1.5mm}
\noindent\underline{\textbf{Update Forward.}}
% \subsubsection{Update Forward}

\vspace{1mm}
\noindent\textbf{Donors.} We first consider donor updates in the forward direction.  
Our goal is to construct a normalized function $\bar{F}(x)$ such that $\upsilon \bar{F}(x) = F(x)$, where $F(x)$ is the original donor function and $\upsilon$ is the update term to be propagated to receivers.

% \vspace{1mm}
\noindent\textit{Average Pooling.}
The average pooling operation can be represented as $\mu(x) = k^{-1} \sum_i{x_i}$, where $k$ is the kernel size.
In the normalized function we set \(k = 1\), reducing the operation to $\bar{\mu}(x) = \sum_i x_i$.
Thus, for the equality to be true, we have $\upsilon = k^{-1}$.

% \vspace{1mm}
\noindent\textit{Polynomial Functions.}
Batch normalization and activation are polynomial functions.
For a univariate polynomial $P(x) = \sum_{i=0}^{d}{c_i x^i}$ of degree $d$, and a normalized polynomial $\bar{P}(x) = x^d + \sum_{i=0}^{d-1}{\bar{c}_i x^i}$, where $\bar{c}_i$ are the normalized coefficients, we have $\upsilon = c_d$ and $\bar{c}_i = c_i \upsilon^{-1}$ for $\upsilon \bar{P}(x) = P(x)$ to hold.
The extension to the bivariate polynomial activation is analogous and deferred to \S \ref{sec:weight_redistribution_derivation}.

% \vspace{1mm}
\sloppy \noindent\textbf{Receivers.}
After a donor update, receivers following the donor node must be adjusted to maintain model equivalence.  
The original composed function is $G(F(\cdot))$, with $F(\cdot)$ being the donor and $G(\cdot)$ the receiver.  
Updating the donor to $\bar{F}(\cdot)$ without compensating the receiver results in $G(\upsilon^{-1} F(\cdot))$, breaking model equivalence.
To restore it, we modify the receiver to $\bar{G}(\cdot)$ so that it satisfies $\bar{G}(\bar{F}(\cdot)) = \bar{G}(\upsilon^{-1} F(\cdot)) = G(F(\cdot))$.
Receivers fall into two categories: kernel functions and polynomial functions, each requiring different parameter updates to absorb the propagated update.

% \vspace{1mm}
\noindent\textit{Kernel Functions.}
Convolution, linear, and average pooling share the form $K(x) = \sum_i w_i x_i + \beta$ ($w_i = k^{-1}$ and $\beta = 0$ for average pooling).
To maintain model equivalence, we update the kernel function such that $\bar{K}(\upsilon^{-1} x) = K(x)$ is valid.
For that, we set $\bar{w}_i = w_i \upsilon$ and $\bar{\beta} = \beta$.

% \vspace{1mm}
\noindent\textit{Polynomial Functions.}
Batch normalization and polynomial activation update as $\bar{P}(\upsilon^{-1} x) = P(x)$, implying that $\bar{c}_i = c_i \upsilon^i$.
For the bivariate polynomial activation, in $\bar{S}(\upsilon^{-1} x, y) = S(x,y)$ the index $i$ corresponds to the exponent of $x$, while in $\bar{S}(x, \upsilon^{-1} y) = S(x,y)$ it corresponds to the exponent of $y$.

\vspace{1.5mm}
\noindent\underline{\textbf{Update Backward.}}

\vspace{1mm}
\noindent\textbf{Donors.}
Backward updates propagate the network modifications in the backward direction.
The equations are essentially the inverse of forward updates.  
For \textit{average pooling}, the update is identical. Since $\bar{\mu}(x)$ is linear, one can show that $\bar{\mu}(\upsilon x) = \upsilon \bar{\mu}(x)$. For \textit{polynomial donors}, $\bar{P}(\upsilon x) = P(x)$ must hold; hence, $\upsilon = c_d^{1/d}$ and $\bar{c}_i = c_i \upsilon^{-i}$.

\vspace{1mm}
\noindent\textbf{Receivers.}
To preserve model equivalence, receivers must satisfy $\bar{G}(x) = \upsilon G(x)$.
For that, \textit{kernel functions} update as $\bar{K}(x) = \upsilon K(x)$, which implies that $\bar{w}_i = w_i \upsilon$ and $\bar{\beta} = \beta \upsilon$.
And \textit{polynomial functions} update according to $\bar{P}(x) = \upsilon P(x)$; hence, $\bar{c}_i = c_i \upsilon$.

% \vspace{1mm}
Weight redistribution is performed after model training, immediately following node fusing, modifying the network's parameters, and prior to inference.
It reduces multiplicative depth, while maintaining model equivalence.
Figure \ref{fig:weight_redistribution} illustrates common cases for CNNs, including residual connections, where redistribution must be applied consistently along all converging or diverging branches.
Refer to \S\ref{sec:weight_redistribution_derivation} for detailed derivations for all cases.

% \clearpage

\begin{figure}
    \centering
    \includegraphics[width=1.0\linewidth]{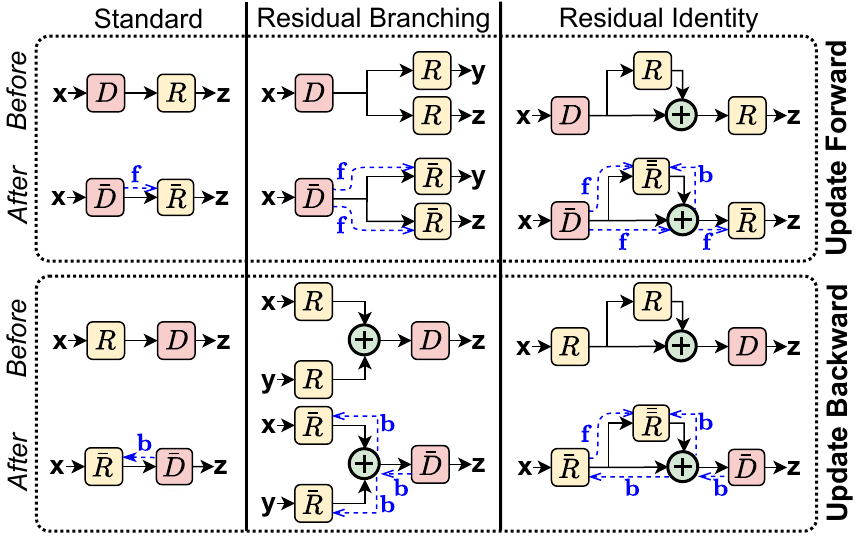}
    \caption{\small
        % Representative subset of weight redistribution cases.
        Illustration of weight redistribution cases.
        % $\bar{R}$ $\bar{\bar{R}}$
        \protect\Don represents a donor and \protect\Rec a receiver.
        The number of bars~($\bar{\cdot}$) on a \protect\Don or \protect\Rec represents the number of updates the node received.
        \textcolor{blue}{$\mathbf{f}$} represents a forward and \textcolor{blue}{$\mathbf{b}$} a backward update. 
    }
    \label{fig:weight_redistribution}
    % \vspace{-2mm}
    % \vspace{-4pt}
\end{figure}

\subsection{Tower Reuse}\label{sec:tower_reuse}

In CKKS, the RNS moduli are typically selected as primes close to the scaling factor $\Delta$.
This choice is motivated by the fact that the scaling factors of operands multiply after each homomorphic multiplication, yielding a ciphertext with scale $\Delta^2$.
To prevent exponential growth of the scale, the \emph{rescale} operation is used to reduce the scale back to $\Delta$, consuming one modulus from the modulus chain in the process.
For this reason, the number of moduli depends on the application's multiplicative depth $\delta$.
Since homomorphic operations are carried out across all moduli, their number directly impacts computational efficiency.
Rescaling also introduces approximation errors: the further the dropped modulus from $\Delta$, the larger the error.
This issue is especially relevant for small scaling factors, as fewer RNS primes lie relatively close to $\Delta$.
These errors can accumulate and amplify through subsequent homomorphic operations, increasing the likelihood of incorrect results.
Using a large $\Delta$ has drawbacks, as it requires larger RNS moduli and therefore increases the ciphertext modulus $Q = \prod_{i=0}^{\mathcal{L}} q_i$, where $\mathcal{L}$ is the number of levels and $q_i$ are the RNS primes, which reduces the security level.
One can increase security by increasing the polynomial degree $N$.
However, it is undesirable because operating on larger polynomials slows computation, even when amortized over additional ciphertext slots, as the computational complexity of FHE operations scale superlinearly with polynomial degree and modulus size.
Thus, for a fixed security level and polynomial degree, there is a maximum allowable $Q$.
Consequently, for an application-defined $\delta$, this imposes a bound on the maximum $\Delta$.
In practice, when $\delta$ is large, as is common in deep learning models, the resulting $\Delta$ is small.

\vspace{1mm}
\noindent\textbf{\underline{Proposed Solution.}}  
We introduce \textit{tower reuse}, a more general method for determining the RNS moduli by introducing \emph{sublevels} within each level.
Instead of enforcing $q_i \approx \Delta$, we allow $q_i \approx \Delta^\ell$, where $\ell$ denotes the number of sublevels.
Rescaling is performed only when a ciphertext scale exceeds the sublevel capacity of the modulus to be dropped.
Formally,
$\lambda(x) > \lambda(q_i) \implies \downarrow x$,
where $x$ represents a ciphertext, $\downarrow$ denotes the rescale operation, and $\lambda(x) = \round{\log_{\Delta}{\Delta_x}}$ returns the sublevel of a ciphertext, plaintext, or modulus:
with $\Delta_x$ referring to the scale of $x$ and $\Delta$ being the default scaling factor defined by the encryption parameters.

\vspace{1mm}
\noindent\textbf{\underline{Example.}} % of Tower Reuse.}}
Suppose following an activation function, there is a convolution.
Let modulus $q_i \approx \Delta^2$, thus $\lambda(q_i) = 2$.
Consider an input $x$ with scaling factor $\Delta_x = \Delta$, giving $\lambda(x) = 1$. 
Applying a quadratic activation $y = x^2 + c_1 x + c_0$ yields $\lambda(y) = 2$, with $\lambda(c_1) = 1$ and $\lambda(c_0) = 2$. 
Next, a convolution $z = \sum_i \omega_i y_i + \alpha$ is performed.
Since $\lambda(y) = 2$, by setting $\lambda(\omega) = 1$, we obtain $\lambda(z) = 3$.
As $\lambda(z) > \lambda(q_i)$, a rescale is triggered.  
The rescale operation drops $q_i$ from the modulus chain, reducing the level $\Lambda(\cdot)$ by one and resetting the sublevel to one:
\[
\Lambda(\downarrow z) = \Lambda(z) - 1, \qquad 
\lambda(\downarrow z) = \lambda(z) - \lambda(q_i) = 1.
\]
In terms of $\Delta$, a ciphertext $z$ with $\Delta_z = \Delta^3$ rescaled by $q_i \approx \Delta^2$ results in output scale $\Delta_{\downarrow z} \approx \Delta$.

Tower reuse is an automatic inference-time technique governed by $q_i$, $\Delta$, and the network structure. It improves computational efficiency and reduces approximation errors.
Like delayed (lazy) rescaling techniques, it seeks to reduce the number of rescaling operations. However, it achieves this objective through a different mechanism.
The efficiency gains arise from decreasing the number of RNS moduli by a factor of $\ell$, since the method allows $\ell$ levels of multiplication to be performed before rescaling, and from the corresponding reduction in the number of required rescale operations.
Rather than heuristically determining when rescaling should occur, tower reuse introduces a sublevel-based CKKS parameterization that incorporates the multiplication-level requirements of the target PPML architecture into the modulus-chain design, enabling multiple multiplications to share a modulus.
The reduction in approximation error is likewise due to the fewer rescale operations, and also to the fact that the method permits the use of smaller scaling factor $\Delta$ values together with larger moduli $q_i$, making the moduli relatively closer to $\Delta^\ell$.

\subsection{Impact on RNS Levels}

The combination of our polynomial activation function with the structural optimizations reduces the number of RNS levels $\mathcal{L}$ required for multiplication compared to PILLAR \cite{pillar}, which previously achieved the lowest $\mathcal{L}$: from 79 to 16 for VGG-16, 87 to 18 for ResNet-18, 97 to 20 for ResNet-20, and 157 to 32 for ResNet-32.

\vspace{1mm}
\noindent\textbf{\underline{Level Analysis.}}
Table \ref{tab:levels} summarizes the number of levels $\mathcal{L}$ required for multiplication in VGG and ResNet models under various optimization techniques.
To illustrate, we focus on ResNet-18, which comprises 17 convolutions, batch normalizations, and activations, along with a single average pooling and linear layer:

\begin{itemize}[leftmargin=*, itemsep=1pt, topsep=3pt]
    \item $P_4$. Prior work with the lowest $\mathcal{L}$~\cite{pillar} uses a degree-4 polynomial activation, giving $\mathcal{L}=87$: convolutions, batch normalizations, average pooling, and linear layers each contribute $\mathcal{L}=1$, and the activation contributes $\mathcal{L}=3$.

    \item $P_2$.
    By replacing the activation function with a degree-2 polynomial, as described in \S\ref{sec:model_approx_train}, the required level for activations drops to 2, reducing the model $\mathcal{L}$ to 70.

    \item $P_2 F$.
    Employing node fusing (\S \ref{sec:node_fusing}) on $P_2$ removes all batch normalizations, which lowers $\mathcal{L}$ to 53.
    
    \item $P_2R$. Applying weight redistribution (\S\ref{sec:weight_redistribution}) to $P_2$ instead yields $\mathcal{L}=35$. The highest-order coefficients in batch normalizations and polynomial activations are normalized to one, simplifying them to $x+b_0$ ($\mathcal{L}=0$) and $x^2+c_1x+c_0$ ($\mathcal{L}=1$); the average-pooling divisor is likewise normalized to one ($\mathcal{L}=0$). These simplifications give $\mathcal{L}=35$.

    \item $P_2 F R$.
    Combining node fusing and weight redistribution retains $\mathcal{L}=35$, but reduces homomorphic operations, since node fusing eliminates batch normalizations.
    
    \item $P_2FRT$. Applying tower reuse (\S\ref{sec:tower_reuse}) to $P_2FR$ lets a non-linear polynomial activation and subsequent linear kernel share a level, halving the number of levels to $\mathcal{L}=18$, following $\mathcal{L}_{P_2 F R T} = \ceil{\frac{\mathcal{L}_{P_2 F R}}{2}}$ for VGG and ResNet models.
\end{itemize}

In summary, by combining a lower-degree polynomial activation with structural optimizations, we reduce the number of levels by nearly a factor of five compared to prior work.
This reduction results from complementary optimizations: node fusing eliminates unnecessary intermediate nodes, weight redistribution reduces the depth of selected operations through coefficient normalization, and tower reuse decreases the number of required rescaling operations through sublevel-based CKKS parameterization.

\begin{table}[!t]
    % \label{tab:levels}
    \centering
    \caption
    {\small
        Number of RNS levels required for homomorphic multiplications in VGG and ResNet models. 
        % We evaluate six variants:
        $P_4$: Degree-4 polynomial activation used in PILLAR~\cite{pillar}. 
        $P_2$: Degree-2 polynomial activation (\S\ref{sec:model_approx_train}). 
        $P_2 F$: $P_2$ with node fusing (\S\ref{sec:node_fusing}). 
        $P_2 R$: $P_2$ with weight redistribution (\S\ref{sec:weight_redistribution}). 
        $P_2 FR$: $P_2$ with node fusing and weight redistribution. 
        $P_2 FRT$: $P_2 FR$ with tower reuse (\S\ref{sec:tower_reuse}).
    }
    \begin{tabular}
    {
        >{\hspace{-2pt}}c<{\hspace{-2pt}} | 
        >{\hspace{ 0pt}}c<{\hspace{ 0pt}} | 
        >{\hspace{ 0pt}}c<{\hspace{ 0pt}} | 
        >{\hspace{ 0pt}}c<{\hspace{ 0pt}} | 
        >{\hspace{ 0pt}}c<{\hspace{ 0pt}} | 
        >{\hspace{-3pt}}c<{\hspace{-3pt}} | 
        >{\hspace{-5pt}}c<{\hspace{-5pt}}
    }
        \hline
        {\bf Model} &
        {\bf $P_4$} &
        {\bf $P_2$} &
        {\bf $P_2 F$} &
        {\bf $P_2 R$} &
        {\bf $P_2 FR$} &
        {\bf $P_2 FRT$} \\ \hline

        VGG-16    &  79 &  64 &  51 &  31 &  31 &  16 \\
        ResNet-18 &  87 &  70 &  53 &  35 &  35 &  18 \\
        ResNet-20 &  97 &  78 &  59 &  39 &  39 &  20 \\
        ResNet-32 & 157 & 126 &  95 &  63 &  63 &  32 \\
        % ResNet-34 & 167 & 134 & 101 &  67 &  67 &  34 \\
        \hline
    \end{tabular}
    \label{tab:levels}
    
    % \begin{tablenotes}
    %     \footnotesize
    %     \item[a] $P_4$: Degree-4 polynomial activation used in PILLAR~\cite{pillar}.
    %     \item[b] $P_2$: Degree-2 polynomial activation (\S\ref{sec:model_approx_train}).
    %     \item[c] $P_2 F$: $P_2$ with node fusing (\S\ref{sec:node_fusing}).
    %     \item[d] $P_2 R$: $P_2$ with weight redistribution (\S\ref{sec:weight_redistribution}).
    %     \item[e] $P_2 FR$: $P_2$ with node fusing and weight redistribution.
    %     \item[f] $P_2 FRT$: $P_2 FR$ with tower reuse (\S\ref{sec:tower_reuse}).
    % \end{tablenotes}
\end{table}
\section{Co-design Techniques}\label{sec:codesign_techniques}
\subsection{Data Layout}\label{sec:data_layout}

\textit{Data layout} defines how model inputs and weights are mapped into FHE ciphertext and plaintext slots during encoding.
In this work, we adopt the \textit{HW layout}~\cite{chet}, which assigns each input channel to a separate ciphertext and fills its slots with spatial values indexed by two–dimensional coordinates.
Unlike the naive layout~\cite{cryptonets}, where each ciphertext holds only one value and incurs significant overhead, HW layout is more efficient, packing $h \cdot w$ values per ciphertext, with $h$ and $w$ denoting input height and width. 
Figure~\ref{fig:hw_layout} shows this for a $3 \times 4 \times 4$ input with a $3 \times 2 \times 2$ filter (stride 1, no padding), illustrating slot arrangements and gaps that may appear after convolution or other kernel operations. 
These gaps could be removed by remapping, but that requires additional multiplications and increases multiplicative depth. 
Instead, we use an \emph{adaptive (lazy) mapping} strategy: some slots are left unused, and subsequent weights, biases, and coefficients are aligned with the slot arrangement of the preceding layer's output. 
HW layout does not always fully utilize slots, especially when input channels have far fewer values than available slots.
More compact layouts such as \emph{CHW layout}~\cite{chet} and its variants~\cite{mpcnn} address this by mapping multiple channels into one ciphertext, which reduces ciphertext counts and decreases the number of multiplications and additions, ultimately leading to faster inference. 
However, PPML relies on polynomial approximations that reduce accuracy compared to plaintext ReLU models.
For this reason, we retain the less efficient HW layout, repurposing unused slots to improve accuracy instead of minimizing inference time. 
To offset the added latency, we introduce a parameter clustering technique.

\subsection{Clustering of Convolution Parameters}\label{sec:weight_clustering}
Convolution layers account for the majority of parameters in deep neural networks, and their handling under FHE is particularly costly.
Each convolution weight must be encoded repeatedly to match varying ciphertext layouts and levels, which depend on kernel position and layer structure.
As a result, the same scalar weight is often redundantly encoded many times. 
Pre-encoding all possible weight–layout combinations would lead to prohibitive memory usage, while encoding weights on-demand reduces memory usage but introduces heavy computational cost.
To mitigate this overhead, we restrict weights to a small, fixed set of representative values.
Each representative is encoded once into a plaintext codebook, and during inference, the appropriate plaintext is retrieved rather than recomputed. 
Consequently, both runtime and memory scale with the codebook size rather than with the total number of weights. 
Clustering provides a mechanism for constructing such a representative set by grouping similar weights and replacing each with its nearest \textit{centroid}, thereby bounding the number of unique encodings.
Clustering is performed after training and applies directly to the models obtained in \S\ref{sec:model_approx_train}, independently of the architecture.

\begin{figure}[!t]
    \centering
    \includegraphics[width=0.95\linewidth]{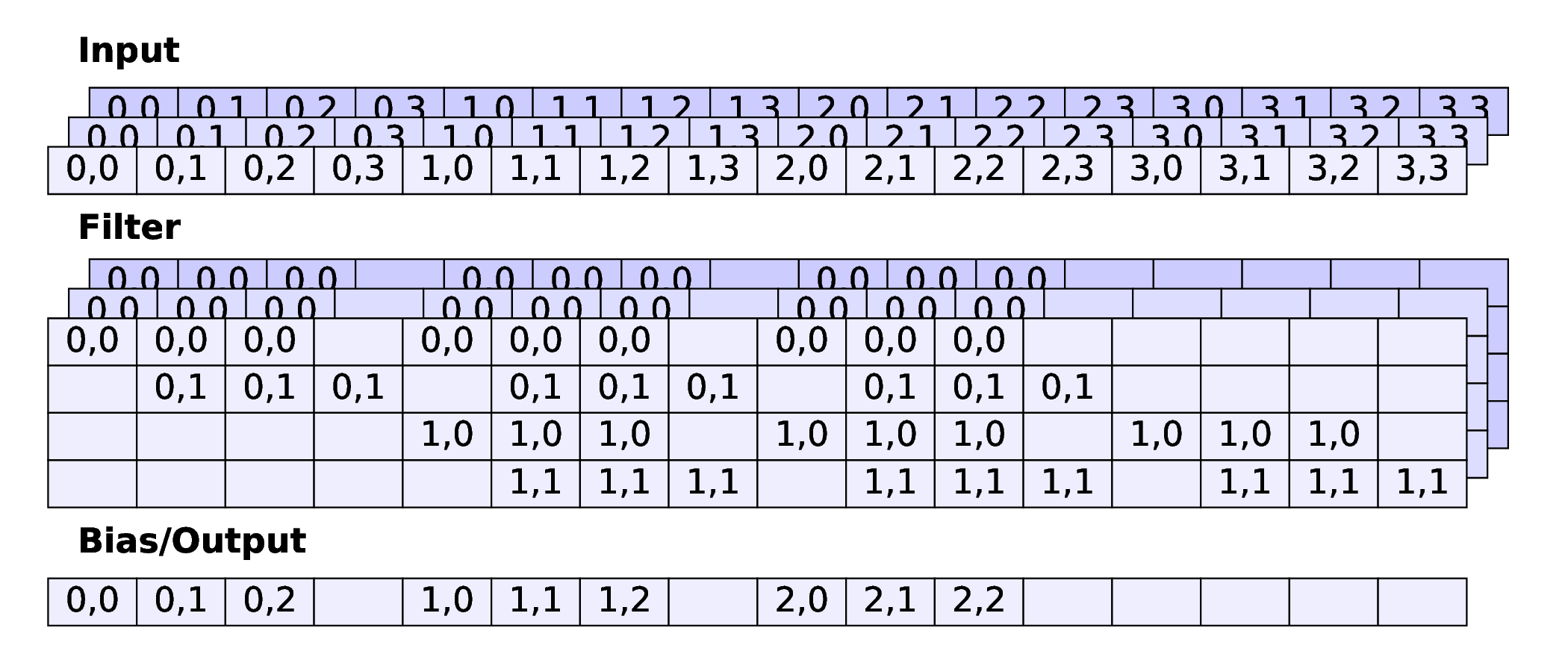}
    % \vspace{-2mm}
    \caption{\small Illustration of HW layout for convolution on a $3 \times 4 \times 4$ input and a $3 \times 2 \times 2$ filter, configured with padding 0 and stride 1.}
    \label{fig:hw_layout}
    % \vspace{-4pt}
\end{figure}

Formally, consider a polynomial network with $L$ convolution layers. 
The weight tensor of layer $l$ is $W^{(l)} \in \mathbb{R}^{O_l \times I_l \times H_l \times W_l}$, where $O_l$ and $I_l$ are the numbers of output and input channels, and $H_l \times W_l$ the kernel size. 
Flattening all convolution weights gives
$
\theta = (\theta_1, \dots, \theta_P) \in \mathbb{R}^P,
$
where $P$ is the total number of convolution parameters. 
To cluster the model, each $\theta_j$ is replaced by its nearest codebook value from $\mathcal{C} = \{c_1, \dots, c_k\} \subset \mathbb{R}$, where $k \ll P$ is the number of representatives in the codebook.
The codebook entries act as \textit{centroids} approximating the original weights.
The centroid count $k$ controls the efficiency-accuracy tradeoff, since a small $k$ reduces the number of encodings but increases approximation error, while a large $k$ approximates the weights more finely at higher computational and memory cost.
Given a non-negative distance function $d: \mathbb{R} \times \mathbb{R} \to \mathbb{R}$, 
each parameter is quantized as:
$
\tilde{\theta}_j = \arg\min_{c \in \mathcal{C}} d(\theta_j, c) \text{ for all } j = 1, \dots, P,
$
producing a quantized parameter vector $\tilde{\theta} \in \mathcal{C}^P$. 
Quantization is applied elementwise, so each convolutional kernel retains its shape while its values are drawn from the codebook.
The choice of clustering strategy determines how the codebook is constructed and how assignments are performed, which we examine next under two granularities.
% \vspace{2mm}

\vspace{1mm}
\noindent\textbf{\underline{Full Clustering.}}
A straightforward strategy is to cluster the entire parameter vector $\theta \in \mathbb{R}^P$ using a single global codebook $\mathcal{C}$.
All convolution parameters, regardless of layer or kernel position, are quantized to the same shared set of representative values under a unified mapping.
This produces a uniform quantization scheme with simple implementation and a small codebook size.

\vspace{1mm}
\noindent\textit{Limitation.}
While full clustering reduces the number of distinct weight values, it does not eliminate redundant plaintext encodings. 
Each weight must still match the ciphertext layout and level of its layer, so the same scalar values appearing in different layers must be re-encoded.
Even within a single convolution layer, although all weights share the same level, weights in different filter columns map to different layouts and thus require separate encodings. 
Only weights in the same column across layer filters align in layout and level. 
We refer to such a group as a \textit{slice}, which can share plaintext encoding. 
Full clustering therefore bounds the codebook size but not the number of encodings, as its centroids are spread across slices that cannot share plaintexts.
% \vspace{2mm}

\vspace{1mm}
\sloppy \noindent\textbf{\underline{Slice Clustering.}}
To address the redundancy remaining under full clustering, we refine quantization at the slice level.
Since only weights within the same slice can share plaintext encodings, clustering is applied independently to each slice rather than across all parameters. 
For a given layer $l$ and spatial position $s \in \{1, \dots, W_l\}$, we define the slice
$
S^{(l)}_s = \{ W^{(l)}_{o, i, h, s} \mid o = 1, \dots, O_l;\ i = 1, \dots, I_l;\ h = 1, \dots, H_l \},
$
which can be viewed as a vector in $\mathbb{R}^{O_l I_l H_l}$. 
Each slice $S^{(l)}_s$ is clustered independently using its own codebook $\mathcal{C}^{(l)}_s \subset \mathbb{R}$, and quantization is applied elementwise as
$
\tilde{W}^{(l)}_{o, i, h, s} = \arg\min_{c \in \mathcal{C}^{(l)}_s} d(W^{(l)}_{o, i, h, s}, c).
$
This ensures that clustering respects the FHE structure while adapting to the local weight distribution of each slice. 
Figure~\ref{fig:slice_clustering} illustrates slice clustering for a single convolution layer.
Convolution filters are decomposed into vertical slices along the kernel width dimension $W_l$, giving $W_l$ slices per layer, three as depicted.
Each colored slice $S^{(l)}_s$ is clustered independently using its own codebook.
Each three-dimensional block represents a convolution filter associated with one output channel.
The three axes correspond to kernel height $H_l$ (vertical), kernel width $W_l$ (horizontal), and input channels $I_l$ (depth).
The collection of such filters forms the output channels $O_l$.
Slices are taken as vertical strips along the width dimension $W_l$, and the figure highlights them with distinct colors.
The blue strip corresponds to slice $S^{(l)}_1$, the yellow strip to slice $S^{(l)}_2$, and the pink strip to slice $S^{(l)}_3$. 
Each slice $S^{(l)}_s$ groups together all weights ${W^{(l)}_{o,i,h,s}}$ at a fixed width position $s$ across output channels $O_l$, input channels $I_l$, and kernel height $H_l$, and is clustered independently with its own codebook $\mathcal{C}^{(l)}_s$.
This illustration emphasizes that clustering is performed slice by slice, rather than across the entire filter volume, a decomposition we extend to ensembles in \S\ref{sec:model_ensemble}.

\begin{figure}[!t]
    \centering
    \includegraphics[width=0.76\linewidth]{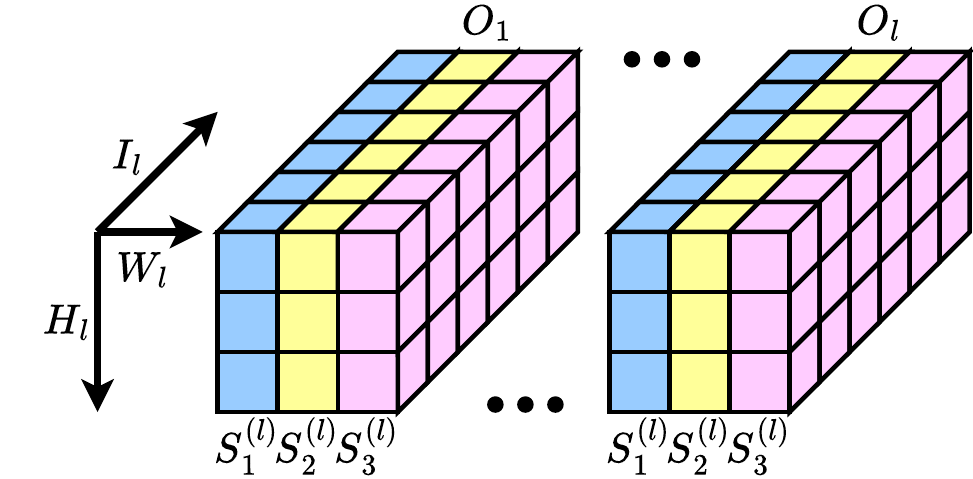}
    \caption{\small Illustration of single model slice clustering. Convolution filters are decomposed into slices along the kernel width, and each slice is clustered independently with its own codebook.}
    \label{fig:slice_clustering}
    % \vspace{-2mm}
    % \vspace{-4pt}
\end{figure}

\subsection{Ensemble of Polynomial Networks}\label{sec:model_ensemble}

Inference with a single polynomial network $g(x)$ often shows high variance across training runs.  
Let $g_{(m)}(x) \in \mathbb{R}^K$ denote the logits from the $m$-th independently trained instance. 
For a fixed input, these logits can differ significantly, especially near decision boundaries, due to (1) approximation error introduced by the fixed-point polynomial activation $p_d(\cdot)$ and (2) training stochasticity such as random weight initialization and mini-batch ordering. 
As a result, different $g_{(m)}$ instances may predict different classes for the same input.
To reduce this variance, we use an ensemble of $M$ polynomial networks $\{g_{(m)}\}_{m=1}^M$. 
All ensemble members share the same architecture and activation $p_d(\cdot)$, but are trained independently with different random seeds and mini-batch orders. 
Given an input $x$, the ensemble output is defined as the average of logits:
$
\overline{g}(x) = \frac{1}{M} \sum_{m=1}^M g_{(m)}(x)
$
with the predicted class
$
\hat{y} = \arg\max_{k \in \{1, \dots, K\}} \overline{g}(x)_k.
$
Logit averaging smooths both training noise and polynomial approximation errors, producing more stable and accurate predictions. 
Each $g_{(m)}$ is trained using the same regularized loss $\mathcal{L}_{\mathcal{B}}$ (Eq.~\ref{eq:penalty_loss}), ensuring compatibility with fixed-point polynomial inference under FHE. 
% Crucially, this ensemble design adds no additional computation or memory overhead.
Crucially, when sufficient unused ciphertext slots are available, this ensemble design adds no additional computation or memory overhead.
As described in \S\ref{sec:data_layout}, our HW layout leaves a subset of ciphertext slots unused. 
We populate these slots with weights from different ensemble members, while all models share the same ciphertexts for inputs and activations.
This reuse of ciphertext–plaintext structures avoids redundant ciphertexts and repeated encodings, making ensemble inference practical within the FHE framework.
% \vspace{2mm}

\noindent \textit{Limits of Ensemble with Clustering.}
While ensemble inference reuses unused ciphertext slots, naively combining it with parameter clustering introduces a challenge.
In slice clustering, each slice $S_s^{(l)}$ containing $O_l \times I_l \times H_l$ weights is represented by $k$ centroids in codebook $\mathcal{C}_s^{(l)}$, reducing the number of plaintext encodings to $k$ per slice.
In an ensemble setting, however, each plaintext must encode parameters from all $M$ models.
Since centroids are unlikely to align across independently trained models, the number of distinct encodings per slice grows as $\mathcal{O}(k^M)$, negating the benefits of clustering and making the naive combination impractical.
% \vspace{2mm}

\vspace{1mm}
\noindent\textbf{\underline{Slice Ensemble Clustering.}}
To address the inefficiency of independent clustering, we extend slice clustering to ensemble models by enforcing a shared set of representatives across all ensemble members. 
Specifically, we cluster weights jointly at the same kernel position so that all models use a common codebook.
Consider an ensemble of $M$ polynomial networks. 
For convolutional layer $l$, model $m$ has weights $W^{(l,m)} \in \mathbb{R}^{O_l \times I_l \times H_l \times W_l}$.
For each kernel width position $s \in \{1, \dots, W_l\}$, we extract slices:
$
S^{(l,m)}_s = \{ W^{(l,m)}_{o,i,h,s} \mid o = 1,\dots,O_l; i = 1,\dots,I_l; h = 1,\dots,H_l\}.
$
Stacking slices from all $M$ models produces
$
X^{(l)}_s = [ S^{(l,1)}_s \; S^{(l,2)}_s  \cdots  S^{(l,M)}_s ] \in \mathbb{R}^{N_s \times M},
$
where $N_s = O_l I_l H_l$. 
Each row of $X^{(l)}_s$ is an $M$-dimensional vector corresponding to the same weight coordinate across all models. 
We cluster these rows in $\mathbb{R}^M$ by solving:
$
\min_{\mathcal{C}^{(l)}_s \subset \mathbb{R}^M} 
\sum_{j=1}^{N_s} 
\min_{c \in \mathcal{C}^{(l)}_s} 
d( X^{(l)}_{s,j}, c),
$
where $\mathcal{C}^{(l)}_s = \{c_1, \dots, c_k\}$ is the shared codebook.
Each row $X^{(l)}_{s,j}$ is replaced by its nearest centroid, producing quantized slices 
$
\tilde{X}^{(l)}_{s,j} \in \mathcal{C}^{(l)}_s.
$
Quantized weights $\tilde{W}^{(l,m)}$ are then reconstructed by reshaping the corresponding slices. 
This shared clustering aligns weights across ensemble members, mapping small per-model variations at the same coordinate to a common centroid.
It thus requires fewer distinct encodings, reducing codebook size and inference time while retaining the accuracy gains of ensembling.
Figure~\ref{fig:slice_clustering_ensemble} illustrates how slice clustering is extended from a single model to an ensemble of polynomial networks.
\begin{figure}[!t]
    \centering
    \includegraphics[width=0.97\linewidth]{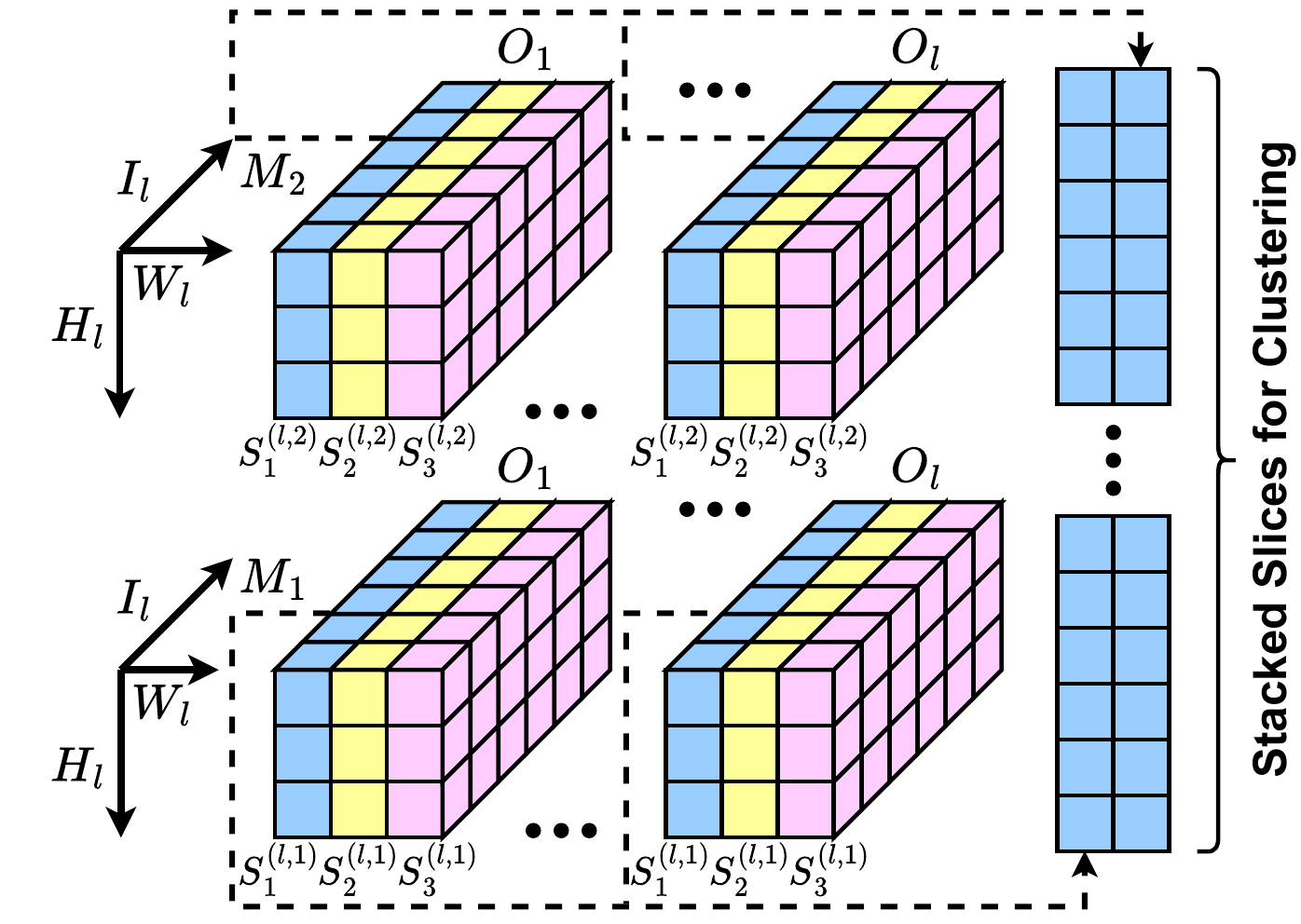}
    \vspace{-2mm}
    \caption{\small Illustration of slice-wise clustering for an ensemble. Slices at the same kernel width position from different models are stacked and clustered jointly using a shared codebook.}
    \label{fig:slice_clustering_ensemble}
    % \vspace{-4pt}
\end{figure}
Each three-dimensional block in the figure represents the convolution filters of a single model instance, spanning kernel height $H_l$, kernel width $W_l$, and input channels $I_l$.
The collection of such filters across output channels $O_l$ forms the layer weights $W^{(l,m)}$ for model $m$.
For a fixed layer $l$, slices are extracted as vertical strips along the width dimension $W_l$, with colors indicating different slices $S^{(l,m)}_s$ within each model.
Unlike the single-model case, however, we now align slices from multiple models at the same width position $s$ and stack them together.
This produces a slice matrix $X^{(l)}_s \in \mathbb{R}^{N_s \times M}$, where each row corresponds to one weight coordinate across all $M$ models, and $N_s = O_l I_l H_l$.
Clustering is then performed jointly in $\mathbb{R}^M$ across models, producing a shared codebook $\mathcal{C}^{(l)}_s$ for slice $s$.
The illustration emphasizes that ensemble clustering is not applied to each model independently.
Instead, slices are stacked across models at the same spatial position, and clustering is performed jointly.

\subsection{Overall Component Effects}
With the complete design now described, we summarize the contribution of each component before presenting the end-to-end evaluation.
Table~\ref{tab:component_ablation} reports the corresponding changes in multiplicative depth, accuracy, or latency.
Degree-2 training and structural optimizations reduce FHE depth, the HW layout and clustering balance slot availability against computation, and ensembling uses the available slots to restore accuracy.
This decomposition separates function-preserving transformations from deliberate accuracy--latency trade-offs.
\begin{table}[!t]
    \centering
    \caption{\small Component ablation showing the effect of each optimization. Level counts are for ResNet-18.}
    \label{tab:component_ablation}
    \small
    \vspace{-3mm}
    \setlength{\tabcolsep}{3pt}
    \renewcommand{\arraystretch}{1}
    \begin{tabular}{
        >{\raggedright\arraybackslash}p{0.4\columnwidth}
        p{0.56\columnwidth}
    }
        \toprule
        \textbf{Component} & \textbf{Effect} \\
        \midrule
        Degree-2 Activation and Training (\S3)
        & Reduces multiplicative depth from 87 to 70 levels, while regularized training limits the accuracy loss relative to ReLU. \\
        \addlinespace[3pt]
        Node Fusing (\S\ref{sec:node_fusing})
        & Uses exact identity transforms to remove batch-normalization operations, reducing the level count from 70 to 53. \\
        \addlinespace[3pt]
        Weight Redistribution (\S\ref{sec:weight_redistribution})
        & Applies exact graph-level transformations and provides the largest structural depth reduction, from 70 to 35 levels. \\
        \addlinespace[3pt]
        Tower Reuse (\S\ref{sec:tower_reuse})
        & Uses a functionally equivalent modulus-chain transformation to reduce the remaining level count from 35 to 18. \\
        \addlinespace[3pt]
        HW Data Layout (\S\ref{sec:data_layout})
        & Trades latency relative to CHW layout for unused ciphertext slots that support accuracy recovery through ensembling. \\
        \addlinespace[3pt]
        Parameter Clustering (\S\ref{sec:weight_clustering})
        & Reduces plaintext encodings during inference to offset the HW layout cost without degrading accuracy. \\
        \addlinespace[3pt]
        Ensemble (\S\ref{sec:model_ensemble})
        & Uses available HW-layout slots without added latency, recovering accuracy to match or exceed plaintext ReLU accuracy. \\
        \bottomrule
    \end{tabular}
    \vspace{-1mm}
\end{table}

\section{Experimental Results}\label{sec:results}

\subsection{Evaluation of Training Methodology}\label{sec:eval_training}

% \subsection{Training Setup}\label{sec:training_setup}
% \vspace{2mm}
\noindent\underline{\textbf{Training Setup.}}
% We trained ResNet-18, ResNet-20, and ResNet-32 on CIFAR-10 and CIFAR-100 using PyTorch~2.4.1.
% Each experiment was repeated 10 times, and we report mean accuracy and private inference time with standard deviations.
% \hl{...}
% For model training (see \S\ref{sec:model_approx_train}), we used hyperparameters $c = 2$ and $\zeta = 0.001$. 
% An ablation study for both is provided in \S\ref{sec:ablation}. 
% We set $b = 10$, aligned with PILLAR~\cite{pillar}, as it provides accurate polynomial approximation after quantization-aware fitting. 
% For all clustering strategies, we used the $\ell_2$ norm as the distance metric, with \textit{k-means} as the clustering algorithm~\cite{kmeans}.
We evaluate our polynomial approximation strategy on the CIFAR-10, CIFAR-100, and Tiny-ImageNet (referred to as Tiny) datasets, which are widely adopted benchmarks in prior work on privacy-preserving neural network inference. 
Our experiments use ResNet-18, ResNet-20, and ResNet-32 architectures, and additionally include VGG-16 to assess generality beyond residual networks. 
% All models are trained using PyTorch 2.4.1. 
Each experiment is repeated 10 times with different random seeds, and we report mean classification accuracy along with standard deviations to ensure statistical reliability of the observed performance. 
For model training (see~\S\ref{sec:model_approx_train}), we set the clipping parameter to $c = 2$ and the penalty strength to $\zeta = 0.001$ across all architectures and datasets without model-specific tuning. The impact of both hyperparameters is analyzed subsequently.
We fix the fixed-point precision to $b = 10$, following PILLAR~\cite{pillar}, as it provides accurate polynomial approximation under quantization-aware fitting. 
For all clustering strategies, we use the $\ell_2$ norm as the distance metric and employ \textit{k}-means clustering~\cite{kmeans}.

% \subsection{Ablation Study (this column)}\label{sec:ablation}
% \subsection{Evaluation of Training Methodology}

% Some introductory text about both ablation studies.\\
% We conduct an ablation study to evaluate the effect of \hl{...}

% \subsubsection{Ablation on Clipping Range (Value of $c$)}
\vspace{1mm}
\noindent\underline{\textbf{Impact of Clipping Parameter ($c$).}}
We evaluate the effect of the clipping interval $[-c, c]$ on model performance. 
The choice of $c$ determines the trade-off between preserving the activation distribution and truncating extreme values. 
If $c$ is too small, informative activations are excessively clipped, reducing representational capacity; conversely, if $c$ is too large, the clipping 
% mechanism 
provides little stabilization, allowing unstable activations to propagate.
Results are reported for ResNet-18 on CIFAR-10, which we present as a representative case. 
Figure~\ref{fig:relu_inputs_hist} shows the distribution of pre-activation inputs to ReLU across the network.
The values are largely concentrated within the interval $[-2, 2]$, with only a small fraction outside this range, supporting the choice of restricting pre-activations to $[-2, 2]$. 
Figure~\ref{fig:diff_clip} shows the classification accuracy obtained under different clipping intervals. 
Among the tested ranges, $[-2, 2]$ consistently achieves the highest mean accuracy across runs. 
Smaller intervals degraded performance by discarding a substantial portion of activations, while larger intervals reduce stability and result in lower accuracy with average accuracy decreasing from 93.50\% at $c = 2$ to 91.25\% at $c = 10$.
This trend is consistent across other dataset–architecture combinations. 
Based on these observations, we adopt $[-2, 2]$ as the clipping range in all experiments.

% \subsubsection{Ablation on Penalty Strength (Value of $\zeta$)}
\vspace{1mm}
\noindent\underline{\textbf{Impact of Penalty Strength ($\zeta$).}}
We analyze the impact of the penalty strength $\zeta$ on model performance. 
The parameter $\zeta$ controls the magnitude of the regularization term that penalizes pre-activations outside the target interval $[-c, c]$. 
Figure~\ref{fig:penalty_ablation} reports results ResNet-18 on CIFAR-10.
The penalty is the primary component enabling stable degree-2 training.
At smaller penalty values ($\zeta \leq 10^{-5}$), training collapses to near-chance average accuracy of 11--15\%, confirming that escaping activations destabilize training without sufficient regularization.
For penalty strengths $\zeta \geq 10^{-4}$, average accuracy varies by only 3.3\%.
The highest mean accuracy is obtained for $\zeta = 10^{-3}$.
Larger penalty values ($\zeta > 10^{-3}$) reduce performance by overly constraining the model.
This trend is consistent across other dataset–architecture combinations. 
Based on these observations, we set $\zeta = 10^{-3}$ in all experiments.

\begin{figure*}[!t]
    \centering
    \begin{subfigure}[t]{0.33\linewidth}
        \centering
        \includegraphics[width=\linewidth]{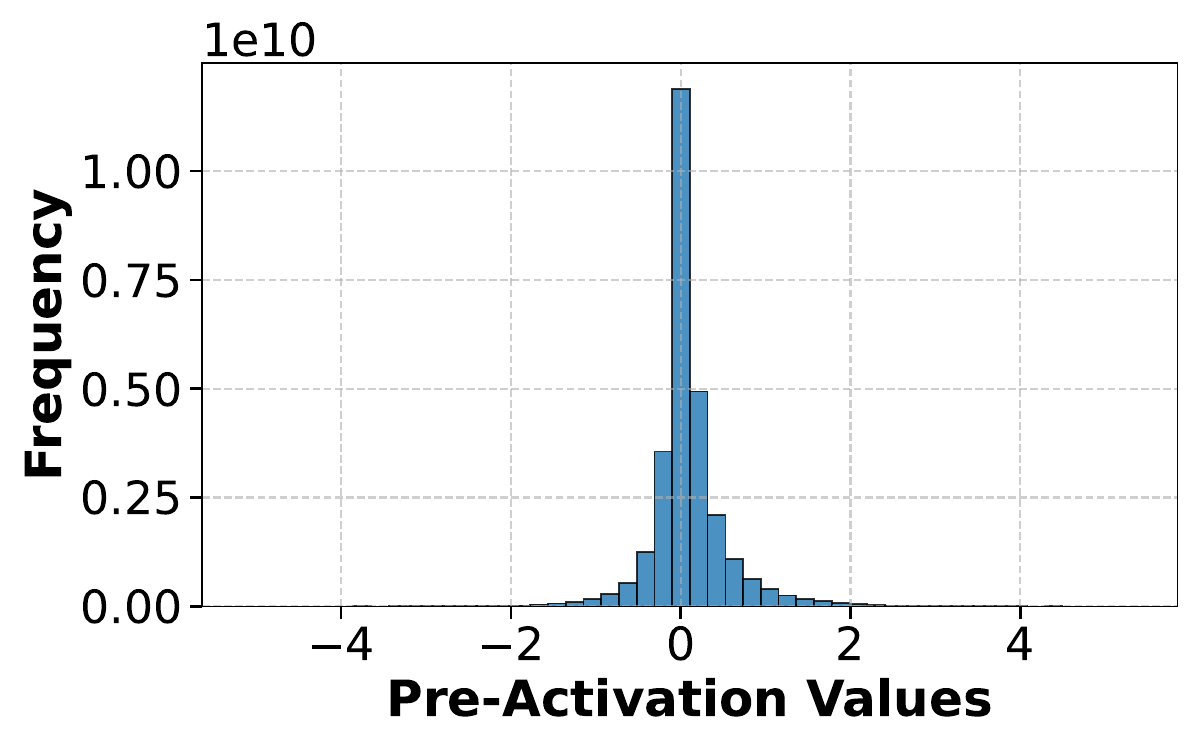}
        \caption{}
        \label{fig:relu_inputs_hist}
    \end{subfigure}
    \hfill
    \begin{subfigure}[t]{0.33\linewidth}
        \centering
        \includegraphics[width=\linewidth]{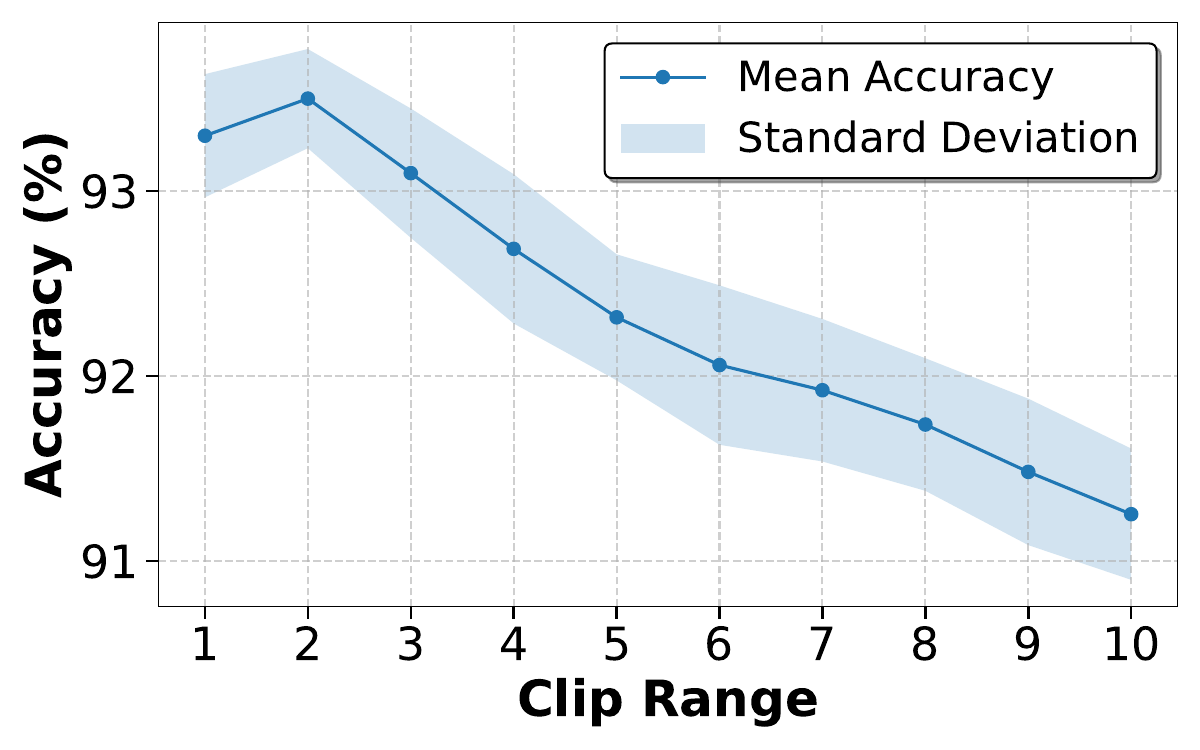}
        \caption{}
        \label{fig:diff_clip}
    \end{subfigure}
    \hfill
    \begin{subfigure}[t]{0.33\linewidth}
        \includegraphics[width=\linewidth]{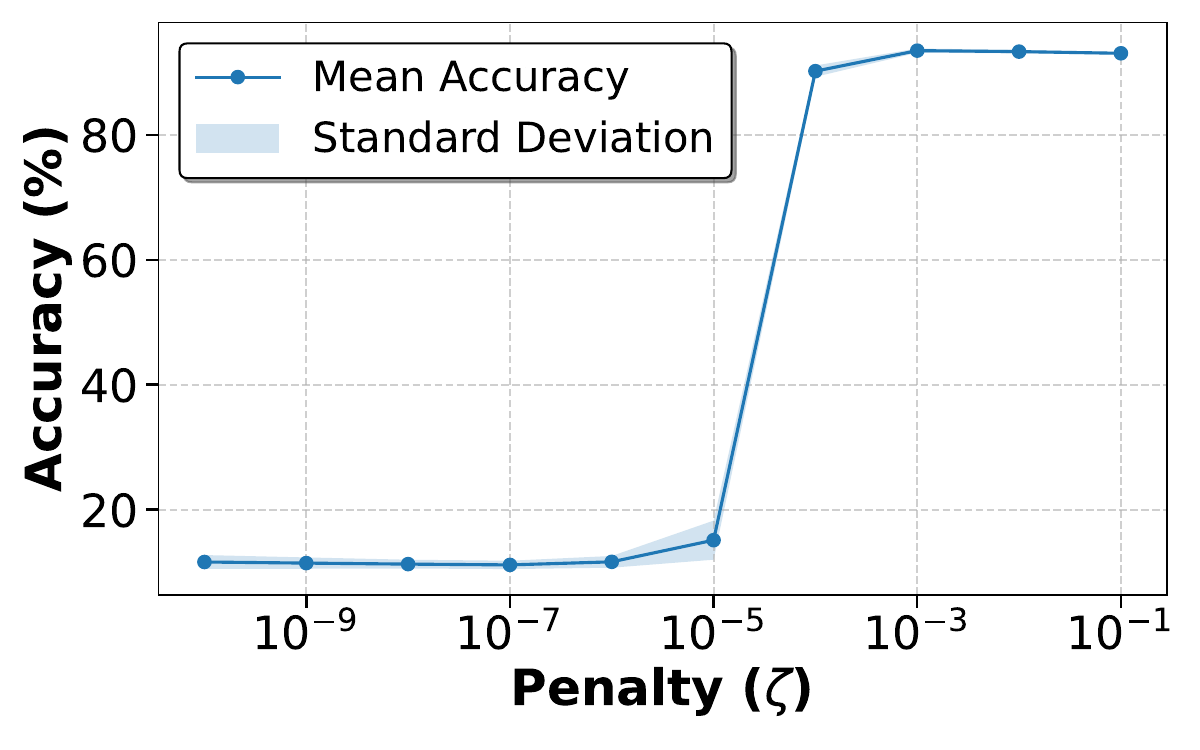}
        \caption{}
        \label{fig:penalty_ablation}
    \end{subfigure}
    \vspace{-3mm}
    % \caption{\small 
    % Results of the clipping range ablation study on CIFAR-10 with ResNet-18. 
    % (a) Distribution of pre-activation inputs to ReLU across the entire network, showing concentration within $[-2, 2]$. 
    % (b) Test accuracy under different clipping intervals, where $[-2, 2]$ produces the best performance.
    % (c) Results of the penalty strength ablation study on CIFAR-10 with ResNet-18, showing test accuracy across different values of $\zeta$. The best performance is achieved at $\zeta = 10^{-3}$; smaller values fail to regularize effectively, while larger values start to decrease accuracy. \hl{FIX}}
    \caption{\small 
    Effects of activation clipping and regularization for ResNet-18 on CIFAR-10. 
    (a)~Pre-Activation value distribution across all layers, showing concentration within [-2, 2]. 
    (b)~Accuracy across clipping ranges, with [-2, 2] performing best. 
    (c)~Accuracy versus penalty strength $\zeta$, with best performance at $\zeta = 10^{-3}$; smaller values under-regularize the model, while larger values overly constrain it and reduce accuracy.
    }
    \label{fig:ablation}
\end{figure*}

\vspace{1mm}
\noindent\underline{\textbf{Impact of Quantization-Aware Polynomial Fitting.}}
We analyze the impact of quantization-aware polynomial fitting for ResNet-18 on CIFAR-10 by comparing it with an ordinary least-squares fit with coefficient rounding.
The ordinary and quantization-aware fits achieve average accuracies of 93.73\% and 93.94\%, respectively.
The 0.21\% average accuracy improvement shows that quantization-aware fitting is primarily a precision safeguard rather than a source of major accuracy gains.
Degree-2 activations have only three coefficients, so truncation risk is minimal in our setting.

% \begin{figure}[!t]
%     \centering
%     \begin{subfigure}[t]{0.495\linewidth}
%         \centering
%         \includegraphics[width=\linewidth]{figures/relu_inputs_hist.pdf}
%         \caption{}
%         \label{fig:relu_inputs_hist}
%     \end{subfigure}
%     \hfill
%     \begin{subfigure}[t]{0.495\linewidth}
%         \centering
%         \includegraphics[width=\linewidth]{figures/clip_ablation.pdf}
%         \caption{}
%         \label{fig:diff_clip}
%     \end{subfigure}
%     \caption{Results of the clipping range ablation study on CIFAR-10 with ResNet-18. (a) Distribution of pre-activation inputs to ReLU across the entire network, showing concentration within $[-2, 2]$. (b) Test accuracy under different clipping intervals, where $[-2, 2]$ produces the best performance.}
%     \label{fig:ablation_clip}
% \end{figure}

% \begin{figure}[!t]
%     \centering
%     \includegraphics[width=0.495\linewidth]{figures/penalty_ablation.pdf}
%     \caption{Results of the penalty strength ablation study on CIFAR-10 with ResNet-18, showing test accuracy across different values of $\zeta$. The best performance is achieved at $\zeta = 10^{-3}$; smaller values fail to regularize effectively, while larger values start to decrease accuracy.}
%     \label{fig:penalty_ablation}
% \end{figure}

% \subsection{Comparison with PILLAR}\label{sec:pillar_result}
% \subsection{Evaluation of Training Methodology}\label{sec:pillar_result}
\vspace{1mm}
\noindent\underline{\textbf{Comparison to Related Work.}}
We compare the training methodology of our approach against PILLAR, which represents the lowest-degree polynomial-approximation method previously shown to train deep neural networks reliably.
We evaluate PILLAR using its open-source implementation~\cite{github_pillar}. 
PILLAR was originally designed in the context of interactive PPML, which differs from our non-interactive setting, and its repository does not provide any implementation for private inference.
For these reasons, we restrict the comparison to classification accuracy and do not evaluate inference time.
Table~\ref{tab:pillar-compare} reports classification accuracies (mean $\pm$ standard deviation) across all evaluated datasets and architectures
% , including} VGG-16, ResNet-18, ResNet-20, and ResNet-32 on CIFAR-10, CIFAR-100, and TinyImageNet 
using ReLU, PAPER (our method), and two versions of PILLAR.
The original PILLAR
% , available through its open-source implementation~\cite{github_pillar}, 
method employs a degree-4 polynomial activation.
For completeness, we additionally adapt this implementation to use a degree-2 polynomial 
while keeping the original training procedure unchanged. 
PILLAR was originally designed for degree-4 polynomials and is not optimized for the degree-2 setting. 
These results are therefore included only as a diagnostic to highlight the instability of low-degree polynomial training, rather than for direct comparison.
As shown in Table~\ref{tab:pillar-compare}, when averaged across all datasets and architectures, the adapted degree-2 variant of PILLAR exhibits a substantial degradation in accuracy, with an average drop of 7.99\% relative to ReLU and notably high variance, particularly on more challenging datasets such as CIFAR-100 and Tiny-ImageNet.
The original degree-4 variant of PILLAR is more stable but still underperforms ReLU, with an average accuracy gap of 6.54\%.
In contrast, PAPER with degree-2 polynomial activation remains much closer to ReLU, incurring an average accuracy drop of only 1.84\% across all datasets and architectures.
Importantly, PAPER 
% demonstrates 
shows
consistently stable training 
% behavior 
across all evaluated settings, avoiding the large run-to-run fluctuations seen in both variants of PILLAR.
These results demonstrate that while PILLAR either degrades significantly when reduced to a degree-2 configuration or remains less accurate in its intended degree-4 configuration, PAPER enables accurate and stable training with quadratic polynomial activations.
Overall, these results indicate that the accuracy gap typically associated with low-degree polynomial approximations arises primarily from limitations in training methodology, and can be substantially mitigated through the regularization techniques introduced in PAPER.

\begin{table}[!t]
    \centering
    \caption{\small Classification accuracy (\%, mean $\pm$ standard deviation) on CIFAR and Tiny-ImageNet datasets for VGG and ResNet models using ReLU, PAPER (degree-2), and two versions of PILLAR.}
    \label{tab:pillar-compare}
    \vspace{-2mm}
    \resizebox{\linewidth}{!}{
    \begin{tabular}{c|c|c|c|c}
        \hline
        \multirow{2}{*}{\textbf{Method}}  & \textbf{VGG-16} & \textbf{ResNet-18} & \textbf{ResNet-20} & \textbf{ResNet-32} \\
        \cline{2-5}
         & \multicolumn{4}{c}{\textbf{CIFAR-10}} \\
        \hline
        ReLU 
            & 92.57 $\pm$ 0.36
            & 94.57 $\pm$ 0.23 
            & 94.84 $\pm$ 0.25 
            & 95.23 $\pm$ 0.30 \\
        \textbf{PAPER}
            & 90.41 $\pm$ 0.46
            & 93.73 $\pm$ 0.27 
            & 94.06 $\pm$ 0.26 
            & 93.86 $\pm$ 0.24 \\
        PILLAR ($d=2$)
            & 88.32 $\pm$ 0.19
            & 91.78 $\pm$ 0.13 
            & 90.43 $\pm$ 0.40 
            & 90.40 $\pm$ 1.73 \\
        PILLAR ($d=4$)
            & 90.76 $\pm$ 0.18
            & 93.26 $\pm$ 0.15 
            & 90.98 $\pm$ 0.29 
            & 92.24 $\pm$ 0.27 \\
        \hline
        \hline
        \multirow{2}{*}{\textbf{Method}}  & \textbf{ResNet-18} & \textbf{ResNet-20} & \textbf{ResNet-32} & \textbf{ResNet-32} \\
        \cline{2-5}
        & \multicolumn{3}{c|}{\textbf{CIFAR-100}} & \textbf{Tiny}\\
        \hline
        ReLU 
            & 76.22 $\pm$ 0.44 
            & 74.96 $\pm$ 0.36 
            & 76.58 $\pm$ 0.40
            & 61.37 $\pm$ 0.36 \\
        \textbf{PAPER}
            & 74.96 $\pm$ 0.51 
            & 73.89 $\pm$ 0.37 
            & 73.93 $\pm$ 0.53
            & 56.78 $\pm$ 0.42 \\
        PILLAR ($d=2$) 
            & 71.92 $\pm$ 0.28 
            & 62.30 $\pm$ 3.96 
            & 66.38 $\pm$ 1.78
            & 40.85 $\pm$ 0.68\\
        PILLAR ($d=4$) 
            & 74.81 $\pm$ 0.22 
            & 62.55 $\pm$ 7.51 
            & 66.10 $\pm$ 1.12
            & 43.29 $\pm$ 0.83\\
        \hline
    \end{tabular}}
    
    \vspace{1mm}
    \footnotesize
    \parbox{\linewidth}
    {\footnotesize
        \textit{Note:} The original PILLAR method uses a degree-4 polynomial~\cite{pillar}; the degree-2 variant is our adaptation of its open-source code for comparison.
    }
    % \textit{Note:} The original PILLAR method uses a degree-4 polynomial~\cite{pillar}; the degree-2 variant is our adaptation of its open-source code for comparison.
    % \vspace{-8mm} <-- this is useless
\end{table}

\begin{table*}[!t]
    \centering
    \caption
    {
        \small
        {CKKS parameter sets used to evaluate the proposed methods. $N$ is the polynomial degree, $\Delta$ is the scaling factor, $\mathcal{L}$ is the number of computation levels, $Q$ is the computation modulus chain, $P$ is the special prime used for hybrid key switching, $\log_{2}(PQ)$ is the total coefficient modulus used for RLWE security estimation, and the security level (bits) is computed using the LWE Estimator \cite{lwe_estimator}.}
    }
    \vspace{-2mm}
    \begin{tabular}
    {
        >{\hspace{0pt}}c<{\hspace{0pt}} | 
        >{\hspace{0pt}}c<{\hspace{0pt}} | 
        >{\hspace{0pt}}c<{\hspace{0pt}} | 
        >{\hspace{0pt}}c<{\hspace{0pt}} | 
        >{\hspace{0pt}}c<{\hspace{0pt}} | 
        >{\hspace{0pt}}c<{\hspace{0pt}} | 
        >{\hspace{0pt}}c<{\hspace{0pt}} | 
        >{\hspace{0pt}}c<{\hspace{0pt}} | 
        >{\hspace{0pt}}c<{\hspace{0pt}}
    }
        \hline
        {\bf Model} &
        {\bf $N$} &
        {\bf $\Delta$} &
        {\bf $\mathcal{L}$} &
        {\bf Modulus chain $Q$} &
        {\bf $\log_2 Q$} &
        {\bf $\log_2 P$} &
        {\bf $\log_2 (PQ)$} &
        {\bf Security level} \\ \hline

        {VGG-16   } & {$2^{15}$} & {$2^{25}$} & {16} & {$6 \times 50 + 10 \times 49 + 26$} & { 816} & {54} & { 870} & {129.7} \\
        {ResNet-18} & {$2^{15}$} & {$2^{22}$} & {18} & {$18 \times 44 + 23$              } & { 815} & {54} & { 869} & {130.0} \\
        {ResNet-20} & {$2^{15}$} & {$2^{21}$} & {20} & {$20 \times 42 + 22$              } & { 862} & {44} & { 906} & {124.4} \\
        {ResNet-32} & {$2^{16}$} & {$2^{26}$} & {32} & {$32 \times 52 + 27$              } & {1691} & {54} & {1745} & {130.4} \\
        \hline
    \end{tabular}
    \label{tab:fhe_parameters}
\end{table*}

% \vspace{1mm}
\subsection{Evaluation of Private Inference}\label{sec:private_inference}

\noindent\underline{\textbf{Private Inference Setup.}}
Private inference uses a C++ framework developed in-house.
The framework includes a tool to automatically convert trained models from PyTorch for C++ inference.
Our goal is to compare privacy-preserving CNN architectures.
Thus, for a fair comparison with state-of-the-art models, whose implementations are designed for CPU execution~\cite{github_mpcnn, github_autofhe}, we linked our framework to Microsoft SEAL 4.1~\cite{seal}, a CPU-based library for leveled FHE operations, and GMP 6.2.1 for multiprecision arithmetic.
All experiments, including related work, were run on a machine with two AMD EPYC 64-core processors, 2 TB memory across 32 DDR4 DIMMs at 2933 MT/s, on Ubuntu 22.04.5 LTS.
% Experiments include VGG-16, ResNet-18, ResNet-20, and ResNet-32 models on CIFAR and TinyImageNet datasets.
For PAPER, we set the polynomial degree $N = 2^{16}$ for ResNet-32 and $N = 2^{15}$ for the other models, as these are the smallest polynomial degrees enabling correct inference.
Table~\ref{tab:fhe_parameters} reports the CKKS parameter sets used to evaluate the proposed methods.
The security level is computed using the LWE estimator \cite{lwe_estimator} under the BKZ.sieve cost model, considering a uSVP-based lattice attack and a ternary distribution with support $\{-1,0,1\}$, in accordance with the Homomorphic Encryption Standard \cite{he_standard}.
Additional details on FHE encryption parameters are in Appendix \S\ref{sec:fhe_additional_details}. 
We compare our approach against state-of-the-art non-interactive PPML baselines, MPCNN~\cite{mpcnn} and AutoFHE~\cite{autofhe}, in terms of accuracy and inference time, using their publicly available implementations~\cite{github_autofhe, github_mpcnn} as is.
Both baselines implement ResNet-20 on CIFAR-10 and ResNet-32 on CIFAR-10 and CIFAR-100, all with $N = 2^{16}$.
We also evaluate plaintext ensembles with standard ReLU activations for $M=\{1,2\}$.
Their accuracies provide ceilings for the corresponding polynomial-activation ensembles before replacing ReLU with FHE-compatible activations. 
Because they operate on unencrypted inputs and incur no FHE inference cost, they are not privacy-preserving baselines.

Figure~\ref{fig:fullpage_results} shows private inference time versus accuracy across several CNN models and datasets, comparing clustering methods, ensembles, and baselines.
The figure is organized into 8 groups corresponding to dataset–architecture pairs, with three horizontally aligned plots per group.
The left column compares clustering strategies, the center column analyzes the effect of ensemble size on accuracy, and the right column presents the Pareto front of our approach alongside related work and plaintext ReLU baseline.
Additionally, we present a summary with the main results in Table~\ref{tab:summary}, highlighting the trade-off between accuracy and inference time.

\begin{table*}
    \centering
    \small
    \caption
    {
        PAPER's main results using \textit{Slice Clustering} compared to plaintext ReLU and privacy-preserving related work in terms of accuracy and inference time for CIFAR and Tiny-ImageNet datasets on VGG and ResNet models. \texttt{Acc.}: accuracy (\%), \texttt{Time}: inference time (s), \texttt{Boot.}: \# of bootstrapping layers, \texttt{M}: \# of models in the ensemble, and \texttt{k}: centroid count in clustering.
    }
    % \vspace{-2mm}
    \renewcommand{\arraystretch}{0.99}
    \setlength{\tabcolsep}{3pt}
    \begin{tabular}
    {
        >{\centering\arraybackslash}m{2.1cm} |
        >{\centering\arraybackslash}m{1.5cm} |
        >{\centering\arraybackslash}m{0.8cm}
        >{\centering\arraybackslash}m{0.5cm} |
        >{\centering\arraybackslash}m{0.8cm}
        >{\centering\arraybackslash}m{0.8cm}
        >{\centering\arraybackslash}m{0.8cm} |
        >{\centering\arraybackslash}m{0.8cm}
        >{\centering\arraybackslash}m{0.8cm}
        >{\centering\arraybackslash}m{0.8cm} |
        >{\centering\arraybackslash}m{0.8cm}
        >{\centering\arraybackslash}m{0.8cm}
        >{\centering\arraybackslash}m{0.8cm}
        >{\centering\arraybackslash}m{0.8cm}
        >{\centering\arraybackslash}m{0.8cm}
    }
        \hline
        \multirow{2}{*}{Dataset}      & \multirow{2}{*}{Model}    & \multicolumn{2}{c|}{\textbf{ReLU}} & \multicolumn{3}{c|}{\textbf{MPCNN}~\cite{mpcnn}{\textsuperscript{$\dagger$}}} & \multicolumn{3}{c|}{\textbf{AutoFHE}~\cite{autofhe}{\textsuperscript{$\dagger$}}} & \multicolumn{5}{c}{\textbf{PAPER}}                  \\ \cline{3-15}
         &  & Acc. & {M} & Acc. & Time & Boot.      & Acc. & Time & Boot.          & Acc. & Time & Boot. & M & k \\
        \hline

        \multirow{16}*{CIFAR-10} &
        \multirow{4}*{VGG-16} &
        % ReLU
        \multirow[c]{4}{*}{\begin{tabular}{@{}c@{}}92.57 \\ {93.03}\end{tabular}} &
        \multirow[c]{4}{*}{\begin{tabular}{@{}c@{}}{1} \\ {2}\end{tabular}} &
        % MPCNN
        \multicolumn{3}{c|}{\multirow{4}*{not available}} &
        % AutoFHE
        \multicolumn{3}{c|}{\multirow{4}*{not available}} &
        % PAPER
        90.26 & 966 & 0 & 1 & 64 \\
        % Dataset, Model, ReLU
        {} & {} & {} & {} &
        % MPCNN
        {} & {} & {} &
        % AutoFHE
        {} & {} & {} &
        % PAPER
        90.46 & 1293 & 0 & 2 & 256 \\
        % Dataset, Model, ReLU
        {} & {} & {} & {} &
        % MPCNN
        {} & {} & {} &
        % AutoFHE
        {} & {} & {} &
        % PAPER
        91.19 & 1510 & 0 & 2 & 512 \\
        % Dataset, Model, ReLU
        {} & {} & {} & {} &
        % MPCNN
        {} & {} & {} &
        % AutoFHE
        {} & {} & {} &
        % PAPER
        91.42 & 1881 & 0 & 2 & 1024 \\
        \cline{2-15}

        {} &
        \multirow{4}*{ResNet-18} &
        % ReLU
        \multirow[c]{4}{*}{\begin{tabular}{@{}c@{}}94.57 \\ {95.29}\end{tabular}} &
        \multirow[c]{4}{*}{\begin{tabular}{@{}c@{}}{1} \\ {2}\end{tabular}} &
        % MPCNN
        \multicolumn{3}{c|}{\multirow{4}*{not available}} &
        % AutoFHE
        \multicolumn{3}{c|}{\multirow{4}*{not available}} &
        % PAPER
        93.62 & 504 & 0 & 1 & 64 \\
        % Dataset, Model, ReLU
        {} & {} & {} & {} &
        % MPCNN
        {} & {} & {} &
        % AutoFHE
        {} & {} & {} &
        % PAPER
        94.35 & 695 & 0 & 2 & 256 \\
        % Dataset, Model, ReLU
        {} & {} & {} & {} &
        % MPCNN
        {} & {} & {} &
        % AutoFHE
        {} & {} & {} &
        % PAPER
        94.50 & 854 & 0 & 2 & 512 \\
        % Dataset, Model, ReLU
        {} & {} & {} & {} &
        % MPCNN
        {} & {} & {} &
        % AutoFHE
        {} & {} & {} &
        % PAPER
        94.73 & 1164 & 0 & 2 & 2048 \\
        \cline{2-15}

        {} &
        \multirow{4}*{ResNet-20} &
        % ReLU
        \multirow[c]{4}{*}{\begin{tabular}{@{}c@{}}94.84 \\ {95.57}\end{tabular}} &
        \multirow[c]{4}{*}{\begin{tabular}{@{}c@{}}{1} \\ {2}\end{tabular}} &
        % MPCNN
        \multirow{4}*{92.59} & \multirow{4}*{2748} & \multirow{4}*{18} &
        % AutoFHE
        \multirow[c]{4}{*}{\begin{tabular}{@{}c@{}}91.96 \\ 92.89\end{tabular}} &
        \multirow[c]{4}{*}{\begin{tabular}{@{}c@{}}1173 \\ 2147\end{tabular}} &
        \multirow[c]{4}{*}{\begin{tabular}{@{}c@{}}5 \\ 11\end{tabular}} &
        % PAPER
        94.01 & 314 & 0 & 1 & 64 \\
        % Dataset, Model, ReLU
        {} & {} & {} & {} &
        % MPCNN
        {} & {} & {} &
        % AutoFHE
        {} & {} & {} &
        % PAPER
        94.34 & 370 & 0 & 2 & 128 \\
        % Dataset, Model, ReLU
        {} & {} & {} & {} &
        % MPCNN
        {} & {} & {} &
        % AutoFHE
        {} & {} & {} &
        % PAPER
        94.81 & 443 & 0 & 2 & 256 \\
        % Dataset, Model, ReLU
        {} & {} & {} & {} &
        % MPCNN
        {} & {} & {} &
        % AutoFHE
        {} & {} & {} &
        % PAPER
        95.08 & 611 & 0 & 2 & 1024 \\
        \cline{2-15}

        {} &
        \multirow{4}*{ResNet-32} &
        % ReLU
        \multirow[c]{4}{*}{\begin{tabular}{@{}c@{}}95.23 \\ {95.59}\end{tabular}} &
        \multirow[c]{4}{*}{\begin{tabular}{@{}c@{}}{1} \\ {2}\end{tabular}} &
        % MPCNN
        \multirow{4}*{93.39} & \multirow{4}*{4392} & \multirow{4}*{30} &
        % AutoFHE
        \multirow[c]{4}{*}{\begin{tabular}{@{}c@{}}93.16 \\ 93.56\end{tabular}} &
        \multirow[c]{4}{*}{\begin{tabular}{@{}c@{}}1939 \\ 3114\end{tabular}} &
        \multirow[c]{4}{*}{\begin{tabular}{@{}c@{}}8 \\ 19\end{tabular}} &
        % PAPER
        93.75 & 1833 & 0 & 1 & 64 \\
        % Dataset, Model, ReLU
        {} & {} & {} & {} &
        % MPCNN
        {} & {} & {} &
        % AutoFHE
        {} & {} & {} &
        % PAPER
        94.06 & 2388 & 0 & 2 & 256 \\
        % Dataset, Model, ReLU
        {} & {} & {} & {} &
        % MPCNN
        {} & {} & {} &
        % AutoFHE
        {} & {} & {} &
        % PAPER
        94.59 & 2788 & 0 & 2 & 512 \\
        % Dataset, Model, ReLU
        {} & {} & {} & {} &
        % MPCNN
        {} & {} & {} &
        % AutoFHE
        {} & {} & {} &
        % PAPER
        94.76 & 3093 & 0 & 2 & 1024 \\
        \hline

        \multirow{12}*{CIFAR-100} &
        \multirow{4}*{ResNet-18} &
        % ReLU
        \multirow[c]{4}{*}{\begin{tabular}{@{}c@{}}76.22 \\ {78.20}\end{tabular}} &
        \multirow[c]{4}{*}{\begin{tabular}{@{}c@{}}{1} \\ {2}\end{tabular}} &
        % MPCNN
        \multicolumn{3}{c|}{\multirow{4}*{not available}} &
        % AutoFHE
        \multicolumn{3}{c|}{\multirow{4}*{not available}} &
        % PAPER
        74.31 & 520 & 0 & 1 & 64 \\
        % Dataset, Model, ReLU
        {} & {} & {} & {} &
        % MPCNN
        {} & {} & {} &
        % AutoFHE
        {} & {} & {} &
        % PAPER
        75.79 & 734 & 0 & 2 & 256 \\
        % Dataset, Model, ReLU
        {} & {} & {} & {} &
        % MPCNN
        {} & {} & {} &
        % AutoFHE
        {} & {} & {} &
        % PAPER
        76.35 & 871 & 0 & 2 & 512 \\
        % Dataset, Model, ReLU
        {} & {} & {} & {} &
        % MPCNN
        {} & {} & {} &
        % AutoFHE
        {} & {} & {} &
        % PAPER
        76.82 & 1202 & 0 & 2 & 2048 \\
        \cline{2-15}

        {} &
        \multirow{4}*{ResNet-20} &
        % ReLU
        \multirow[c]{4}{*}{\begin{tabular}{@{}c@{}}74.96 \\ {77.80}\end{tabular}} &
        \multirow[c]{4}{*}{\begin{tabular}{@{}c@{}}{1} \\ {2}\end{tabular}} &
        % MPCNN
        \multicolumn{3}{c|}{\multirow{4}*{not available}} &
        % AutoFHE
        \multicolumn{3}{c|}{\multirow{4}*{not available}} &
        % PAPER
        73.68 & 323 & 0 & 1 & 64 \\
        % Dataset, Model, ReLU
        {} & {} & {} & {} &
        % MPCNN
        {} & {} & {} &
        % AutoFHE
        {} & {} & {} &
        % PAPER
        74.31 & 380 & 0 & 2 & 128 \\
        % Dataset, Model, ReLU
        {} & {} & {} & {} &
        % MPCNN
        {} & {} & {} &
        % AutoFHE
        {} & {} & {} &
        % PAPER
        76.08 & 464 & 0 & 2 & 256 \\
        % Dataset, Model, ReLU
        {} & {} & {} & {} &
        % MPCNN
        {} & {} & {} &
        % AutoFHE
        {} & {} & {} &
        % PAPER
        77.16 & 690 & 0 & 2 & 2048 \\
        \cline{2-15}

        {} &
        \multirow{4}*{ResNet-32} &
        % ReLU
        \multirow[c]{4}{*}{\begin{tabular}{@{}c@{}}76.58 \\ {78.43}\end{tabular}} &
        \multirow[c]{4}{*}{\begin{tabular}{@{}c@{}}{1} \\ {2}\end{tabular}} &
        % MPCNN
        % {} & {} & {} &
        \multirow{4}*{70.59} & \multirow{4}*{4316} & \multirow{4}*{30} &
        % AutoFHE
        \multirow{4}*{71.34} & \multirow{4}*{3100} & \multirow{4}*{16} &
        % PAPER
        73.49 & 1803 & 0 & 1 & 64 \\
        % Dataset, Model, ReLU
        {} & {} & {} & {} &
        % MPCNN
        {} & {} & {} &
        % AutoFHE
        {} & {} & {} &
        % PAPER
        74.51 & 2340 & 0 & 2 & 256 \\
        % Dataset, Model, ReLU
        {} & {} & {} & {} &
        % MPCNN
        {} & {} & {} &
        % AutoFHE
        {} & {} & {} &
        % PAPER
        76.30 & 2809 & 0 & 2 & 512 \\
        % Dataset, Model, ReLU
        {} & {} & {} & {} &
        % MPCNN
        {} & {} & {} &
        % AutoFHE
        {} & {} & {} &
        % PAPER
        77.24 & 3797 & 0 & 2 & 2048 \\
        \hline

        \multirow{4}*{Tiny-ImageNet} &
        \multirow{4}*{ResNet-32} &
        % ReLU
        \multirow[c]{4}{*}{\begin{tabular}{@{}c@{}}61.37 \\ {65.33}\end{tabular}} &
        \multirow[c]{4}{*}{\begin{tabular}{@{}c@{}}{1} \\ {2}\end{tabular}} &
        % MPCNN
        % {} & {} & {} &
        \multicolumn{3}{c|}{\multirow{4}*{not available}} &
        % AutoFHE
        \multicolumn{3}{c|}{\multirow{4}*{not available}} &
        % PAPER
        56.47 & 2088 & 0 & 1 & 128 \\
        % Dataset, Model, ReLU
        {} & {} & {} & {} &
        % MPCNN
        {} & {} & {} &
        % AutoFHE
        {} & {} & {} &
        % PAPER
        58.73 & 2836 & 0 & 2 & 512 \\
        % Dataset, Model, ReLU
        {} & {} & {} & {} &
        % MPCNN
        {} & {} & {} &
        % AutoFHE
        {} & {} & {} &
        % PAPER
        60.78 & 3285 & 0 & 2 & 1024 \\
        % Dataset, Model, ReLU
        {} & {} & {} & {} &
        % MPCNN
        {} & {} & {} &
        % AutoFHE
        {} & {} & {} &
        % PAPER
        62.13 & 4054 & 0 & 2 & 4096 \\
        \hline
        
    \end{tabular}
    \parbox{0.95\textwidth}{\footnotesize{\textsuperscript{$\dagger$}MPCNN and AutoFHE results were reproduced using their released implementations. Their code supports only ResNet-20 on CIFAR-10 and ResNet-32 on CIFAR-10 and CIFAR-100. Therefore, entries marked ``not available'' denote unsupported dataset--model pairs
    of the released frameworks.
%    , whose evaluation would require nontrivial reimplementation and parameter tuning.
    }}
    \label{tab:summary}
\end{table*}

\vspace{1mm}
\noindent\underline{\textbf{Analysis of Parameter Clustering.}}
Figure~\ref{fig:fullpage_results} (left column) shows accuracy and inference time for three configurations: \textit{Standard}, \textit{Full Clustering}, and \textit{Slice Clustering}.
\textit{Standard} applies techniques from \S\ref{sec:model_approx_train} and \S\ref{sec:struct_optmizations}, while \textit{Full Clustering} and \textit{Slice Clustering} extend them with parameter clustering described in \S\ref{sec:weight_clustering}.
Clustering reduces inference time several-fold compared to \textit{Standard}, while maintaining similar accuracy.
At very low centroid counts ($k$), accuracy degrades due to insufficient parameter representation.
Beyond a threshold (e.g., $k=64$ for \textit{Slice Clustering}), accuracy matches \textit{Standard} while providing 4-7$\times$ speedups.
For a fixed $k$, \textit{Full Clustering} is slightly faster, as its centroids are shared across slices and convolutional layers and thus are not fully utilized within individual slices.
By contrast, \textit{Slice Clustering} employs all centroids within each slice, incurring additional computational overhead from extra plaintext encodings, but providing finer granularity and improved parameter representation, resulting in higher accuracy.
In fact, \textit{Slice Clustering with $k=32$ achieves higher accuracy and faster inference than Full Clustering with $k=64$}.
Overall, \textit{Slice Clustering} is superior, surpassing \textit{Full Clustering} in both accuracy and latency.

\begin{figure}
    \centering
    \includegraphics[width=1.0\linewidth]{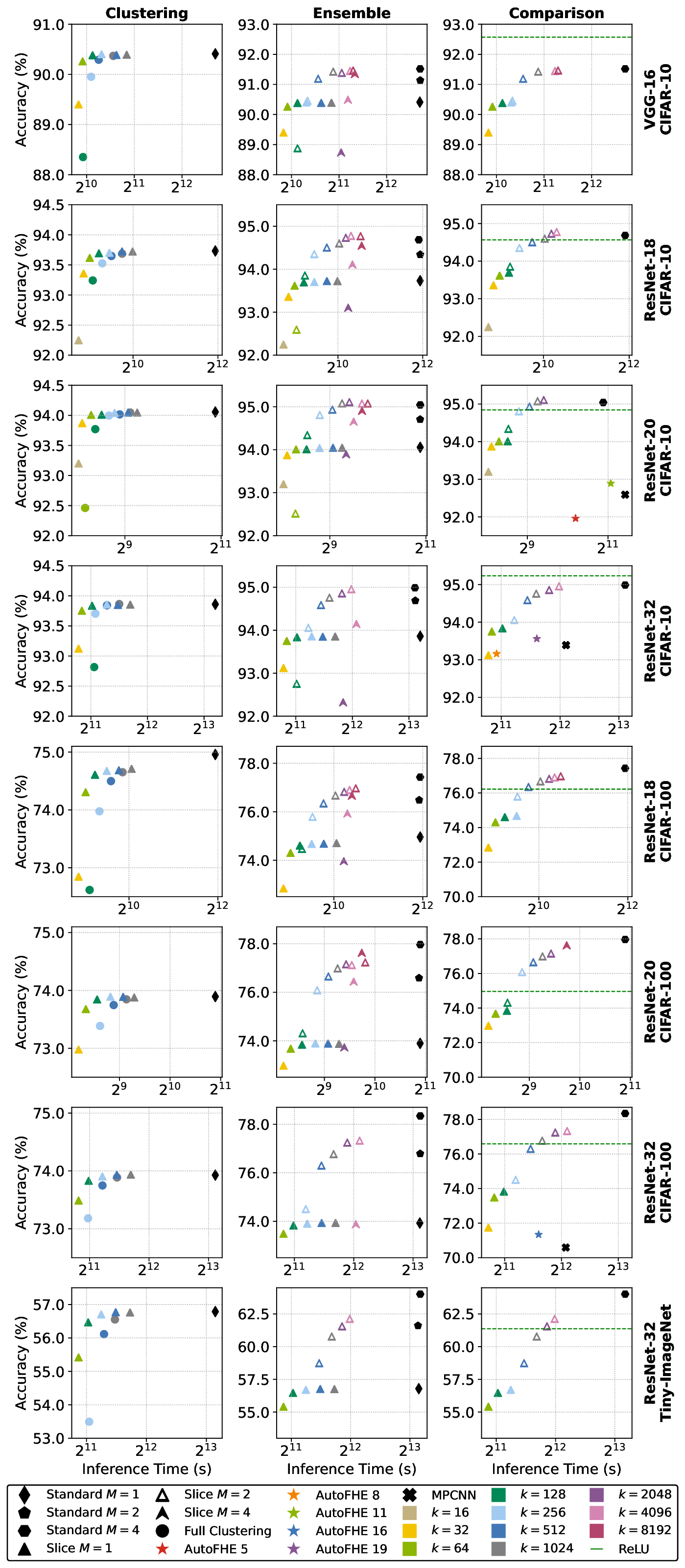}
    \vspace{-6mm}
    \caption{\small Accuracy vs. private inference time for VGG and ResNet on CIFAR and Tiny-ImageNet. Ensembles are evaluated with $M \in \{1, 2, 4\}$ ($M = 1$ is a single model). We compare Standard ensembles (no clustering), Slice ensembles (Slice Clustering), and Full Clustering. Colored squares represent centroid counts ($k = 16, \dots, 8192$), and the same color is used for the same centroid count in both slice and full clustering. Baselines include AutoFHE (numbers = bootstrapping layers) and MPCNN. The green line shows plaintext ReLU accuracy.}
    \label{fig:fullpage_results}
\end{figure}

\vspace{1mm}
\noindent\underline{\textbf{Ensemble Evaluation.}}
In Figure~\ref{fig:fullpage_results} (center column), we evaluate ensembles with $M = \{1,2,4\}$ models.
The figure depicts accuracy and inference time for two approaches: \textit{Standard-$M$} (ensembles without clustering) and \textit{Slice-$M$} (ensembles with slice clustering as in \S\ref{sec:model_ensemble}).
Inference time does not increase with ensemble size since all models fit in one ciphertext, requiring the same number of homomorphic operations.
For \textit{Standard-$M\!$}, accuracy consistently improves as $M$ increases.
\textit{Slice-$M\!$}, however, does not always scale the same way. 
Since each centroid represents an $M$-dimensional space, higher $M$ can require more centroids to capture model parameters.
Thus, \textit{Slice-4} can underperform \textit{Slice-2} at a fixed $k$, and \textit{Slice-2} needs more centroids to match \textit{Standard-2} accuracy than either method with $M = 1$.
Despite, \textit{Slice-$M$} is advantageous as it produces models with higher accuracy.
For instance, \textit{Slice-2} outperforms \textit{Standard-2} and reaches accuracy comparable to \textit{Standard-4}. 
This suggests clustering acts as an implicit regularizer that mitigates overfitting.

\vspace{1mm}
\noindent\underline{\textbf{Accuracy of Polynomial Approximation.}}
Figure~\ref{fig:fullpage_results} (right column) presents Pareto fronts for \textit{Slice Clustering} in terms of accuracy and private inference time, and compares them against AutoFHE and MPCNN.
For reference, the figure also reports \textit{Standard-4} and plaintext ReLU accuracy.
Beyond the gains obtained from our training methodology, \textit{Slice Ensemble Clustering} (\S\ref{sec:model_ensemble}) further improves accuracy by ensembling multiple models within a single ciphertext, allowing higher predictive performance without increasing the number of ciphertext evaluations.
For the simpler CIFAR-10 dataset, \textit{Slice Clustering} achieves accuracy ranging from 1.1\% below to 0.3\% above ReLU.
For the more complex CIFAR-100 and Tiny-ImageNet datasets, our approach consistently surpasses ReLU by 0.7–2.6\%.
These results demonstrate that \textit{our method enables the effective use of low-degree polynomial activations without sacrificing accuracy}.

% 	V16C10	R18C10	R20C10	R32C10
% ReLU	92.6	94.6	94.8	95.2
% Std1	90.4	93.7	94.1	93.9
% Std4	91.5	94.7	95.0	95.0
% FC1	90.4	93.7	94.1	93.9
% SC1	90.4	93.7	94.1	93.9
% SCMax	91.5	94.8	95.1	95.0

% 	R18C100	R20C100	R32C100 R32Tiny
% ReLU	76.2	75.0	76.6	61.4
% Std1	75.0	73.9	73.9	56.8
% Std4	77.4	78.0	78.3	64.0
% FC1	74.7	73.8	73.9	56.6
% SC1	74.7	73.9	73.9	56.8
% SCMax	77.0	77.6	77.3	62.1

\vspace{1mm}
\noindent\underline{\textbf{Comparison to Related Work.}}
As shown in Table~\ref{tab:summary}, for ResNet-20, where our models operate with $N = 2^{15}$, \textit{Slice Clustering} achieves consistently faster inference than related work (up to $4 \times$).  
For ResNet-32, where all approaches use $N = 2^{16}$, our method is comparable in inference time to AutoFHE on CIFAR-10 and $2 \times$ faster on CIFAR-100.
The limited speedup for ResNet-32 arises because it requires $N=2^{16}$ rather than the $N=2^{15}$ used by our smaller models.
Inference time does not change across datasets for our approach, whereas AutoFHE slows down because it requires a high number of bootstrapping operations to maintain accuracy.
The main factors affecting our inference time relative to related work are:
(i) the polynomial activation (\S \ref{sec:model_approx_train}) and structural optimizations (\S \ref{sec:struct_optmizations}) reduce both multiplicative depth and number of RNS moduli, allowing execution under LFHE without bootstrapping;
(ii) our clustering technique (\S \ref{sec:weight_clustering}) further reduces the number of unique parameters, thereby limiting plaintext encodings per convolution;
(iii) for ResNet-20, we operate at a smaller polynomial degree, which directly accelerates homomorphic operations;
(iv) our choice of data layout (\S \ref{sec:data_layout}) deliberately sacrifices some latency to better exploit ciphertext slots for accuracy, but the above optimizations offset this overhead.
The reduction in multiplicative depth is substantial. AutoFHE consumes 14 levels each for bootstrapping and AppReLU, 2 for ConvBN, and 1 each for downsampling, average pooling, and the fully connected layer. In our design, polynomial activation, convolution, and the fully connected layer each have a multiplicative depth of 1, while batch normalization and average pooling consume no levels. Tower reuse further halves level consumption.
In terms of accuracy, our approach consistently outperforms related work thanks to our training strategy (\S \ref{sec:model_approx_train}).
It closely matches plaintext ReLU models, effectively removing the accuracy gap typically observed in polynomial approximations of CNNs.
Beyond accuracy and runtime, memory usage is another important practical consideration. 
For ResNet-20, MPCNN and AutoFHE use about 290~GB, while our method uses roughly 170--280~GB depending on the accuracy setting (Table~\ref{tab:cifar10_resnet20}). 
Our deepest ResNet-32 configurations require terabyte-scale memory (Tables~\ref{tab:cifar10_resnet32}, \ref{tab:cifar100_resnet32}, and \ref{tab:tiny_resnet32}), limiting current use to large-memory servers and privacy-critical workloads where minute-scale latency and high server-side cost are acceptable.

\vspace{1mm}
\noindent\underline{\textbf{Scalability.}}
The training methodology and optimization techniques exhibit scalability along multiple dimensions: they
(i) are applicable to different CNN architectures, as demonstrated with VGG and ResNet models;
(ii) extend to deeper networks, allowing more layers before bootstrapping, as shown with ResNet-32, whereas prior LFHE works were limited to AlexNet and VGG-11 \cite{sisyphus};
(iii) support datasets with both a small (CIFAR-10) and a large number of classes (CIFAR-100); and
(iv) accommodate higher-resolution images (Tiny-ImageNet) with no inference slowdown provided the number of ciphertexts for data encoding remains unchanged.
Moreover, while we restrict our evaluation to models that avoid bootstrapping, our reductions in multiplicative depth directly translate to fewer required bootstrapping operations in larger architectures.
As an analytical projection, our ResNet-32 evaluation consumes 32 levels without bootstrapping.
Assuming a bootstrapping layer consumes the same 14 levels reported in AutoFHE, 18 levels would remain for subsequent computation, sufficient to extend the network to approximately ResNet-50.
Future work will experimentally validate this projection and investigate extending our polynomial approximation, training methodology, and structural optimizations to architectures such as transformers.

% \\ \hl{LINE}
% otherwise, inference time is expected to increase linearly with the number of ciphertexts.
% If additional ciphertexts were to be needed, the inference time should increase linearly with the number of ciphertexts used.

\section{Conclusion}
\label{sec:conclusions}
This work advances the practicality of PPML under FHE. 
We demonstrated that by integrating low-degree polynomial activations with structural and co-design optimizations, it is possible to execute large-scale models such as VGG and ResNet under leveled FHE without bootstrapping, an achievement previously considered not possible. 
State-of-the-art methods typically incur a 2-5\% accuracy loss compared to plaintext models.
Our method is the first to close this gap, achieving accuracy on par with plaintext ReLU baselines while delivering up to 4$\times$ faster private inference on CIFAR and Tiny-ImageNet datasets.
These results mark a step toward making PPML deployable in real-world applications, particularly in domains such as healthcare and finance, where data confidentiality is non-negotiable. 
Future research will extend these methods to larger datasets and deeper architectures to further improve PPML.

% \begin{acks}  %% CCS: MUST BE REMOVED at submission time
% To Robert, for the bagels and explaining CMYK and color spaces.
% \end{acks}

\bibliographystyle{ACM-Reference-Format}
\bibliography{references}

\appendix

\vspace{-1mm}
\section{Open Science}
We release all artifacts through a GitHub repository, available at:
\url{https://github.com/momalab/PAPER}

\section{Ethical Considerations}
Our work improves the performance of PPML using leveled FHE. 
All experiments use public datasets such as CIFAR-10, CIFAR-100, and Tiny-ImageNet. 
The research involves no human subjects, personal data, or interactions with live systems. 
The main stakeholders are end users needing strong protection for sensitive data, providers of encrypted inference services, and the research community.
By reducing reliance on costly bootstrapping, our techniques let deeper neural networks run efficiently on encrypted inputs, which makes practical privacy-preserving machine learning more accessible. 
The method stays within the security guarantees of underlying cryptographic scheme and adds no new vulnerabilities or mechanisms that could enable data misuse or surveillance, focusing only on improving the performance of established secure computation methods. 
Since no human subjects, sensitive information, or activities requiring institutional oversight are involved, no IRB review is required.

\section{Generative AI Usage}
We used ChatGPT to improve readability. It was not involved in generating ideas, analysis, or writing beyond basic language edits.

\clearpage

\section{Why Does the Penalty Function Work?} \label{sec:proof_penalty}
\begin{lemma}[\textbf{Pre-activation Update Decomposition}]\label{lemma:1}
    Let $z_p^{(l)} \;=\; W^{(l)}\;h^{(l-1)}_{p_d}(x) \in \mathbb{R}^{n_l}$ denote the pre-activation vector at layer $l$ for input $x$. We define two quantities based on this vector: the clipping residual
    \begin{equation*}
    	d^{(l)} \;=\; z_p^{(l)} \;-\; \mathrm{clip}\left(z_p^{(l)}; [-c, c]\right)
    \end{equation*}
    which measures the amount by which $z_p^{(l)}$ exceeds the clipping range, and the gradient of the cross-entropy loss with respect to the pre-activation
    \begin{equation*}
    g^{(l)} \;=\; \frac{\partial}{\partial z_p^{(l)}}\ell_{\mathrm{CE}}\left(g(x), y\right)
    \end{equation*}
    which captures the sensitivity of the loss to changes in $z_p^{(l)}$.
    After a single gradient-descent update with learning rate $\eta$ on the minibatch loss $\mathcal{L}_{\mathcal{B}}$, the change in $z_p^{(l)}$ decomposes as
    \begin{equation*}
    \boxed
    {
        \Delta z_p^{(l)}=\underbrace{\left(-\eta\right)\left\|h^{(l-1)}_{p_d}(x)\right\|_2^2g^{(l)}}_{\displaystyle \Delta z_{\rm CE}^{(l)}}+\underbrace{\left(-\eta\right)\zeta\left\|h^{(l-1)}_{p_d}(x)\right\|_2^2\frac{d^{(l)}}{\|d^{(l)}\|_2}}_{\displaystyle\Delta z_{\rm pen}^{(l)}}
    }
    \end{equation*}
    Here, $\Delta z_{\rm CE}^{(l)}$ is the component induced by the cross‐entropy loss, while $\Delta z_{\rm pen}^{(l)}$ is the component induced by the clip‐range penalty.
\end{lemma}

\begin{proof}
    Consider a single training example $(x,y)$, where $x$ is the input and $y$ is the corresponding true label.
    The total per-sample loss includes two terms:
    \begin{equation*}
        \ell(x) = \ell_{\mathrm{CE}}\left(g(x),y\right)\;+\;\zeta\;R\left(z_p^{(l)}\right), \quad R\left(z_p^{(l)}\right)\;=\;\left\|d^{(l)}\right\|_2
    \end{equation*}
    The first term is the standard cross‑entropy loss 
    % computed 
    from the network output $g(x)$.
    The second term is a layer-specific penalty that acts only on the pre-activations at layer $l$, penalizing values that fall outside the clipping interval. 
    A gradient-descent step with learning rate $\eta$ updates the weight matrix by
    % \vspace{-1mm}
    \begin{equation*}
        \begin{aligned}
            \Delta W^{(l)}
            &\;=\; -\eta\;\frac{\partial}{\partial W^{(l)}}\;\ell(x) \\
            &\;=\; -\eta\left(\frac{\partial}{\partial W^{(l)}}\;\ell_{\mathrm{CE}}\left(g(x), y\right)
            \;+\; \zeta\;\frac{\partial}{\partial W^{(l)}}\;R\left(z_p^{(l)}\right)\right)
        \end{aligned}
    \end{equation*}
    Since $z_{p}^{(l)}$ depends linearly on the weights, its Jacobian with respect to $W^{(l)}$ is
    % \vspace{-1mm}
    \begin{equation*}
        \frac{\partial}{\partial W^{(l)}}\;z_{p}^{(l)} = \left[h^{(l-1)}_{p_d}(x)\right]^{\top}.
    \end{equation*}
    Using the chain rule, the gradient of the cross-entropy term with respect to the weights becomes
    % \vspace{-1mm}
    \begin{equation*}
        \begin{aligned}
            \frac{\partial}{\partial W^{(l)}}\;\ell_{\mathrm{CE}}\left(g(x),y\right)
            &\;=\; \frac{\partial}{\partial z_p^{(l)}}\;\ell_{\mathrm{CE}}\left(g(x),y\right)\;\frac{\partial}{\partial W^{(l)}}\;z_p^{(l)}\\
            &\;=\;g^{(l)}\;\left[h^{(l-1)}_{p_d}(x)\right]^{\top}
        \end{aligned}
    \end{equation*}
    The penalty function $R\left(z_p^{(l)}\right)\;=\;\left\|d^{(l)}\right\|_2$ is nonzero only when some elements lie outside $[-c, c]$. Its gradient with respect to $z_p^{(l)}$ is
    % \vspace{-3mm}
    \begin{equation*}
        \frac{\partial}{\partial z_p^{(l)}}\,R\left(z_p^{(l)}\right)\;=\;
        \begin{cases}
            \dfrac{d^{(l)}}{\left\|d^{(l)}\right\|_2} & \text{if } \left\|d^{(l)}\right\|_2 > 0, \\
            0 & \text{if } \left\|d^{(l)}\right\|_2 = 0.
        \end{cases}
    \end{equation*}
    This expression accounts for the case where the clipping has no effect; that is, when $z_p^{(l)} \in [-c, c]$ elementwise, the clipping residual $d^{(l)}$ becomes zero, and the gradient vanishes accordingly. 
    Applying the chain rule again, we compute
    % \vspace{-1mm}
    \begin{equation*}
        \begin{aligned}
            \frac{\partial}{\partial W^{(l)}}\;R\left(z_p^{(l)}\right)
            &\;=\;\frac{\partial}{\partial z_p^{(l)}}\;R\left(z_p^{(l)}\right)\;\frac{\partial}{\partial W^{(l)}}\;z_p^{(l)}\\
            &\;=\;\frac{d^{(l)}}{\left\|d^{(l)}\right\|_2}\;\left[h^{(l-1)}_{p_d}(x)\right]^{\top}
        \end{aligned}
    \end{equation*}
    Putting it all together, we obtain the total 
    % gradient-based 
    update for $W^{(l)}$
    % \vspace{-1mm}
    \begin{equation*}
        \Delta W^{(l)}\;=\;-\eta\left(g^{(l)}\;+\;\zeta\;\frac{d^{(l)}}{\left\|d^{(l)}\right\|_2}\right)\;\left[h^{(l-1)}_{p_d}(x)\right]^{\top}
    \end{equation*}
    By definition, the change in pre-activation $z_p^{(l)}$ after one gradient descent step is given by

    \begin{equation*}
        \begin{aligned}
            \Delta z_p^{(l)}
            &\;=\;\overline{z}_{p}^{(l)}\;-\;z_p^{(l)} \;=\; \overline{W}^{(l)}\;h^{(l-1)}_{p_d}(x)\;-\;W^{(l)}\;h^{(l-1)}_{p_d}(x)\\
            % &\;=\;\overline{W}^{(l)}\;h^{(l-1)}_{p_d}(x)\;-\;W^{(l)}\;h^{(l-1)}_{p_d}(x)\\
            &\;=\;\left(W^{(l)}\;+\;\Delta W^{(l)}\right)\;h^{(l-1)}_{p_d}(x)\;-\;W^{(l)}\;h^{(l-1)}_{p_d}(x)\\
            &\;=\;\Delta W^{(l)}\;h^{(l-1)}_{p_d}(x)\\
        \end{aligned}
    \end{equation*}
    where $\overline{z}_{p}^{(l)}$ and $\overline{W}^{(l)}$ are updated pre-activations and weight matrix after one gradient-descent step. 
    Substituting the previously derived expression for $\Delta W^{(l)}$
    % \vspace{-1mm}
    \begin{equation*}
        \begin{aligned}
            \Delta z_p^{(l)}
            &\;=\;-\eta\left(g^{(l)}\;+\;\zeta\;\frac{d^{(l)}}{\left\|d^{(l)}\right\|_2}\right)\;\left[h^{(l-1)}_{p_d}(x)\right]^{\top}\;h^{(l-1)}_{p_d}(x)\\
            &\;=\;-\eta\left(g^{(l)}\;+\;\zeta\;\frac{d^{(l)}}{\left\|d^{(l)}\right\|_2}\right)\;\left\|h^{(l-1)}_{p_d}(x)\right\|_2^2\\
            &\;=\;\underbrace{\left(-\eta\right)\left\|h^{(l-1)}_{p_d}(x)\right\|_2^2\;g^{(l)}}_{\displaystyle \Delta z_{\rm CE}^{(l)}}+\underbrace{\left(-\eta\right)\zeta\left\|h^{(l-1)}_{p_d}(x)\right\|_2^2\frac{d^{(l)}}{\|d^{(l)}\|_2}}_{\displaystyle\Delta z_{\rm pen}^{(l)}}
        \end{aligned}
    \end{equation*}
    we obtain the penalty components $\Delta z_{\rm CE}^{(l)}$ and $\Delta z_{\rm pen}^{(l)}$. 
    % as mentioned in the lemma.
\end{proof}

\begin{lemma}[\textbf{Clipping Gradient Pullback}]
    Let $\Delta z_{\rm pen}^{(l)}$ be the penalty-induced component of the pre-activation update from Lemma~\ref{lemma:1}, and $d^{(l)}$ be the clipping residual.
    Then the inner product between $\Delta z_{\rm pen}^{(l)}$ and $d^{(l)}$ satisfies
    \begin{equation*}
        \boxed{
            \left\langle \Delta z_{\rm pen}^{(l)},\;d^{(l)}\right\rangle = -\;\eta\;\zeta\;\left\|h^{(l-1)}_{p_d}(x)\right\|_{2}^{2}\;\left\|d^{(l)}\right\|_{2}\;<\;0
        }
    \end{equation*}
    whenever $\left\|d^{(l)}\right\|_2 \neq 0$. 
    Therefore, $\Delta z_{\mathrm{pen}}^{(l)}$ ``pulls back" each element of $z_p^{(l)}$ lying outside $[-c,c]$ towards the clipping interval.
\end{lemma}

\begin{proof}
    The clipping residual $d^{(l)}$ is defined as the difference between the pre‑activation $z_p^{(l)}$ and its clipped version. 
    $\Delta z_{\rm pen}^{(l)}$ is the component in $\Delta z_p^{(l)}$ induced by the clip-range penalty. 
    Taking the inner product of $\Delta z_{\rm pen}^{(l)}$ with $d^{(l)}$ gives
    % \vspace{-1mm}
    
    \begin{equation*}
        \begin{aligned}
            \left\langle \Delta z_{\rm pen},\;d^{(l)}\right\rangle
            &\;=\;\left\langle-\;\eta\;\zeta\;\left\|h^{(l-1)}_{p_d}(x)\right\|_2^{2}\;\frac{d^{(l)}}{\left\|d^{(l)}\right\|_2},\;d^{(l)}\right\rangle \\
            &\;=\;-\eta\;\zeta\;\left\|h^{(l-1)}_{p_d}(x)\right\|_2^2\left\langle \frac{d^{(l)}}{\|d^{(l)}\|_2},\;d^{(l)}\right\rangle\\
            % &\;=\;-\eta\;\zeta\;\left\|h^{(l-1)}_{p_d}(x)\right\|_2^{2}\;\frac{1}{\left\|d^{(l)}\right\|_2}\;\langle d^{(l)},\,d^{(l)}\rangle\\
            &\;=\;-\eta\;\zeta\;\left\|h^{(l-1)}_{p_d}(x)\right\|_2^{2}\;\frac{1}{\left\|d^{(l)}\right\|_2}\,\left\|d^{(l)}\right\|_2^{2}\\
            &\;=\;-\eta\;\zeta\;\left\|h^{(l-1)}_{p_d}(x)\right\|_{2}^{2}\;\left\|d^{(l)}\right\|_{2}
        \end{aligned}
    \end{equation*}
    Each factor on the right-hand side: $\eta$ (learning rate), $\zeta$ (regularization parameter), $\left\|h^{(l-1)}_{p_d}(x)\right\|_2^{2}$ (squared norm of the preceding hidden representation) and $\left\|d^{(l)}\right\|_2$ (magnitude of the clipping residual) are strictly positive whenever any pre‑activation lies outside $[-c,c]$. 
    Hence the entire product is strictly negative
    \begin{equation*}
        \left\langle \Delta z_{\rm pen}^{(l)},\;d^{(l)}\right\rangle\;<\;0
    \end{equation*}
    A negative inner product implies that $\Delta z_{\mathrm{pen}}^{(l)}$ points in the exact opposite direction of the clipping residual $d^{(l)}$. 
    Consequently, every out‑of‑range component of pre‑activation $z_p^{(l)}$ is pulled back towards the clipping boundary, while pre‑activation already within $[-c,c]$ experience no change.
\end{proof}
\section{Details on Structural Optimizations}

\subsection{Derivations of Node Fusing}
\label{sec:node_fusing_derivation}

Batch normalization in linear form is given by
$$
B(x) = b_1 x + b_0, \quad b_1 = \frac{\gamma}{\sigma}, \quad b_0 = \beta_b - b_1 \mu.
$$
A quadratic activation $P(x)$ serves as a polynomial substitute for ReLU, where $P(x)$ is given by
$$
P(x) = c_2 x^2 + c_1 x + c_0.
$$

\textbf{Case 1: $P(B(x)) \mapsto P(x)$.}
Start with $B(x) = b_1 x + b_0$.
$$
\begin{aligned}
P(B(x)) 
&= c_2 \big(B(x)\big)^2 + c_1 B(x) + c_0 \\
= &c_2 (b_1 x + b_0)^2 + c_1 (b_1 x + b_0) + c_0 \\
= &c_2 \big(b_1^2 x^2 + 2 b_1 b_0 x + b_0^2\big) + c_1 b_1 x + c_1 b_0 + c_0 \\
= &(c_2 b_1^2) x^2 + \big(2 c_2 b_1 b_0 + c_1 b_1\big) x + \big(c_2 b_0^2 + c_1 b_0 + c_0\big).
\end{aligned}
$$
Resulting polynomial coefficients
$$
p_2 = b_1^2 c_2,\quad
p_1 = b_1(2 b_0 c_2 + c_1),\quad
p_0 = b_0^2 c_2 + b_0 c_1 + c_0.
$$

\textbf{Case 2: $B(C(x)) \mapsto C(x)$.}
Consider a convolution $C(x) = \sum_i w_i x_i + \beta_c$. When batch normalization is applied, it can be rewritten as a rescaled convolution with updated weights and bias.
$$
\begin{aligned}
B(C(x)) 
&= b_1 C(x) + b_0 \\
&= b_1 \Big(\sum_i w_i x_i + \beta_c\Big) + b_0 \\
&= \sum_i (b_1 w_i) x_i + (b_1 \beta_c + b_0).
\end{aligned}
$$
The equivalent convolution $\sum_i \omega_i x_i + \alpha$ uses the following updated parameters
$$
\omega_i = b_1 w_i,\quad \alpha = b_1 \beta_c + b_0.
$$

\textbf{Case 3: $P(B_X(x)+B_Y(y)) \mapsto S(x,y)$.}
Define $B_X(x) = b_{X1} x + b_{X0}$ and $B_Y(y) = b_{Y1} y + b_{Y0}$.
Each of these represents a linearized batch normalization applied to one input variable.
Let
$
z = B_X(x) + B_Y(y) = b_{X1} x + b_{Y1} y + (b_{X0} + b_{Y0}).
$
Applying the activation
$$
\begin{aligned}
P(z) 
&= c_2 z^2 + c_1 z + c_0 \\
&= c_2 \big(b_{X1} x + b_{Y1} y + b_{X0} + b_{Y0}\big)^2\\
&\quad + c_1 \big(b_{X1} x + b_{Y1} y + b_{X0} + b_{Y0}\big) + c_0.
\end{aligned}
$$
Expanding the square term
$$
\begin{aligned}
z^2 
&= b_{X1}^2 x^2 + b_{Y1}^2 y^2 + 2 b_{X1} b_{Y1} xy 
 + 2 b_{X1}(b_{X0}+b_{Y0}) x \\
&\quad + 2 b_{Y1}(b_{X0}+b_{Y0}) y 
 + (b_{X0}+b_{Y0})^2.
\end{aligned}
$$
Substituting and grouping terms by monomials
$$
\begin{aligned}
P(z) 
&= \underbrace{c_2 b_{X1}^2}_{d_{X2}} x^2
 + \underbrace{c_2 b_{Y1}^2}_{d_{Y2}} y^2
 + \underbrace{2 c_2 b_{X1} b_{Y1}}_{d_{XY}} xy \\
&\quad + \underbrace{\big[b_{X1}\big(2 c_2(b_{X0}+b_{Y0}) + c_1\big)\big]}_{d_X} x\\
&\quad + \underbrace{\big[b_{Y1}\big(2 c_2(b_{X0}+b_{Y0}) + c_1\big)\big]}_{d_Y} y \\
&\quad + \underbrace{\big[c_2(b_{X0}+b_{Y0})^2 + c_1(b_{X0}+b_{Y0}) + c_0\big]}_{d_0}.
\end{aligned}
$$
Hence $S(x,y) = d_{X2} x^2 + d_{Y2} y^2 + d_{XY} xy + d_X x + d_Y y + d_0$ with the coefficients above.

\textbf{Case 4: $P(B_X(x)+y) \mapsto S(x,y)$.}
This is the identity shortcut case. 
Set $B_X(x) = b_{X1} x + b_{X0}$ and define the combined input as
$
z = B_X(x) + y = b_{X1} x + y + b_{X0}.
$
Applying the activation
$$
P(z) = c_2 z^2 + c_1 z + c_0.
$$
Expanding the square term
$$
z^2 = b_{X1}^2 x^2 + y^2 + 2 b_{X1} xy + 2 b_{X1} b_{X0} x + 2 b_{X0} y + b_{X0}^2.
$$
Substituting and grouping terms by monomials
$$
\begin{aligned}
P(z)
&= \underbrace{c_2 b_{X1}^2}_{d_{X2}} x^2
 + \underbrace{c_2}_{d_{Y2}} y^2
 + \underbrace{2 c_2 b_{X1}}_{d_{XY}} xy
 + \underbrace{\big[b_{X1}(2 c_2 b_{X0} + c_1)\big]}_{d_X} x\\
&\quad  + \underbrace{\big[2 c_2 b_{X0} + c_1\big]}_{d_Y} y
+ \underbrace{\big[c_2 b_{X0}^2 + c_1 b_{X0} + c_0\big]}_{d_0}.
\end{aligned}
$$
Hence $S(x,y) = d_{X2} x^2 + d_{Y2} y^2 + d_{XY} xy + d_X x + d_Y y + d_0$ with the coefficients above.

\subsection{Derivations of Weight Redistribution}
\label{sec:weight_redistribution_derivation}

\subsubsection{Update Forward}

\textbf{Donors.}
\textit{Average Pooling.}
Start with $\mu(x) = k^{-1} \sum_i{x_i}$ and its normalized form $\bar{\mu}(x) = \sum_i x_i$.
We must find $\upsilon$ such that $\upsilon \bar{\mu}(x) = \mu(x)$ is valid
$$
\begin{aligned}
    \upsilon \bar{\mu}(x) &= \mu(x) \\
    \upsilon \sum_i x_i &= k^{-1} \sum_i x_i \\
    \upsilon &= k^{-1}.
\end{aligned}
$$

\textit{Polynomial Functions.}
Let $P(x) = \sum_{i=0}^{d}{c_i x^i}$ and $\bar{P}(x) = x^d + \sum_{i=0}^{d-1}{\bar{c}_i x^i}$ be a degree-$d$ polynomial and its normalization, respectively.
For $\upsilon \bar{P}(x) = P(x)$ to hold we have $\upsilon$ given by
$$
\begin{aligned}
    \upsilon \bar{P}(x) &= P(x) \\
    \upsilon \Big(x^d + \sum_{i=0}^{d-1}{\bar{c}_i x^i}\Big) &= \sum_{i=0}^{d}{c_i x^i} \\
    \upsilon x^d &= c_d x^d \\
    \upsilon &= c_d.
\end{aligned}
$$
The normalized coefficients must satisfy
$$
\begin{aligned}
    \upsilon \bar{P}(x) &= P(x) \\
    \upsilon \Big(x^d + \sum_{i=0}^{d-1}{\bar{c}_i x^i}\Big) &= \sum_{i=0}^{d}{c_i x^i} \\
    \upsilon \bar{c}_i x^i &= c_i x^i \\
    \bar{c}_i &= c_i \upsilon^{-1} \quad \forall i \in \{0,\dots,d-1\}.
\end{aligned}
$$

\textit{Bivariate Polynomial.}
Consider the degree-$d$ bivariate polynomial $S(x,y) = \sum_{i=0}^{d}\sum_{j=0}^{d-i}{c_{ij} x^i y^j}$, and let $\bar{S}(x,y) = x^d + \sum_{i=0}^{d-1}\sum_{j=0}^{d-i}{\bar{c}_{ij}x^i y^j}$ be its normalization on $x$ (normalizing on $y$ is equivalent).
As in the univariate case, we require $\upsilon \bar{S}(x,y) = S(x,y)$.
Thus, $\upsilon$ is determined by  
$$
\begin{aligned}
    \upsilon \bar{S}(x,y) &= S(x,y) \\
    \upsilon \Big(x^d + \sum_{i=0}^{d-1}\sum_{j=0}^{d-i}{\bar{c}_{ij}x^i y^j}\Big) &= \sum_{i=0}^{d}\sum_{j=0}^{d-i}{c_{ij} x^i y^j} \\
    \upsilon x^d &= c_{d0} x^d \\
    \upsilon &= c_{d0}.
\end{aligned}
$$
The normalized coefficients are obtained through
$$
\begin{aligned}
    \upsilon \bar{S}(x,y) &= S(x,y) \\
    \upsilon \Big(x^d + \sum_{i=0}^{d-1}\sum_{j=0}^{d-i}{\bar{c}_{ij}x^i y^j}\Big) &= \sum_{i=0}^{d}\sum_{j=0}^{d-i}{c_{ij} x^i y^j} \\
    \upsilon \bar{c}_{ij} x^i y^j &= c_{ij} x^i y^j \\
    \bar{c}_{ij} &= c_{ij} \upsilon^{-1} \qquad 
    \begin{aligned}
    &\forall i \in \{0,\dots,d-1\},\\
    &\forall j \in \{0,\dots,d-i\}.
    \end{aligned}
\end{aligned}
$$

\textbf{Receivers.}
\textit{Kernel Functions.}
Let $K(x) = \sum_{i \in \mathcal{I}} w_i x_i + \beta$ be the original kernel function and $\bar{K}(x) = \sum_{i \in \mathcal{I}} \bar{w}_i x_i + \bar{\beta}$ its updated version.
To ensure model equivalence, it must hold that $\bar{K}(\upsilon^{-1} x) = K(x)$.
Therefore, the updated parameters are given as follows
$$
\begin{aligned}
    \bar{K}(\upsilon^{-1} x) &= K(x) \\
    \sum_{i \in \mathcal{I}}{\bar{w}_i \upsilon^{-1} x_i} + \bar{\beta} &= \sum_{i \in \mathcal{I}} w_i x_i + \beta \\
    \bar{w}_i \upsilon^{-1} x_i &= w_i x_i \\
    \bar{w}_i &= w_i \upsilon \quad \forall i \in \mathcal{I}, \quad \bar{\beta} = \beta.
\end{aligned}
$$

\textit{Polynomial Functions.}
For a degree-$d$ polynomial receiver $P(x) = \sum_{i=0}^{d}{c_i x^i}$, the coefficients of its updated version $\bar{P}(x) = \sum_{i=0}^{d}{\bar{c}_i x^i}$ must be chosen such that $\bar{P}(\upsilon^{-1} x) = P(x)$ holds
$$
\begin{aligned}
    \bar{P}(\upsilon^{-1} x) &= P(x) \\
    \sum_{i=0}^{d}{\bar{c}_i (\upsilon^{-1} x)^i} &= \sum_{i=0}^{d}{c_i x^i} \\
    \bar{c}_i \upsilon^{-i} x^i &= c_i x^i \\
    \bar{c}_i &= c_i \upsilon^i \quad \forall i \in \{0,\dots,d\}.
\end{aligned}
$$

\textit{Bivariate Polynomial.}
For a degree-$d$ bivariate polynomial $S(x,y) = \sum_{i=0}^{d}\sum_{j=0}^{d-i}{c_{ij} x^i y^j}$ and its updated version $\bar{S}(x,y) = \sum_{i=0}^{d}\sum_{j=0}^{d-i}{\bar{c}_{ij}x^i y^j}$,, normalizing with respect to $x$ requires that $\bar{S}(\upsilon^{-1} x,y) = S(x,y)$.
Consequently, the coefficients of $\bar{S}(\cdot)$ are
$$
\begin{aligned}
    \bar{S}(\upsilon^{-1} x,y) &= S(x,y) \\
    \sum_{i=0}^{d}\sum_{j=0}^{d-i}{\bar{c}_{ij}(\upsilon^{-1} x)^i y^j} &= \sum_{i=0}^{d}\sum_{j=0}^{d-i}{c_{ij} x^i y^j} \\
    \bar{c}_{ij} \upsilon^{-i} x^i y^j &= c_{ij} x^i y^j \\
    \bar{c}_{ij} &= c_{ij} \upsilon^i \qquad 
    \begin{aligned}
    &\forall i \in \{0,\dots,d\},\\
    &\forall j \in \{0,\dots,d-i\}.
    \end{aligned}    
\end{aligned}
$$

\subsubsection{Update Backward}

\textbf{Donors.}
\textit{Average Pooling.}
Start with $\mu(x) = k^{-1} \sum_i{x_i}$ and its normalization $\bar{\mu}(x) = \sum_i x_i$.
We must determine $\upsilon$ such that $\bar{\mu}(\upsilon x) = \mu(x)$ holds
$$
\begin{aligned}
    \bar{\mu}(\upsilon x) &= \mu(x) \\
    \sum_i \upsilon x_i &= k^{-1} \sum_i x_i \\
    \upsilon \sum_i x_i &= k^{-1} \sum_i x_i \\
    \upsilon &= k^{-1}.
\end{aligned}
$$

\textit{Polynomial Functions.}
For a degree-$d$ polynomial donor $P(\cdot)$, the equality $\bar{P}(\upsilon x) = P(x)$ must be satisfied, where $\bar{P}(\cdot)$ is the normalized version.
The appropriate $\upsilon$ is found from
$$
\begin{aligned}
    \bar{P}(\upsilon x) &= P(x) \\
    (\upsilon x)^d + \sum_{i=0}^{d-1}{\bar{c}_i (\upsilon x)^i} &= \sum_{i=0}^{d}{c_i x^i} \\
    \upsilon^d x^d &= c_d x^d \\
    \upsilon &= c_d^{1/d}.
\end{aligned}
$$
and the normalized coefficients are

$$
\begin{aligned}
    \bar{P}(\upsilon x) &= P(x) \\
    (\upsilon x)^d + \sum_{i=0}^{d-1}{\bar{c}_i (\upsilon x)^i} &= \sum_{i=0}^{d}{c_i x^i} \\
    \bar{c}_i \upsilon^i x^i &= c_i x^i \\
    \bar{c}_i &= c_i \upsilon^{-i} \quad \forall i \in \{0,\dots,d-1\}.
\end{aligned}
$$

\textit{Bivariate Polynomial.}
Assuming normalization along $x$, we require $\bar{S}(\upsilon x,y) = S(x,y)$ for the degree-$d$ bivariate polynomial $S(x,y)$ and its normalization $\bar{S}(x,y)$.
For that, $\upsilon$ is determined from
$$
\begin{aligned}
    \bar{S}(\upsilon x,y) &= S(x,y) \\
    (\upsilon x)^d + \sum_{i=0}^{d-1}\sum_{j=0}^{d-i}{\bar{c}_{ij}(\upsilon x)^i y^j} &= \sum_{i=0}^{d}\sum_{j=0}^{d-i}{c_{ij} x^i y^j} \\
    \upsilon^d x^d &= c_{d0} x^d \\
    \upsilon &= c_{d0}^{1/d}.
\end{aligned}
$$
with the normalized coefficients being
$$
\begin{aligned}
    \bar{S}(\upsilon x,y) &= S(x,y) \\
    (\upsilon x)^d + \sum_{i=0}^{d-1}\sum_{j=0}^{d-i}{\bar{c}_{ij}(\upsilon x)^i y^j} &= \sum_{i=0}^{d}\sum_{j=0}^{d-i}{c_{ij} x^i y^j} \\
    \bar{c}_{ij} \upsilon^i x^i y^j &= c_{ij} x^i y^j \\
    \bar{c}_{ij} &= c_{ij} \upsilon^{-i} \qquad 
    \begin{aligned}
    &\forall i \in \{0,\dots,d-1\},\\
    &\forall j \in \{0,\dots,d-i\}.
    \end{aligned}
\end{aligned}
$$

\textbf{Receivers.}
\textit{Kernel Functions.}
Let $K(x) = \sum_{i \in \mathcal{I}} w_i x_i + \beta$ and $\bar{K}(x) = \sum_{i \in \mathcal{I}} \bar{w}_i x_i + \bar{\beta}$ be a kernel function and its updated version.
To preserve model equivalence, we require $\bar{K}(x) = \upsilon K(x)$, thus
$$
\begin{aligned}
    \bar{K}(x) &= \upsilon K(x) \\
    \sum_{i \in \mathcal{I}}{\bar{w}_i x_i} + \bar{\beta} &= \upsilon \Big(\sum_{i \in \mathcal{I}} w_i x_i + \beta \Big) \\
    \bar{w}_i x_i &= \upsilon w_i x_i \\
    \bar{w}_i &= w_i \upsilon \quad \forall i \in \mathcal{I}, \quad \bar{\beta} = \beta \upsilon.
\end{aligned}
$$

\textit{Polynomial Functions.}
For a degree-$d$ polynomial receiver $P(x) = \sum_{i=0}^{d} c_i x^i$, the updated polynomial $\bar{P}(x) = \sum_{i=0}^{d} \bar{c}_i x^i$ must satisfy $\bar{P}(x) = \upsilon P(x)$.
This yields
$$
\begin{aligned}
    \bar{P}(x) &= \upsilon P(x) \\
    \sum_{i=0}^{d}{\bar{c}_i x^i} &= \upsilon \sum_{i=0}^{d}{c_i x^i} \\
    \bar{c}_i x^i &= \upsilon c_i x^i \\
    \bar{c}_i &= c_i \upsilon \quad \forall i \in \{0,\dots,d\}.
\end{aligned}
$$

\textit{Bivariate Polynomial.}
For a degree-$d$ bivariate polynomial $S(x,y) = \sum_{i=0}^{d}\sum_{j=0}^{d-i} c_{ij} x^i y^j$ and its updated version $\bar{S}(x,y) = \sum_{i=0}^{d}\sum_{j=0}^{d-i} \bar{c}_{ij} x^i y^j$, we require $\bar{S}(x,y) = \upsilon S(x,y)$.
Therefore
$$
\begin{aligned}
    \bar{S}(x,y) &= \upsilon S(x,y) \\
    \sum_{i=0}^{d}\sum_{j=0}^{d-i}{\bar{c}_{ij} x^i y^j} &= \upsilon \sum_{i=0}^{d}\sum_{j=0}^{d-i}{c_{ij} x^i y^j} \\
    \bar{c}_{ij} x^i y^j &= \upsilon c_{ij} x^i y^j \\
    \bar{c}_{ij} &= c_{ij} \upsilon \qquad 
    \begin{aligned}
    &\forall i \in \{0,\dots,d\},\\
    &\forall j \in \{0,\dots,d-i\}.
    \end{aligned}
\end{aligned}
$$
\section{Details on FHE Parameters}
\label{sec:fhe_additional_details}

For each model, we detail the encryption parameters used in our experiments.
Our framework automatically determines the RNS moduli according to the specified sizes $|q_i|$.
We chose the smallest values of $\Delta$ and $|q_i|$ that preserve accuracy during private inference and adjust $N$ for approximately 128-bit security, where the security level is computed using the LWE Estimator~\cite{lwe_estimator} based on the hardness of the unique Shortest Vector Problem (uSVP).
The moduli are ordered as $q_0, \dots, q_L, P$, where $q_0$ holds the output, $q_1, \ldots, q_L$ are for rescaling, and $P$ is for modulus expansion during key-switching.
FHE terminology is provided in \S \ref{sec:lfhe}.

\vspace{1mm}
\sloppy \textbf{VGG-16:} $N = 2^{15}$, $\Delta = 2^{25}$, $\mathcal{L} = 16$, $\log_{2}{Q} = 6 \cdot 50 + 10 \cdot 49 + 26 = 816$, $\log_{2}{P} = 54$, $\log_2{PQ} = 870$, security level (uSVP) $= 129.7$ bits.
Moduli: 
$\{q_{0},...,q_{\mathcal{L}}\} = $ \{0x4000000120001, 0x3ffffffdf0001, 0x40000001b0001, 0x3ffffffd20001, 0x4000000270001, 0x3ffffffcd0001, 0x20000000b0001, 0x20000001a0001, 0x1fffffff50001, 0x1ffffffea0001, 0x20000003b0001, 0x1ffffffd40001, 0x1ffffffba0001, 0x20000005e0001, 0x1ffffffb40001, 0x20000006d0001, 0x3ff0001\}, $P =$ 0x3fffffffd60001.

\vspace{1mm}
\textbf{ResNet-18:} $N = 2^{15}$, $\Delta = 2^{22}$, $\mathcal{L} = 18$, $\log_{2}{Q} = 18 \cdot 44 + 23 = 815$, $\log_{2}{P} = 54$, $\log_2{PQ} = 869$, security level (uSVP) $= 130.0$ bits.
Moduli: 
$\{q_{0},...,q_{\mathcal{L}}\} = $ \{0x100000020001, 0x100000050001, 0x100000090001, 0x1000000b0001, 0xfffffcf0001, 0x100000180001, 0x1000001a0001, 0xfffffc60001, 0x1000002c0001, 0xfffffb70001, 0x1000002d0001, 0x1000003c0001, 0xfffffb50001, 0x1000003e0001, 0xfffffaf0001, 0x100000480001, 0xfffffac0001, 0x100000570001, 0x820001\}, $P =$ 0x3fffffffd60001.

\vspace{1mm}
\textbf{ResNet-20:} $N = 2^{15}$, $\Delta = 2^{21}$, $\mathcal{L} = 20$, $\log_{2}{Q} = 20 \cdot 42 + 22 = 862$, $\log_{2}{P} = 44$, $\log_2{PQ} = 906$, security level (uSVP) $= 124.4$ bits.
Moduli: 
$\{q_{0},...,q_{\mathcal{L}}\} = $ \{0x400000b0001, 0x3ffffe80001, 0x400002f0001, 0x3ffffd20001, 0x40000330001, 0x3ffffca0001, 0x40000390001, 0x3ffffc30001, 0x400003b0001, 0x3ffffbe0001, 0x400004d0001, 0x3ffff850001, 0x40000560001, 0x400005c0001, 0x3ffff7b0001, 0x400006c0001, 0x3ffff550001, 0x40000770001, 0x3ffff4f0001, 0x400007a0001, 0x390001\}, $P =$ 0x100000020001.

\vspace{1mm}
\textbf{ResNet-32:} $N = 2^{16}$, $\Delta = 2^{26}$, $\mathcal{L} = 32$, $\log_{2}{Q} = 32 \cdot 52 + 27 = 1691$, $\log_{2}{P} = 54$, $\log_2{PQ} = 1745$, security level (uSVP) $= 130.4$ bits.
Moduli:
$\{q_{0},...,q_{\mathcal{L}}\} = $ \{0x10000000060001, 0xffffffff00001, 0x10000000180001, 0xfffffffe40001, 0x10000000200001, 0xfffffffe20001, 0x100000003e0001, 0xfffffffbe0001, 0x10000000500001, 0xfffffffa60001, 0x100000006e0001, 0xfffffff820001, 0x100000007e0001, 0xfffffff480001, 0x10000000960001, 0xfffffff280001, 0x10000000c80001, 0x10000000d80001, 0xffffffed60001, 0x10000000ec0001, 0xffffffec40001, 0x10000000fc0001, 0xffffffeb00001, 0x100000010e0001, 0xffffffe9e0001, 0x10000001380001, 0xffffffe9a0001, 0x100000016a0001, 0xffffffe940001, 0x10000001bc0001, 0xffffffe6a0001, 0x10000001be0001, 0x8020001\}, $P =$ 0x3fffffffd60001.
\section{Detailed Experimental Results}
\label{sec:table_results}

Tables \ref{tab:cifar10_vgg16}, \ref{tab:cifar10_resnet18}, \ref{tab:cifar10_resnet20}, \ref{tab:cifar10_resnet32}, \ref{tab:cifar100_resnet18}, \ref{tab:cifar100_resnet20}, \ref{tab:cifar100_resnet32}, and \ref{tab:tiny_resnet32} report the accuracy, inference time, and peak memory usage of VGG and ResNet models on the CIFAR and Tiny-ImageNet datasets using the Standard, Full Clustering, and Slice Clustering methods.
For Full Clustering and Slice Clustering, we varied the number of centroids as $k = 2^i\ \forall i \in \{1, \dots, 10\}$.
For Standard and Slice Clustering, we additionally evaluated ensemble models with $M \in \{1, 2, 4\}$.
In the case of Slice Clustering with $M \geq 2$, we extended the centroid range to $k = 2^i\  \forall i \in \{6, \dots, 12\}$ in VGG-16, ResNet-18, and ResNet-20, and to the range $k = 2^i\  \forall i \in \{6, \dots, 11\}$ for ResNet-32.
% We did not extend this configuration to ResNet-32 due to system memory limitations (\S \ref{sec:setup}).
% For ResNet-32, we set it to $k = 2^i\  \forall i \in \{6, \dots, 11\}$ due to system memory limitations (\S \ref{sec:setup}).

% \input{tables/cifar10_vgg16}
% \input{tables/cifar10_resnet18}
% \input{tables/cifar10_resnet20}
% \input{tables/cifar10_resnet32}
% \input{tables/cifar100_resnet18}
% \input{tables/cifar100_resnet20}
% \input{tables/cifar100_resnet32}
% \input{tables/tiny_resnet32}

% \clearpage
% \clearpage

% \section{Open Science}
% To support reproducibility, we release all artifacts associated with this work through an anonymous GitHub repository, available at:
% \url{https://gitfront.io/r/CCS2026/d8xMQEXmqfBg/PAPER/}.
% The repository contains the complete LFHE inference stack, including all source code, trained models, and execution scripts necessary to reproduce the experimental results reported in this paper.

\clearpage
\begin{table}[!t]
    \centering
    \small
    \caption
    {
        Accuracy, inference time, and peak memory for the CIFAR-10 on VGG-16.
    }
    \begin{tabular}
    {
        C{1.0cm} |
        C{1.0cm} |
        C{1.0cm} |
        C{1.0cm} |
        C{1.0cm} |
        C{1.0cm}
    }

        % Standard
        \cline{1-6}
        
        \multicolumn{6}{c}{Standard} \\
        \cline{1-6}

        \multirow{2}*{M} &
        \multicolumn{2}{c|}{Accuracy (\%)} &
        \multicolumn{2}{c|}{Inf. Time (s)} &
        {Memory} \\

        {} & \multicolumn{1}{c}{\it $\mu$} & {\it $\sigma$} & \multicolumn{1}{c}{\it $\mu$} & {\it $\sigma$} & {(GB)} \\  
        \cline{1-6} 

        1 & 90.41 & 0.47 & 6648 & 57 & 763 \\
        2 & 91.13 & 0.23 & 6686 & 51 & 765 \\
        4 & 91.52 & 0.20 & 6675 & 55 & 785 \\
        \cline{1-6}
        \multicolumn{6}{c}{} \\

        % Full Clustering (M = 1)
        \cline{1-6}
        
        \multicolumn{6}{c}{Full Clustering (M = 1)} \\
        \cline{1-6}

        \multirow{2}*{k} &
        \multicolumn{2}{c|}{Accuracy (\%)} &
        \multicolumn{2}{c|}{Inf. Time (s)} &
        {Memory} \\

        {} & \multicolumn{1}{c}{\it $\mu$} & {\it $\sigma$} & \multicolumn{1}{c}{\it $\mu$} & {\it $\sigma$} & {(GB)} \\  
        \cline{1-6} 

        2 & 10.02 & 0.28 & 965 & 36 & 653 \\
        4 & 9.93 & 1.70 & 883 & 74 & 651 \\
        8 & 10.04 & 0.55 & 887 & 32 & 646 \\
        16 & 10.23 & 1.41 & 873 & 57 & 651 \\
        32 & 35.22 & 14.41 & 899 & 53 & 653 \\
        64 & 80.20 & 6.32 & 905 & 60 & 647 \\
        128 & 88.35 & 1.24 & 972 & 59 & 653 \\
        256 & 89.95 & 0.53 & 1093 & 78 & 649 \\
        512 & 90.29 & 0.45 & 1220 & 99 & 644 \\
        1024 & 90.37 & 0.36 & 1508 & 142 & 655 \\
        \cline{1-6}
        \multicolumn{6}{c}{} \\

        % Slice Clustering (M = 1)
        \cline{1-6}
        
        \multicolumn{6}{c}{Slice Clustering (M = 1)} \\
        \cline{1-6}

        \multirow{2}*{k} &
        \multicolumn{2}{c|}{Accuracy (\%)} &
        \multicolumn{2}{c|}{Inf. Time (s)} &
        {Memory} \\

        {} & \multicolumn{1}{c}{\it $\mu$} & {\it $\sigma$} & \multicolumn{1}{c}{\it $\mu$} & {\it $\sigma$} & {(GB)} \\  
        \cline{1-6} 

        2 & 10.02 & 0.80 & 865 & 29 & 645 \\
        4 & 13.76 & 2.84 & 878 & 31 & 650 \\
        8 & 56.06 & 9.60 & 858 & 20 & 653 \\
        16 & 84.70 & 2.70 & 901 & 25 & 655 \\
        32 & 89.40 & 0.62 & 911 & 20 & 645 \\
        64 & 90.26 & 0.38 & 966 & 37 & 647 \\
        128 & 90.38 & 0.39 & 1117 & 48 & 649 \\
        256 & 90.40 & 0.38 & 1273 & 70 & 652 \\
        512 & 90.39 & 0.44 & 1579 & 88 & 661 \\
        1024 & 90.39 & 0.44 & 1833 & 114 & 664 \\
        \cline{1-6}
        \multicolumn{6}{c}{} \\

        % Slice Clustering (M = 2)
        \cline{1-6}
        
        \multicolumn{6}{c}{Slice Clustering (M = 2)} \\
        \cline{1-6}

        \multirow{2}*{k} &
        \multicolumn{2}{c|}{Accuracy (\%)} &
        \multicolumn{2}{c|}{Inf. Time (s)} &
        {Memory} \\

        {} & \multicolumn{1}{c}{\it $\mu$} & {\it $\sigma$} & \multicolumn{1}{c}{\it $\mu$} & {\it $\sigma$} & {(GB)} \\  
        \cline{1-6} 

        64 & 82.31 & 2.13 & 983 & 64 & 663 \\
        128 & 88.88 & 0.62 & 1118 & 64 & 656 \\
        256 & 90.46 & 0.26 & 1293 & 72 & 658 \\
        512 & 91.19 & 0.14 & 1510 & 17 & 658 \\
        1024 & 91.42 & 0.09 & 1881 & 37 & 680 \\
        2048 & 91.38 & 0.09 & 2130 & 46 & 691 \\
        4096 & 91.45 & 0.06 & 2405 & 42 & 728 \\
        8192 & 91.46 & 0.04 & 2514 & 74 & 814 \\
        \cline{1-6}
        \multicolumn{6}{c}{} \\

        % Slice Clustering (M = 4)
        \cline{1-6}
        
        \multicolumn{6}{c}{Slice Clustering (M = 4)} \\
        \cline{1-6}

        \multirow{2}*{k} &
        \multicolumn{2}{c|}{Accuracy (\%)} &
        \multicolumn{2}{c|}{Inf. Time (s)} &
        {Memory} \\

        {} & \multicolumn{1}{c}{\it $\mu$} & {\it $\sigma$} & \multicolumn{1}{c}{\it $\mu$} & {\it $\sigma$} & {(GB)} \\  
        \cline{1-6} 

        64 & 13.53 & 2.29 & 1008 & 62 & 675 \\
        128 & 24.53 & 6.48 & 1122 & 73 & 681 \\
        256 & 50.90 & 4.11 & 1384 & 97 & 676 \\
        512 & 73.45 & 1.64 & 1548 & 65 & 686 \\
        1024 & 85.17 & 1.13 & 1917 & 64 & 695 \\
        2048 & 88.74 & 0.36 & 2110 & 65 & 711 \\
        4096 & 90.49 & 0.27 & 2330 & 86 & 738 \\
        8192 & 91.34 & 0.18 & 2579 & 115 & 833 \\
        \cline{1-6}
    \end{tabular}
    \label{tab:cifar10_vgg16}
\end{table}
\begin{table}[!t]
    \centering
    \small
    \caption
    {
        Accuracy, inference time, and peak memory for CIFAR-10 on ResNet-18.
    }
    \begin{tabular}
    {
        C{1.0cm} |
        C{1.0cm} |
        C{1.0cm} |
        C{1.0cm} |
        C{1.0cm} |
        C{1.0cm}
    }

        % Standard
        \cline{1-6}
        
        \multicolumn{6}{c}{Standard} \\
        \cline{1-6}

        \multirow{2}*{M} &
        \multicolumn{2}{c|}{Accuracy (\%)} &
        \multicolumn{2}{c|}{Inf. Time (s)} &
        {Memory} \\

        {} & \multicolumn{1}{c}{\it $\mu$} & {\it $\sigma$} & \multicolumn{1}{c}{\it $\mu$} & {\it $\sigma$} & {(GB)} \\  
        \cline{1-6} 

        1 & 93.73 & 0.27 & 3925 & 51 & 292 \\
        2 & 94.34 & 0.16 & 3930 & 54 & 314 \\
        4 & 94.69 & 0.12 & 3840 & 44 & 347 \\
        \cline{1-6}
        \multicolumn{6}{c}{} \\

        % Full Clustering (M = 1)
        \cline{1-6}
        
        \multicolumn{6}{c}{Full Clustering (M = 1)} \\
        \cline{1-6}

        \multirow{2}*{k} &
        \multicolumn{2}{c|}{Accuracy (\%)} &
        \multicolumn{2}{c|}{Inf. Time (s)} &
        {Memory} \\

        {} & \multicolumn{1}{c}{\it $\mu$} & {\it $\sigma$} & \multicolumn{1}{c}{\it $\mu$} & {\it $\sigma$} & {(GB)} \\  
        \cline{1-6} 

        2 & 10.16 & 0.72 & 404 & 16 & 187 \\
        4 & 10.14 & 0.89 & 406 & 18 & 187 \\
        8 & 9.91 & 1.48 & 420 & 31 & 187 \\
        16 & 15.89 & 7.51 & 427 & 27 & 188 \\
        32 & 70.18 & 14.19 & 439 & 28 & 188 \\
        64 & 91.42 & 1.25 & 481 & 31 & 188 \\
        128 & 93.24 & 0.29 & 530 & 34 & 188 \\
        256 & 93.53 & 0.26 & 619 & 36 & 189 \\
        512 & 93.65 & 0.25 & 719 & 40 & 189 \\
        1024 & 93.69 & 0.28 & 856 & 34 & 196 \\
        \cline{1-6}
        \multicolumn{6}{c}{} \\

        % Slice Clustering (M = 1)
        \cline{1-6}
        
        \multicolumn{6}{c}{Slice Clustering (M = 1)} \\
        \cline{1-6}

        \multirow{2}*{k} &
        \multicolumn{2}{c|}{Accuracy (\%)} &
        \multicolumn{2}{c|}{Inf. Time (s)} &
        {Memory} \\

        {} & \multicolumn{1}{c}{\it $\mu$} & {\it $\sigma$} & \multicolumn{1}{c}{\it $\mu$} & {\it $\sigma$} & {(GB)} \\  
        \cline{1-6} 

        2 & 9.79 & 0.76 & 418 & 32 & 188 \\
        4 & 12.66 & 3.02 & 409 & 23 & 187 \\
        8 & 81.28 & 6.04 & 412 & 23 & 187 \\
        16 & 92.25 & 0.68 & 420 & 19 & 188 \\
        32 & 93.36 & 0.33 & 456 & 28 & 188 \\
        64 & 93.62 & 0.27 & 504 & 30 & 188 \\
        128 & 93.70 & 0.26 & 586 & 30 & 189 \\
        256 & 93.70 & 0.28 & 696 & 34 & 189 \\
        512 & 93.73 & 0.27 & 855 & 58 & 192 \\
        1024 & 93.72 & 0.27 & 1016 & 57 & 203 \\
        \cline{1-6}
        \multicolumn{6}{c}{} \\

        % Slice Clustering (M = 2)
        \cline{1-6}
        
        \multicolumn{6}{c}{Slice Clustering (M = 2)} \\
        \cline{1-6}

        \multirow{2}*{k} &
        \multicolumn{2}{c|}{Accuracy (\%)} &
        \multicolumn{2}{c|}{Inf. Time (s)} &
        {Memory} \\

        {} & \multicolumn{1}{c}{\it $\mu$} & {\it $\sigma$} & \multicolumn{1}{c}{\it $\mu$} & {\it $\sigma$} & {(GB)} \\  
        \cline{1-6} 

        64 & 92.59 & 0.45 & 520 & 43 & 205 \\
        128 & 93.85 & 0.15 & 597 & 38 & 206 \\
        256 & 94.35 & 0.10 & 695 & 22 & 206 \\
        512 & 94.50 & 0.09 & 854 & 28 & 210 \\
        1024 & 94.60 & 0.08 & 1045 & 61 & 220 \\
        2048 & 94.73 & 0.07 & 1164 & 36 & 240 \\
        4096 & 94.78 & 0.04 & 1268 & 42 & 301 \\
        8192 & 94.77 & 0.05 & 1482 & 59 & 459 \\
        \cline{1-6}
        \multicolumn{6}{c}{} \\

        % Slice Clustering (M = 4)
        \cline{1-6}
        
        \multicolumn{6}{c}{Slice Clustering (M = 4)} \\
        \cline{1-6}

        \multirow{2}*{k} &
        \multicolumn{2}{c|}{Accuracy (\%)} &
        \multicolumn{2}{c|}{Inf. Time (s)} &
        {Memory} \\

        {} & \multicolumn{1}{c}{\it $\mu$} & {\it $\sigma$} & \multicolumn{1}{c}{\it $\mu$} & {\it $\sigma$} & {(GB)} \\  
        \cline{1-6} 

        64 & 18.38 & 4.80 & 526 & 27 & 240 \\
        128 & 35.79 & 8.02 & 621 & 45 & 240 \\
        256 & 66.88 & 5.19 & 728 & 51 & 241 \\
        512 & 84.28 & 2.38 & 866 & 28 & 245 \\
        1024 & 91.47 & 0.46 & 1035 & 69 & 255 \\
        2048 & 93.10 & 0.21 & 1211 & 57 & 275 \\
        4096 & 94.11 & 0.14 & 1298 & 17 & 337 \\
        8192 & 94.55 & 0.11 & 1509 & 22 & 494 \\
        \cline{1-6}
    \end{tabular}
    \label{tab:cifar10_resnet18}
\end{table}
\begin{table}[!t]
    \centering
    \small
    \caption
    {
        Accuracy, inference time, and peak memory for CIFAR-10 on ResNet-20.
    }
    \begin{tabular}
    {
        C{1.0cm} |
        C{1.0cm} |
        C{1.0cm} |
        C{1.0cm} |
        C{1.0cm} |
        C{1.0cm}
    }

        % Standard
        \cline{1-6}
        
        \multicolumn{6}{c}{Standard} \\
        \cline{1-6}

        \multirow{2}*{M} &
        \multicolumn{2}{c|}{Accuracy (\%)} &
        \multicolumn{2}{c|}{Inf. Time (s)} &
        {Memory} \\

        {} & \multicolumn{1}{c}{\it $\mu$} & {\it $\sigma$} & \multicolumn{1}{c}{\it $\mu$} & {\it $\sigma$} & {(GB)} \\  
        \cline{1-6} 

        1 & 94.06 & 0.26 & 1886 & 30 & 279 \\
        2 & 94.71 & 0.11 & 1884 & 41 & 300 \\
        4 & 95.04 & 0.10 & 1886 & 32 & 330 \\
        \cline{1-6}
        \multicolumn{6}{c}{} \\

        % Full Clustering (M = 1)
        \cline{1-6}
        
        \multicolumn{6}{c}{Full Clustering (M = 1)} \\
        \cline{1-6}

        \multirow{2}*{k} &
        \multicolumn{2}{c|}{Accuracy (\%)} &
        \multicolumn{2}{c|}{Inf. Time (s)} &
        {Memory} \\

        {} & \multicolumn{1}{c}{\it $\mu$} & {\it $\sigma$} & \multicolumn{1}{c}{\it $\mu$} & {\it $\sigma$} & {(GB)} \\  
        \cline{1-6} 

        2 & 10.08 & 0.26 & 242 & 10 & 172 \\
        4 & 10.08 & 0.32 & 246 & 6 & 172 \\
        8 & 10.17 & 0.89 & 250 & 12 & 172 \\
        16 & 14.44 & 6.34 & 260 & 13 & 172 \\
        32 & 78.18 & 8.99 & 264 & 7 & 172 \\
        64 & 92.46 & 0.95 & 288 & 10 & 172 \\
        128 & 93.77 & 0.37 & 334 & 12 & 172 \\
        256 & 94.00 & 0.26 & 407 & 22 & 172 \\
        512 & 94.02 & 0.26 & 475 & 18 & 173 \\
        1024 & 94.05 & 0.26 & 553 & 25 & 181 \\
        \cline{1-6}
        \multicolumn{6}{c}{} \\

        % Slice Clustering (M = 1)
        \cline{1-6}
        
        \multicolumn{6}{c}{Slice Clustering (M = 1)} \\
        \cline{1-6}

        \multirow{2}*{k} &
        \multicolumn{2}{c|}{Accuracy (\%)} &
        \multicolumn{2}{c|}{Inf. Time (s)} &
        {Memory} \\

        {} & \multicolumn{1}{c}{\it $\mu$} & {\it $\sigma$} & \multicolumn{1}{c}{\it $\mu$} & {\it $\sigma$} & {(GB)} \\  
        \cline{1-6} 

        2 & 10.05 & 0.41 & 249 & 18 & 172 \\
        4 & 14.59 & 4.33 & 244 & 7 & 172 \\
        8 & 85.17 & 2.66 & 247 & 7 & 172 \\
        16 & 93.20 & 0.45 & 262 & 6 & 172 \\
        32 & 93.87 & 0.32 & 276 & 9 & 172 \\
        64 & 94.01 & 0.28 & 314 & 9 & 172 \\
        128 & 94.01 & 0.28 & 366 & 14 & 172 \\
        256 & 94.04 & 0.27 & 442 & 10 & 172 \\
        512 & 94.05 & 0.26 & 536 & 26 & 177 \\
        1024 & 94.05 & 0.26 & 612 & 19 & 190 \\
        \cline{1-6}
        \multicolumn{6}{c}{} \\

        % Slice Clustering (M = 2)
        \cline{1-6}
        
        \multicolumn{6}{c}{Slice Clustering (M = 2)} \\
        \cline{1-6}

        \multirow{2}*{k} &
        \multicolumn{2}{c|}{Accuracy (\%)} &
        \multicolumn{2}{c|}{Inf. Time (s)} &
        {Memory} \\

        {} & \multicolumn{1}{c}{\it $\mu$} & {\it $\sigma$} & \multicolumn{1}{c}{\it $\mu$} & {\it $\sigma$} & {(GB)} \\  
        \cline{1-6} 

        64 & 92.52 & 0.56 & 312 & 9 & 187 \\
        128 & 94.34 & 0.10 & 370 & 8 & 187 \\
        256 & 94.81 & 0.08 & 443 & 20 & 188 \\
        512 & 94.93 & 0.08 & 531 & 31 & 192 \\
        1024 & 95.08 & 0.07 & 611 & 14 & 205 \\
        2048 & 95.11 & 0.06 & 679 & 7 & 239 \\
        4096 & 95.08 & 0.04 & 814 & 60 & 333 \\
        8192 & 95.08 & 0.03 & 885 & 37 & 541 \\
        \cline{1-6}
        \multicolumn{6}{c}{} \\

        % Slice Clustering (M = 4)
        \cline{1-6}
        
        \multicolumn{6}{c}{Slice Clustering (M = 4)} \\
        \cline{1-6}

        \multirow{2}*{k} &
        \multicolumn{2}{c|}{Accuracy (\%)} &
        \multicolumn{2}{c|}{Inf. Time (s)} &
        {Memory} \\

        {} & \multicolumn{1}{c}{\it $\mu$} & {\it $\sigma$} & \multicolumn{1}{c}{\it $\mu$} & {\it $\sigma$} & {(GB)} \\  
        \cline{1-6} 

        64 & 12.79 & 2.88 & 332 & 14 & 219 \\
        128 & 25.95 & 4.81 & 389 & 7 & 219 \\
        256 & 57.23 & 3.93 & 455 & 21 & 219 \\
        512 & 83.18 & 1.59 & 534 & 10 & 224 \\
        1024 & 91.45 & 0.46 & 603 & 18 & 237 \\
        2048 & 93.89 & 0.26 & 648 & 9 & 270 \\
        4096 & 94.65 & 0.16 & 722 & 12 & 365 \\
        8192 & 94.90 & 0.10 & 815 & 17 & 571 \\
        \cline{1-6}
    \end{tabular}
    \label{tab:cifar10_resnet20}
\end{table}
\begin{table}[!t]
    \centering
    \small
    \caption
    {
        Accuracy, inference time, and peak memory for CIFAR-10 on ResNet-32.
    }
    \begin{tabular}
    {
        C{1.0cm} |
        C{1.0cm} |
        C{1.0cm} |
        C{1.0cm} |
        C{1.0cm} |
        C{1.0cm}
    }

        % Standard
        \cline{1-6}
        
        \multicolumn{6}{c}{Standard} \\
        \cline{1-6}

        \multirow{2}*{M} &
        \multicolumn{2}{c|}{Accuracy (\%)} &
        \multicolumn{2}{c|}{Inf. Time (s)} &
        {Memory} \\

        {} & \multicolumn{1}{c}{\it $\mu$} & {\it $\sigma$} & \multicolumn{1}{c}{\it $\mu$} & {\it $\sigma$} & {(GB)} \\  
        \cline{1-6} 

        1 & 93.86 & 0.24 & 9405 & 109 & 1332 \\
        2 & 94.69 & 0.18 & 8854 & 80 & 1394 \\
        4 & 94.99 & 0.13 & 8813 & 63 & 1441 \\
        \cline{1-6}
        \multicolumn{6}{c}{} \\

        % Full Clustering (M = 1)
        \cline{1-6}
        
        \multicolumn{6}{c}{Full Clustering (M = 1)} \\
        \cline{1-6}

        \multirow{2}*{k} &
        \multicolumn{2}{c|}{Accuracy (\%)} &
        \multicolumn{2}{c|}{Inf. Time (s)} &
        {Memory} \\

        {} & \multicolumn{1}{c}{\it $\mu$} & {\it $\sigma$} & \multicolumn{1}{c}{\it $\mu$} & {\it $\sigma$} & {(GB)} \\  
        \cline{1-6} 

        2 & 10.12 & 0.22 & 1578 & 63 & 794 \\
        4 & 10.13 & 0.21 & 1538 & 3 & 795 \\
        8 & 10.12 & 0.22 & 1588 & 79 & 794 \\
        16 & 10.35 & 1.35 & 1570 & 8 & 795 \\
        32 & 37.13 & 16.87 & 1578 & 3 & 795 \\
        64 & 86.61 & 3.69 & 1703 & 80 & 795 \\
        128 & 92.82 & 0.47 & 2131 & 27 & 794 \\
        256 & 93.70 & 0.32 & 2164 & 151 & 795 \\
        512 & 93.84 & 0.27 & 2495 & 164 & 796 \\
        1024 & 93.86 & 0.26 & 2898 & 174 & 837 \\
        \cline{1-6}
        \multicolumn{6}{c}{} \\

        % Slice Clustering (M = 1)
        \cline{1-6}
        
        \multicolumn{6}{c}{Slice Clustering (M = 1)} \\
        \cline{1-6}

        \multirow{2}*{k} &
        \multicolumn{2}{c|}{Accuracy (\%)} &
        \multicolumn{2}{c|}{Inf. Time (s)} &
        {Memory} \\

        {} & \multicolumn{1}{c}{\it $\mu$} & {\it $\sigma$} & \multicolumn{1}{c}{\it $\mu$} & {\it $\sigma$} & {(GB)} \\  
        \cline{1-6} 

        2 & 10.12 & 1.07 & 1605 & 57 & 795 \\
        4 & 10.00 & 0.99 & 1573 & 24 & 795 \\
        8 & 34.61 & 12.19 & 1601 & 65 & 795 \\
        16 & 89.08 & 2.44 & 1559 & 6 & 795 \\
        32 & 93.12 & 0.49 & 1761 & 90 & 795 \\
        64 & 93.75 & 0.29 & 1833 & 54 & 795 \\
        128 & 93.84 & 0.26 & 2077 & 111 & 794 \\
        256 & 93.86 & 0.27 & 2488 & 185 & 796 \\
        512 & 93.85 & 0.25 & 2849 & 128 & 818 \\
        1024 & 93.86 & 0.24 & 3317 & 192 & 883 \\
        \cline{1-6}
        \multicolumn{6}{c}{} \\

        % Slice Clustering (M = 2)
        \cline{1-6}
        
        \multicolumn{6}{c}{Slice Clustering (M = 2)} \\
        \cline{1-6}

        \multirow{2}*{k} &
        \multicolumn{2}{c|}{Accuracy (\%)} &
        \multicolumn{2}{c|}{Inf. Time (s)} &
        {Memory} \\

        {} & \multicolumn{1}{c}{\it $\mu$} & {\it $\sigma$} & \multicolumn{1}{c}{\it $\mu$} & {\it $\sigma$} & {(GB)} \\  
        \cline{1-6} 

        64 & 85.12 & 4.61 & 1831 & 69 & 820 \\
        128 & 92.75 & 0.58 & 2069 & 79 & 820 \\
        256 & 94.06 & 0.21 & 2388 & 156 & 821 \\
        512 & 94.59 & 0.12 & 2788 & 167 & 843 \\
        1024 & 94.76 & 0.12 & 3093 & 11 & 908 \\
        2048 & 94.86 & 0.06 & 3602 & 32 & 1083 \\
        4096 & 94.95 & 0.07 & 4043 & 59 & 1573 \\
        % 8192 & 94.90 & 0.05 & \multicolumn{3}{c}{out of memory} \\
        \cline{1-6}
        \multicolumn{6}{c}{} \\

        % Slice Clustering (M = 4)
        \cline{1-6}
        
        \multicolumn{6}{c}{Slice Clustering (M = 4)} \\
        \cline{1-6}

        \multirow{2}*{k} &
        \multicolumn{2}{c|}{Accuracy (\%)} &
        \multicolumn{2}{c|}{Inf. Time (s)} &
        {Memory} \\

        {} & \multicolumn{1}{c}{\it $\mu$} & {\it $\sigma$} & \multicolumn{1}{c}{\it $\mu$} & {\it $\sigma$} & {(GB)} \\  
        \cline{1-6} 

        64 & 9.27 & 0.65 & 1845 & 89 & 870 \\
        128 & 14.20 & 2.55 & 2096 & 100 & 870 \\
        256 & 30.01 & 9.08 & 2409 & 149 & 871 \\
        512 & 69.18 & 5.72 & 2842 & 193 & 894 \\
        1024 & 88.24 & 1.02 & 3379 & 64 & 958 \\
        2048 & 92.32 & 0.68 & 3657 & 123 & 1133 \\
        4096 & 94.14 & 0.24 & 4301 & 130 & 1624 \\
        % 8192 & 94.89 & 0.16 & \multicolumn{3}{c}{out of memory} \\
        \cline{1-6}
    \end{tabular}
    \label{tab:cifar10_resnet32}
\end{table}
\begin{table}[!t]
    \centering
    \small
    \caption
    {
        Accuracy, inference time, and peak memory for CIFAR-100 on ResNet-18.
    }
    \begin{tabular}
    {
        C{1.0cm} |
        C{1.0cm} |
        C{1.0cm} |
        C{1.0cm} |
        C{1.0cm} |
        C{1.0cm}
    }

        % Standard
        \cline{1-6}
        
        \multicolumn{6}{c}{Standard} \\
        \cline{1-6}

        \multirow{2}*{M} &
        \multicolumn{2}{c|}{Accuracy (\%)} &
        \multicolumn{2}{c|}{Inf. Time (s)} &
        {Memory} \\

        {} & \multicolumn{1}{c}{\it $\mu$} & {\it $\sigma$} & \multicolumn{1}{c}{\it $\mu$} & {\it $\sigma$} & {(GB)} \\  
        \cline{1-6} 

        1 & 74.96 & 0.51 & 3935 & 79 & 338 \\
        2 & 76.48 & 0.28 & 3874 & 40 & 355 \\
        4 & 77.43 & 0.24 & 3929 & 54 & 390 \\
        \cline{1-6}
        \multicolumn{6}{c}{} \\

        % Full Clustering (M = 1)
        \cline{1-6}
        
        \multicolumn{6}{c}{Full Clustering (M = 1)} \\
        \cline{1-6}

        \multirow{2}*{k} &
        \multicolumn{2}{c|}{Accuracy (\%)} &
        \multicolumn{2}{c|}{Inf. Time (s)} &
        {Memory} \\

        {} & \multicolumn{1}{c}{\it $\mu$} & {\it $\sigma$} & \multicolumn{1}{c}{\it $\mu$} & {\it $\sigma$} & {(GB)} \\  
        \cline{1-6} 

        2 & 1.00 & 0.07 & 414 & 25 & 229 \\
        4 & 1.02 & 0.12 & 413 & 17 & 228 \\
        8 & 1.08 & 0.17 & 426 & 23 & 229 \\
        16 & 2.01 & 1.10 & 430 & 25 & 228 \\
        32 & 32.33 & 13.66 & 441 & 23 & 228 \\
        64 & 68.01 & 4.00 & 484 & 34 & 229 \\
        128 & 72.62 & 1.04 & 554 & 41 & 228 \\
        256 & 73.98 & 0.68 & 646 & 44 & 230 \\
        512 & 74.50 & 0.70 & 773 & 53 & 230 \\
        1024 & 74.65 & 0.69 & 926 & 58 & 237 \\
        \cline{1-6}
        \multicolumn{6}{c}{} \\

        % Slice Clustering (M = 1)
        \cline{1-6}
        
        \multicolumn{6}{c}{Slice Clustering (M = 1)} \\
        \cline{1-6}

        \multirow{2}*{k} &
        \multicolumn{2}{c|}{Accuracy (\%)} &
        \multicolumn{2}{c|}{Inf. Time (s)} &
        {Memory} \\

        {} & \multicolumn{1}{c}{\it $\mu$} & {\it $\sigma$} & \multicolumn{1}{c}{\it $\mu$} & {\it $\sigma$} & {(GB)} \\  
        \cline{1-6} 

        2 & 1.07 & 0.17 & 426 & 16 & 228 \\
        4 & 2.76 & 1.17 & 413 & 13 & 229 \\
        8 & 42.21 & 13.71 & 420 & 24 & 230 \\
        16 & 69.28 & 2.37 & 443 & 25 & 228 \\
        32 & 72.85 & 0.92 & 466 & 27 & 229 \\
        64 & 74.31 & 0.77 & 520 & 31 & 229 \\
        128 & 74.61 & 0.73 & 602 & 36 & 229 \\
        256 & 74.68 & 0.71 & 724 & 39 & 230 \\
        512 & 74.69 & 0.73 & 874 & 49 & 235 \\
        1024 & 74.71 & 0.74 & 1066 & 66 & 244 \\
        \cline{1-6}
        \multicolumn{6}{c}{} \\

        % Slice Clustering (M = 2)
        \cline{1-6}
        
        \multicolumn{6}{c}{Slice Clustering (M = 2)} \\
        \cline{1-6}

        \multirow{2}*{k} &
        \multicolumn{2}{c|}{Accuracy (\%)} &
        \multicolumn{2}{c|}{Inf. Time (s)} &
        {Memory} \\

        {} & \multicolumn{1}{c}{\it $\mu$} & {\it $\sigma$} & \multicolumn{1}{c}{\it $\mu$} & {\it $\sigma$} & {(GB)} \\  
        \cline{1-6} 

        64 & 68.93 & 1.94 & 524 & 10 & 247 \\
        128 & 74.47 & 0.59 & 622 & 23 & 247 \\
        256 & 75.79 & 0.43 & 734 & 51 & 247 \\
        512 & 76.35 & 0.18 & 871 & 37 & 251 \\
        1024 & 76.67 & 0.12 & 1049 & 23 & 261 \\
        2048 & 76.82 & 0.14 & 1202 & 72 & 281 \\
        4096 & 76.90 & 0.08 & 1308 & 18 & 342 \\
        8192 & 76.97 & 0.07 & 1440 & 13 & 500 \\
        \cline{1-6}
        \multicolumn{6}{c}{} \\

        % Slice Clustering (M = 4)
        \cline{1-6}
        
        \multicolumn{6}{c}{Slice Clustering (M = 4)} \\
        \cline{1-6}

        \multirow{2}*{k} &
        \multicolumn{2}{c|}{Accuracy (\%)} &
        \multicolumn{2}{c|}{Inf. Time (s)} &
        {Memory} \\

        {} & \multicolumn{1}{c}{\it $\mu$} & {\it $\sigma$} & \multicolumn{1}{c}{\it $\mu$} & {\it $\sigma$} & {(GB)} \\  
        \cline{1-6} 

        64 & 1.73 & 0.41 & 538 & 26 & 281 \\
        128 & 5.31 & 0.94 & 615 & 27 & 281 \\
        256 & 29.69 & 2.16 & 750 & 47 & 282 \\
        512 & 57.88 & 2.20 & 885 & 16 & 286 \\
        1024 & 68.98 & 1.13 & 1057 & 46 & 297 \\
        2048 & 73.96 & 0.38 & 1195 & 64 & 316 \\
        4096 & 75.93 & 0.22 & 1263 & 24 & 378 \\
        8192 & 76.66 & 0.12 & 1355 & 17 & 535 \\
        \cline{1-6}
    \end{tabular}
    \label{tab:cifar100_resnet18}
\end{table}
\begin{table}[!t]
    \centering
    \small
    \caption
    {
        Accuracy, inference time, and peak memory for CIFAR-100 on ResNet-20.
    }
    \begin{tabular}
    {
        C{1.0cm} |
        C{1.0cm} |
        C{1.0cm} |
        C{1.0cm} |
        C{1.0cm} |
        C{1.0cm}
    }

        % Standard
        \cline{1-6}
        
        \multicolumn{6}{c}{Standard} \\
        \cline{1-6}

        \multirow{2}*{M} &
        \multicolumn{2}{c|}{Accuracy (\%)} &
        \multicolumn{2}{c|}{Inf. Time (s)} &
        {Memory} \\

        {} & \multicolumn{1}{c}{\it $\mu$} & {\it $\sigma$} & \multicolumn{1}{c}{\it $\mu$} & {\it $\sigma$} & {(GB)} \\  
        \cline{1-6} 

        1 & 73.89 & 0.37 & 1902 & 26 & 305 \\
        2 & 76.59 & 0.25 & 1874 & 27 & 320 \\
        4 & 77.97 & 0.13 & 1904 & 30 & 351 \\
        \cline{1-6}
        \multicolumn{6}{c}{} \\

        % Full Clustering (M = 1)
        \cline{1-6}
        
        \multicolumn{6}{c}{Full Clustering (M = 1)} \\
        \cline{1-6}

        \multirow{2}*{k} &
        \multicolumn{2}{c|}{Accuracy (\%)} &
        \multicolumn{2}{c|}{Inf. Time (s)} &
        {Memory} \\

        {} & \multicolumn{1}{c}{\it $\mu$} & {\it $\sigma$} & \multicolumn{1}{c}{\it $\mu$} & {\it $\sigma$} & {(GB)} \\  
        \cline{1-6} 

        2 & 1.03 & 0.07 & 246 & 6 & 192 \\
        4 & 1.08 & 0.23 & 246 & 9 & 192 \\
        8 & 1.17 & 0.24 & 252 & 7 & 192 \\
        16 & 2.05 & 1.58 & 258 & 8 & 192 \\
        32 & 32.59 & 14.07 & 272 & 16 & 192 \\
        64 & 68.18 & 3.24 & 290 & 11 & 192 \\
        128 & 71.56 & 1.19 & 334 & 15 & 192 \\
        256 & 73.39 & 0.49 & 392 & 9 & 192 \\
        512 & 73.75 & 0.35 & 473 & 20 & 193 \\
        1024 & 73.85 & 0.35 & 563 & 28 & 201 \\
        \cline{1-6}
        \multicolumn{6}{c}{} \\

        % Slice Clustering (M = 1)
        \cline{1-6}
        
        \multicolumn{6}{c}{Slice Clustering (M = 1)} \\
        \cline{1-6}

        \multirow{2}*{k} &
        \multicolumn{2}{c|}{Accuracy (\%)} &
        \multicolumn{2}{c|}{Inf. Time (s)} &
        {Memory} \\

        {} & \multicolumn{1}{c}{\it $\mu$} & {\it $\sigma$} & \multicolumn{1}{c}{\it $\mu$} & {\it $\sigma$} & {(GB)} \\  
        \cline{1-6} 

        2 & 1.12 & 0.22 & 254 & 10 & 192 \\
        4 & 2.11 & 0.62 & 255 & 13 & 192 \\
        8 & 49.94 & 5.99 & 264 & 21 & 192 \\
        16 & 69.60 & 1.71 & 262 & 7 & 192 \\
        32 & 72.98 & 0.45 & 292 & 12 & 192 \\
        64 & 73.68 & 0.38 & 323 & 13 & 192 \\
        128 & 73.85 & 0.38 & 377 & 18 & 192 \\
        256 & 73.89 & 0.34 & 452 & 17 & 193 \\
        512 & 73.89 & 0.36 & 538 & 10 & 198 \\
        1024 & 73.88 & 0.36 & 626 & 21 & 209 \\
        \cline{1-6}
        \multicolumn{6}{c}{} \\

        % Slice Clustering (M = 2)
        \cline{1-6}
        
        \multicolumn{6}{c}{Slice Clustering (M = 2)} \\
        \cline{1-6}

        \multirow{2}*{k} &
        \multicolumn{2}{c|}{Accuracy (\%)} &
        \multicolumn{2}{c|}{Inf. Time (s)} &
        {Memory} \\

        {} & \multicolumn{1}{c}{\it $\mu$} & {\it $\sigma$} & \multicolumn{1}{c}{\it $\mu$} & {\it $\sigma$} & {(GB)} \\  
        \cline{1-6} 

        64 & 68.98 & 0.97 & 321 & 14 & 208 \\
        128 & 74.31 & 0.44 & 380 & 16 & 207 \\
        256 & 76.08 & 0.25 & 464 & 7 & 208 \\
        512 & 76.65 & 0.19 & 541 & 22 & 213 \\
        1024 & 76.99 & 0.11 & 614 & 22 & 226 \\
        2048 & 77.16 & 0.13 & 690 & 23 & 259 \\
        4096 & 77.12 & 0.12 & 746 & 17 & 354 \\
        8192 & 77.23 & 0.09 & 897 & 26 & 556 \\
        \cline{1-6}
        \multicolumn{6}{c}{} \\

        % Slice Clustering (M = 4)
        \cline{1-6}
        
        \multicolumn{6}{c}{Slice Clustering (M = 4)} \\
        \cline{1-6}

        \multirow{2}*{k} &
        \multicolumn{2}{c|}{Accuracy (\%)} &
        \multicolumn{2}{c|}{Inf. Time (s)} &
        {Memory} \\

        {} & \multicolumn{1}{c}{\it $\mu$} & {\it $\sigma$} & \multicolumn{1}{c}{\it $\mu$} & {\it $\sigma$} & {(GB)} \\  
        \cline{1-6} 

        64 & 1.72 & 0.48 & 341 & 13 & 240 \\
        128 & 3.86 & 0.96 & 407 & 10 & 239 \\
        256 & 19.97 & 2.42 & 463 & 19 & 239 \\
        512 & 52.26 & 2.81 & 546 & 8 & 244 \\
        1024 & 68.37 & 0.63 & 631 & 14 & 256 \\
        2048 & 73.73 & 0.53 & 673 & 11 & 290 \\
        4096 & 76.44 & 0.31 & 765 & 9 & 386 \\
        8192 & 77.62 & 0.17 & 855 & 16 & 588 \\
        \cline{1-6}
    \end{tabular}
    \label{tab:cifar100_resnet20}
\end{table}
\begin{table}[!t]
    \centering
    \small
    \caption
    {
        Accuracy, inference time, and peak memory for CIFAR-100 on ResNet-32.
    }
    \begin{tabular}
    {
        C{1.0cm} |
        C{1.0cm} |
        C{1.0cm} |
        C{1.0cm} |
        C{1.0cm} |
        C{1.0cm}
    }

        % Standard
        \cline{1-6}
        
        \multicolumn{6}{c}{Standard} \\
        \cline{1-6}

        \multirow{2}*{M} &
        \multicolumn{2}{c|}{Accuracy (\%)} &
        \multicolumn{2}{c|}{Inf. Time (s)} &
        {Memory} \\

        {} & \multicolumn{1}{c}{\it $\mu$} & {\it $\sigma$} & \multicolumn{1}{c}{\it $\mu$} & {\it $\sigma$} & {(GB)} \\  
        \cline{1-6} 

        1 & 73.93 & 0.53 & 8865 & 132 & 1401 \\
        2 & 76.79 & 0.34 & 8902 & 152 & 1418 \\
        4 & 78.34 & 0.23 & 8907 & 141 & 1465 \\
        \cline{1-6}
        \multicolumn{6}{c}{} \\

        % Full Clustering (M = 1)
        \cline{1-6}
        
        \multicolumn{6}{c}{Full Clustering (M = 1)} \\
        \cline{1-6}

        \multirow{2}*{k} &
        \multicolumn{2}{c|}{Accuracy (\%)} &
        \multicolumn{2}{c|}{Inf. Time (s)} &
        {Memory} \\

        {} & \multicolumn{1}{c}{\it $\mu$} & {\it $\sigma$} & \multicolumn{1}{c}{\it $\mu$} & {\it $\sigma$} & {(GB)} \\  
        \cline{1-6} 

        2 & 0.98 & 0.07 & 1544 & 5 & 836 \\
        4 & 0.98 & 0.07 & 1577 & 20 & 836 \\
        8 & 0.97 & 0.07 & 1540 & 24 & 836 \\
        16 & 1.03 & 0.21 & 1547 & 27 & 837 \\
        32 & 3.91 & 2.63 & 1565 & 13 & 836 \\
        64 & 52.92 & 7.49 & 1686 & 44 & 837 \\
        128 & 70.26 & 1.24 & 1819 & 5 & 836 \\
        256 & 73.18 & 0.38 & 2013 & 33 & 837 \\
        512 & 73.75 & 0.37 & 2381 & 70 & 837 \\
        1024 & 73.89 & 0.33 & 2822 & 53 & 871 \\
        \cline{1-6}
        \multicolumn{6}{c}{} \\

        % Slice Clustering (M = 1)
        \cline{1-6}
        
        \multicolumn{6}{c}{Slice Clustering (M = 1)} \\
        \cline{1-6}

        \multirow{2}*{k} &
        \multicolumn{2}{c|}{Accuracy (\%)} &
        \multicolumn{2}{c|}{Inf. Time (s)} &
        {Memory} \\

        {} & \multicolumn{1}{c}{\it $\mu$} & {\it $\sigma$} & \multicolumn{1}{c}{\it $\mu$} & {\it $\sigma$} & {(GB)} \\  
        \cline{1-6} 

        2 & 0.99 & 0.08 & 1559 & 11 & 836 \\
        4 & 1.00 & 0.09 & 1526 & 27 & 836 \\
        8 & 9.67 & 5.64 & 1586 & 11 & 837 \\
        16 & 60.50 & 4.61 & 1576 & 8 & 837 \\
        32 & 71.75 & 0.80 & 1679 & 19 & 836 \\
        64 & 73.49 & 0.36 & 1803 & 9 & 836 \\
        128 & 73.83 & 0.33 & 2027 & 30 & 836 \\
        256 & 73.91 & 0.33 & 2379 & 83 & 838 \\
        512 & 73.93 & 0.33 & 2815 & 60 & 860 \\
        1024 & 73.94 & 0.53 & 3313 & 57 & 925 \\
        \cline{1-6}
        \multicolumn{6}{c}{} \\

        % Slice Clustering (M = 2)
        \cline{1-6}
        
        \multicolumn{6}{c}{Slice Clustering (M = 2)} \\
        \cline{1-6}

        \multirow{2}*{k} &
        \multicolumn{2}{c|}{Accuracy (\%)} &
        \multicolumn{2}{c|}{Inf. Time (s)} &
        {Memory} \\

        {} & \multicolumn{1}{c}{\it $\mu$} & {\it $\sigma$} & \multicolumn{1}{c}{\it $\mu$} & {\it $\sigma$} & {(GB)} \\  
        \cline{1-6} 

        64 & 52.29 & 6.60 & 1791 & 17 & 861 \\
        128 & 70.20 & 1.26 & 2014 & 44 & 861 \\
        256 & 74.51 & 0.42 & 2340 & 52 & 862 \\
        512 & 76.30 & 0.25 & 2809 & 38 & 884 \\
        1024 & 76.77 & 0.24 & 3238 & 64 & 951 \\
        2048 & 77.24 & 0.08 & 3797 & 64 & 1125 \\
        4096 & 77.33 & 0.07 & 4388 & 28 & 1615 \\
        % 8192 & 77.31 & 0.08 & \multicolumn{3}{c}{out of memory} \\
        \cline{1-6}
        \multicolumn{6}{c}{} \\

        % Slice Clustering (M = 4)
        \cline{1-6}
        
        \multicolumn{6}{c}{Slice Clustering (M = 4)} \\
        \cline{1-6}

        \multirow{2}*{k} &
        \multicolumn{2}{c|}{Accuracy (\%)} &
        \multicolumn{2}{c|}{Inf. Time (s)} &
        {Memory} \\

        {} & \multicolumn{1}{c}{\it $\mu$} & {\it $\sigma$} & \multicolumn{1}{c}{\it $\mu$} & {\it $\sigma$} & {(GB)} \\  
        \cline{1-6} 

        64 & 0.92 & 0.00 & 1817 & 22 & 912 \\
        128 & 0.92 & 0.01 & 2031 & 26 & 911 \\
        256 & 1.79 & 0.74 & 2367 & 50 & 912 \\
        512 & 16.68 & 6.66 & 2823 & 47 & 935 \\
        1024 & 47.96 & 5.14 & 3495 & 15 & 1001 \\
        2048 & 65.54 & 1.33 & 3740 & 59 & 1175 \\
        4096 & 73.88 & 0.63 & 4199 & 60 & 1664 \\
        % 8192 & 76.89 & 0.23 & \multicolumn{3}{c}{out of memory} \\
        \cline{1-6}
    \end{tabular}
    \label{tab:cifar100_resnet32}
\end{table}
\begin{table}[!t]
    \centering
    \small
    \caption
    {
        Accuracy, inference time, and peak memory for Tiny-ImageNet on ResNet-32.
    }
    \begin{tabular}
    {
        C{1.0cm} |
        C{1.0cm} |
        C{1.0cm} |
        C{1.0cm} |
        C{1.0cm} |
        C{1.0cm}
    }

        % Standard
        \cline{1-6}
        
        \multicolumn{6}{c}{Standard} \\
        \cline{1-6}

        \multirow{2}*{M} &
        \multicolumn{2}{c|}{Accuracy (\%)} &
        \multicolumn{2}{c|}{Inf. Time (s)} &
        {Memory} \\

        {} & \multicolumn{1}{c}{\it $\mu$} & {\it $\sigma$} & \multicolumn{1}{c}{\it $\mu$} & {\it $\sigma$} & {(GB)} \\  
        \cline{1-6} 

        1 & 56.79 & 0.43 & 9086 & 110 & 1532 \\
        2 & 61.61 & 0.52 & 9001 & 128 & 1615 \\
        4 & 64.02 & 0.38 & 9240 & 131 & 1779 \\
        \cline{1-6}
        \multicolumn{6}{c}{} \\

        % Full Clustering (M = 1)
        \cline{1-6}
        
        \multicolumn{6}{c}{Full Clustering (M = 1)} \\
        \cline{1-6}

        \multirow{2}*{k} &
        \multicolumn{2}{c|}{Accuracy (\%)} &
        \multicolumn{2}{c|}{Inf. Time (s)} &
        {Memory} \\

        {} & \multicolumn{1}{c}{\it $\mu$} & {\it $\sigma$} & \multicolumn{1}{c}{\it $\mu$} & {\it $\sigma$} & {(GB)} \\  
        \cline{1-6} 

        2 & 0.46 & 0.04 & 1700 & 13 & 1000 \\
        4 & 0.46 & 0.04 & 1626 & 11 & 996 \\
        8 & 0.46 & 0.04 & 1682 & 18 & 996 \\
        16 & 0.47 & 0.06 & 1655 & 36 & 998 \\
        32 & 0.54 & 0.13 & 1709 & 14 & 998 \\
        64 & 8.30 & 6.52 & 1779 & 7 & 997 \\
        128 & 38.21 & 8.17 & 1889 & 10 & 999 \\
        256 & 53.49 & 1.46 & 2113 & 37 & 997 \\
        512 & 56.11 & 0.68 & 2508 & 10 & 997 \\
        1024 & 56.55 & 0.47 & 2846 & 11 & 1026 \\
        \cline{1-6}
        \multicolumn{6}{c}{} \\

        % Slice Clustering (M = 1)
        \cline{1-6}
        
        \multicolumn{6}{c}{Slice Clustering (M = 1)} \\
        \cline{1-6}

        \multirow{2}*{k} &
        \multicolumn{2}{c|}{Accuracy (\%)} &
        \multicolumn{2}{c|}{Inf. Time (s)} &
        {Memory} \\

        {} & \multicolumn{1}{c}{\it $\mu$} & {\it $\sigma$} & \multicolumn{1}{c}{\it $\mu$} & {\it $\sigma$} & {(GB)} \\  
        \cline{1-6} 

        2 & 0.46 & 0.04 & 1638 & 11 & 997 \\
        4 & 0.47 & 0.04 & 1689 & 27 & 999 \\
        8 & 0.62 & 0.22 & 1711 & 18 & 997 \\
        16 & 15.64 & 5.91 & 1688 & 18 & 996 \\
        32 & 48.27 & 2.89 & 1768 & 9 & 997 \\
        64 & 55.42 & 0.82 & 1869 & 10 & 999 \\
        128 & 56.47 & 0.58 & 2088 & 30 & 997 \\
        256 & 56.71 & 0.43 & 2428 & 37 & 998 \\
        512 & 56.77 & 0.44 & 2871 & 25 & 1022 \\
        1024 & 56.77 & 0.44 & 3388 & 30 & 1083 \\
        \cline{1-6}
        \multicolumn{6}{c}{} \\

        % Slice Clustering (M = 2)
        \cline{1-6}
        
        \multicolumn{6}{c}{Slice Clustering (M = 2)} \\
        \cline{1-6}

        \multirow{2}*{k} &
        \multicolumn{2}{c|}{Accuracy (\%)} &
        \multicolumn{2}{c|}{Inf. Time (s)} &
        {Memory} \\

        {} & \multicolumn{1}{c}{\it $\mu$} & {\it $\sigma$} & \multicolumn{1}{c}{\it $\mu$} & {\it $\sigma$} & {(GB)} \\  
        \cline{1-6} 

        64 & 10.57 & 3.31 & 1894 & 12 & 1083 \\
        128 & 36.63 & 3.34 & 2153 & 58 & 1084 \\
        256 & 52.82 & 2.00 & 2434 & 30 & 1087 \\
        512 & 58.73 & 0.45 & 2836 & 42 & 1110 \\
        1024 & 60.78 & 0.30 & 3285 & 52 & 1171 \\
        2048 & 61.56 & 0.14 & 3706 & 96 & 1349 \\
        4096 & 62.13 & 0.23 & 4054 & 118 & 1833 \\
        % 8192 & 62.13 & 0.15 & \multicolumn{3}{c}{out of memory} \\
        \cline{1-6}
        \multicolumn{6}{c}{} \\

        % Slice Clustering (M = 4)
        \cline{1-6}
        
        \multicolumn{6}{c}{Slice Clustering (M = 4)} \\
        \cline{1-6}

        \multirow{2}*{k} &
        \multicolumn{2}{c|}{Accuracy (\%)} &
        \multicolumn{2}{c|}{Inf. Time (s)} &
        {Memory} \\

        {} & \multicolumn{1}{c}{\it $\mu$} & {\it $\sigma$} & \multicolumn{1}{c}{\it $\mu$} & {\it $\sigma$} & {(GB)} \\  
        \cline{1-6} 

        64 & 0.50 & 0.00 & 1924 & 24 & 1256 \\
        128 & 0.52 & 0.08 & 2124 & 25 & 1256 \\
        256 & 0.70 & 0.29 & 2415 & 45 & 1257 \\
        512 & 2.37 & 1.29 & 2858 & 22 & 1278 \\
        1024 & 16.71 & 4.91 & 3581 & 43 & 1340 \\
        2048 & 39.94 & 2.75 & 4142 & 116 & 1515 \\
        % 4096 & 53.94 & 1.82 & \multicolumn{3}{c}{out of memory} \\
        % 8192 & 60.10 & 0.53 & \multicolumn{3}{c}{out of memory} \\
        \cline{1-6}
    \end{tabular}
    \label{tab:tiny_resnet32}
\end{table}
\clearpage

% \section{Ethical Considerations}
% Our work focuses on improving the computational efficiency and accuracy of privacy-preserving machine learning using leveled FHE. 
% All experiments use publicly available benchmark datasets such as CIFAR-10, CIFAR-100, and Tiny-ImageNet. 
% The research does not involve human subjects, personal data, or interactions with live systems. 
% The primary stakeholders are end users who need strong protection for sensitive data, organizations that provide encrypted inference services, and the research community.
% Our techniques reduce reliance on costly bootstrapping and allow deeper neural networks to run efficiently on encrypted inputs. 
% This provides a clear societal benefit by making practical privacy-preserving machine learning more accessible. 
% The method operates strictly within the security guarantees of the underlying cryptographic scheme and does not introduce new vulnerabilities or mechanisms that could enable data misuse or surveillance. 
% It focuses solely on improving the performance of established secure computation methods. 
% All implementation details and model configurations are included to support reproducibility. 
% Because the research does not involve human subjects, sensitive information, or activities that require institutional oversight, no IRB review is required.

% \section{Generative AI Usage}
% We used ChatGPT only to improve the wording and readability of the paper. They were not involved in generating ideas, analysis, or writing beyond basic language edits.

\end{document}